\documentclass[11pt,a4paper]{article}
\pdfoutput=1
\usepackage[top=25mm,bottom=30mm,left=29mm,right=29mm]{geometry}

\usepackage{tocloft}

\usepackage[nottoc,notlot,notlof]{tocbibind}

\usepackage[pdfencoding=auto]{hyperref}

\hypersetup{
    colorlinks,
    linkcolor={red!50!black},
    citecolor={blue!50!black},
    urlcolor={blue!80!black}
}

\usepackage[utf8]{inputenc}

\usepackage{amssymb,amsmath,amsthm}
\usepackage{enumerate}

\usepackage{graphicx}
\usepackage{cite}

\usepackage{mathtools}
\usepackage{xcolor}
\usepackage{tikz}
\usepackage{tikz-cd}

\usepackage{todonotes}

\theoremstyle{definition}
\newtheorem{theorem}{Theorem}
\newtheorem{proposition}[theorem]{Proposition}
\newtheorem*{theorem*}{Theorem}
\newtheorem{remark}[theorem]{Remark}
\newtheorem*{example*}{Example}
\newtheorem*{examples*}{Examples}
\newtheorem{example}[theorem]{Example}
\newtheorem{examples}[theorem]{Examples}

\numberwithin{equation}{section}
\numberwithin{theorem}{section}
\numberwithin{figure}{section}

\DeclareMathOperator{\Hom}{Hom}
\DeclareMathOperator{\id}{id}
\newcommand{\one}{\mathbf{1}}

\usepackage{dsfont}
\usepackage{booktabs}
\newcommand{\zz}{\mathcal{Z}}
\newcommand{\Bord}{\textrm{Bord}}
\newcommand{\Borddef}{\Bord_{n,n-1}^{\textrm{def}}(\mathds{D})}

\newcommand{\Vect}{\operatorname{Vect}}
\def\lra{\longrightarrow}
\newcommand{\btimes}{\mathbin{\square}} 
\usetikzlibrary{decorations.markings}
\usetikzlibrary{fadings,decorations.pathreplacing}
\usetikzlibrary{matrix,arrows}
\usetikzlibrary{patterns}
\usetikzlibrary{arrows,calc,decorations.pathreplacing,decorations.markings,shapes.geometric,shadows}
\tikzset{%
	string/.style={draw=#1, postaction={decorate}, decoration={markings,mark=at position .51 with {\arrow[draw=#1]{>}}}},
	costring/.style={draw=#1, postaction={decorate}, decoration={markings,mark=at position .51 with {\arrow[draw=#1]{<}}}},
	directed/.style={string=blue!50!black},
	redirected/.style={costring= blue!50!black}}

\makeatletter
\g@addto@macro\bfseries{\boldmath}
\makeatother

\let\emph\undefined
\newcommand{\emph}[1]{\textsl{#1}}
\let\textit\undefined
\newcommand{\textit}[1]{\textsl{#1}}

\newcommand{\topD}[2]{\mathcal{D}_{#1} #2}

\begin{document} 
\thispagestyle{empty}
\begin{center}{\Large \textbf{
Topological defects}}
\end{center}

\begin{center}
Nils Carqueville\textsuperscript{$\alpha$} \quad 
Michele Del Zotto\textsuperscript{$\beta,\beta'$} \quad
Ingo Runkel\textsuperscript{$\gamma$}\footnote{
Emails: 
{\sf nils.carqueville@univie.ac.at}\,,~
{\sf michele.delzotto@math.uu.se}\,,~
{\sf ingo.runkel@uni-hamburg.de} 
}
\end{center}

\begin{center}
	{${}^{\alpha}$}Universit\"at Wien, Fakult\"at f\"ur Physik, \\
	Boltzmanngasse 5, 1090 Wien,  
	Austria
\\[1em]
	{${}^{\beta}$} Institute for Mathematics, Uppsala University,\\
	Box 480, SE-75106 Uppsala, Sweden \\
	{${}^{\beta'}$} Department of Physics and Astronomy, Uppsala University,\\
	Box 516, SE-75120 Uppsala, Sweden 
\\[1em]
	{${}^{\gamma}$}Fachbereich Mathematik, Universität Hamburg,\\
Bundesstraße 55, 20146 Hamburg, Germany
\end{center}


\section*{Abstract}
This is a survey article for the Encyclopedia of Mathematical Physics, 2nd Edition. 
Topological defects are described in the context of the 2-dimensional Ising model on the lattice, in 2-dimensional quantum field theory, in topological quantum field theory in arbitrary dimension, and in higher-dimensional quantum field theory with a focus on 4-dimensional quantum electrodynamics.
\setcounter{tocdepth}{2}


\tableofcontents

\newpage

\section{Introduction}

A \textsl{defect} in a $d$-dimensional Euclidean quantum field theory is an inhomogeneity localised on a submanifold of positive codimension.
For codimension~$1$ this includes interfaces between different theories. 
A defect is \textsl{topological} if the precise location of the defect submanifold does not matter in the sense that it can be deformed without affecting correlation functions, as long as it does not cross other defects or field insertions.
Defects, topological or not, also exist in statistical lattice models and quantum spin systems, though the localisation to submanifolds should be understood to hold up to some fixed number of lattice spacings.

\medskip

Conceptually, topological defects generalise the notion of a group symmetry.
An element~$g$ of the symmetry group corresponds to a topological defect of codimension $1$ across which the value of field operators jumps by the action of~$g$. 
In this generalisation, the mathematical language of group theory is upgraded to that of (higher) category theory. 
Namely, in an $n$-dimensional field theory or statistical model, topological defects are expected to form an $n$-category, with defects localised on $k$-manifolds providing the $(n-k)$-morphisms. 
Thinking of a group as a category with one object where all morphisms are invertible illustrates that the $n$-category of topological defects provides much finer information, since
\begin{itemize}
	\item
it includes $p$-morphism with $p>1$, corresponding to lower-dimensional defects joining those of higher dimensions into defect networks,
\item
it allows for non-invertible morphisms, corresponding to defects of a given dimension which cannot be fused with another defect to arrive at the trivial defect. 
\end{itemize}

Topological defects appear in a variety of settings and using different languages. 
It is remarkable that the same mathematical structure -- higher categories -- appears in all of these descriptions. This survey presents several of these settings in their language of choice -- lattice models, functorial field theory, and action functionals -- and sketches how topological defects and higher categories arise in each case. 
While each approach looks quite distinct at the outset, they are joined by the underlying mathematical structure of higher categories: 

\medskip

\noindent
\textbf{Section~\ref{sec:Ising-lattice}.} The starting point of this survey is the most famous of all statistical models, the 2-dimensional Ising model on a square lattice. Understanding topological defects in this language requires only finite sums, and no knowledge of quantum field theory or category theory. Yet both of the above generalisations appear naturally: a 2-category of topological defects, some of which are non-invertible.  

\medskip

\noindent
\textbf{Section~\ref{sec:2dQFT}.} 
Staying in two dimensions, we take an axiomatic approach to topological defects in QFT and explain how they give rise to a pivotal 2-category. 
A key application of topological defects -- invertible or non-invertible -- is their gauging. 
The datum of a gaugeable symmetry in two dimensions is a symmetric $\Delta$-separable Frobenius algebra. 
Examples under good mathematical control are provided by 
2-dimensional conformal field theories.

\medskip

\noindent
\textbf{Section~\ref{sec:TFT}.} Going beyond dimension two while staying on solid mathematical ground is possible in topological QFTs,
defined as functors out of stratified bordisms categories. 
We discuss the resulting $n$-categories of topological defects for $n=2,3$. 
The seemingly simple notion of an equivalence of objects leads to the important concept of condensations by dropping part of the data and conditions in an equivalence. 
Condensations are related to gaugeable symmetries, which we discuss for both framed and oriented theories. 

\medskip

\noindent
\textbf{Section~\ref{sec:QFT+}.} Much recent progress has been made in the description of topological and non-topological defects in quantum field theory. Higher categories and non-invertible topological defects turn out to be ubiquitous, and we use Maxwell theory and quantum electrodynamics in four dimensions to illustrate some of their features.

\bigskip

\noindent
\textbf{Acknowledgements.}
We thank 
    Lakshya Bhardwaj, 
    Ilka Brunner, 
    Vincentas Mulevi\v{c}ius, 
    Lukas Müller, 
    Kantaro Ohmori,
    and 
    Shu-Heng Shao 
for helpful discussions and comments on an earlier draft of this manuscript. 
We furthermore thank our section editor Kasia Rejzner for her support.
N.\,C.\ acknowledges support from the DFG Heisenberg Programme. 
The work of M.\,D.\,Z.\ has received funding from
the European Research Council (ERC) under the European Union’s Horizon 2020 research and innovation program (grant agreement No.\ 851931). M.\,D.\,Z.\ also acknowledges support from the Simons Foundation (grant 888984, Simons Collaboration on Global Categorical Symmetries).
I.\,R.\ is partially supported by the Deutsche Forschungsgemeinschaft via the Cluster of Excellence EXC 2121 ``Quantum Universe'' - 390833306.

\section{The 2-dimensional Ising model on a lattice}
\label{sec:Ising-lattice}

Despite its simplicity, the 2-dimensional statistical Ising model already possesses topological defects which exhibit several interesting features also observed in continuum models. Kramers--Wannier duality on the lattice is reviewed in \cite{Savit:1979ny}, and topological defects in the lattice Ising model are discussed in detail in \cite{Aasen:2016dop}, which is the main reference for this section.
For a similar discussion of other 2-dimensional models see \cite{Freed:2018cec,Aasen:2020jwb}, and for models with topological surface defects in three dimensions see \cite{Inamura:2023qzl}.

\subsection{The Ising model and its $\mathds{Z}_2$-symmetry}\label{sec:Ising-and-Z2}

We consider the 2-dimensional statistical Ising model on a square lattice $\Lambda$. The degrees of freedom are lattice spins $s_i \in \{ \pm1\}$ located on lattice sites $i \in \Lambda$. To an edge $\langle ij \rangle$ between sites $i$ and $j$ one assigns the edge weight $e^{\beta s_i s_j}$, where $\beta$ is the inverse temperature. The partition function is a sum over spin configurations $s\colon \Lambda \to \{ \pm 1\}$ of the product of the edge weights.
For the description of dualities below it is useful to include a constant weight $V \in \mathds{C}^\times$ per site (this will be a free parameter in the value of duality defect loops). Altogether, the partition function is 
\begin{align}\label{eq:Ising-PF}
  Z = \sum_s \prod_{ \langle ij \rangle } \textrm{e}^{\beta s_i s_j} \prod_i V~.
\end{align}
We assume that $\Lambda$ has a finite number of sites, e.g.\ a finite square with free or fixed boundary conditions, or a torus. The expectation value of lattice spins at sites $j_1,\dots,j_n$ is given by the ratio
\begin{align}
	\big\langle s_{j_1} \cdots s_{j_n} \big\rangle
	= \frac{1}{Z} \sum_s s_{j_1} \cdots s_{j_n} \prod_{ \langle ij \rangle } \textrm{e}^{\beta s_i s_j} 
 \prod_i V
 ~.
\end{align}
We refer to $s_{j_1}, \dots, s_{j_n}$ as operator insertions. 

The Ising model has a global symmetry consisting of changing the sign of all spins in a given configuration. This symmetry can be recovered from a topological defect line on the dual lattice $\widetilde\Lambda$ (which by definition is again a square lattice, but its vertices sit at the centre of the plaquettes of $\Lambda$). 
Namely, fix a path $\gamma$ on $\widetilde\Lambda$. For each edge $\langle ij \rangle$ of $\Lambda$ which intersects the path $\gamma$, the edge weight is changed to $\textrm{e}^{- \beta s_i s_j}$. The partition function does not change if one moves the path across one square of $\widetilde\Lambda$ as in example (a) below:
\begin{align}\label{eq:diag-S-move-1-square}
\raisebox{1.7em}{\text{(a)}}~~
\raisebox{-0.45\height}{\scalebox{.25}{\includegraphics{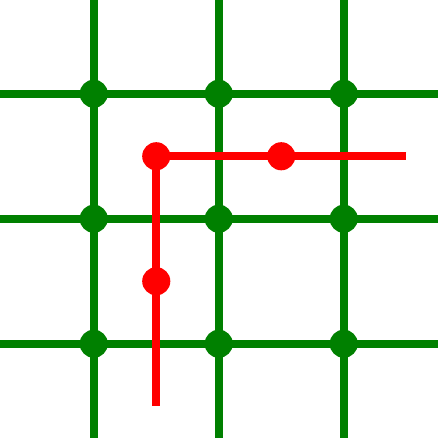}}}
~=~
\raisebox{-0.45\height}{\scalebox{.25}{\includegraphics{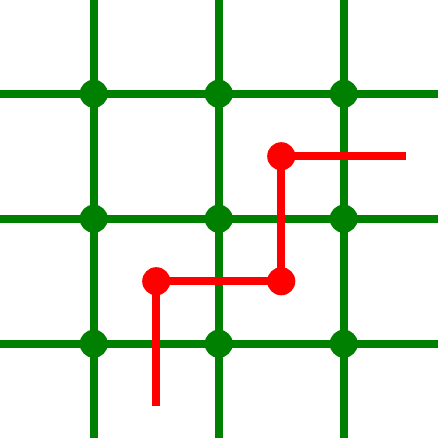}}}
\hspace{4em}
\raisebox{1.7em}{\text{(b)}}~~
\raisebox{-0.45\height}{\scalebox{.25}{\includegraphics{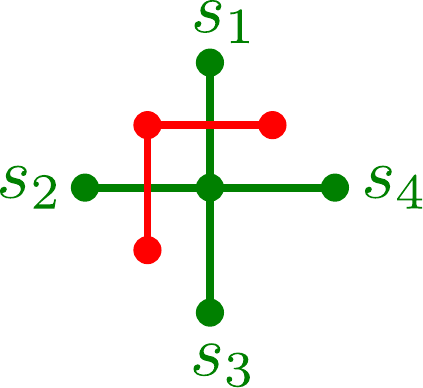}}}
~=~
\raisebox{-0.45\height}{\scalebox{.25}{\includegraphics{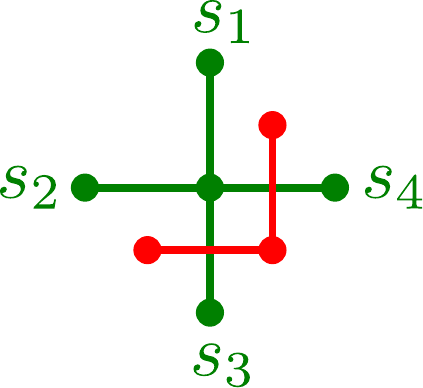}}}
\end{align}
The move in (a) follows from the local identity (b), where the spins at the labelled sites are kept fixed and the spin at the unlabelled site is summed over. If we write $t$ for the spin at the unlabelled site, (b) represents the identity
\begin{align}
V^5
\!\!\sum_{t \in \{\pm 1\}}\!\! \textrm{e}^{-\beta s_1 t} \textrm{e}^{-\beta s_2 t} \textrm{e}^{\beta s_3 t} \textrm{e}^{\beta s_4 t}
\,=\,
V^5
\!\!\sum_{t \in \{\pm 1\}}\!\! \textrm{e}^{\beta s_1 t} \textrm{e}^{\beta s_2 t} \textrm{e}^{-\beta s_3 t} \textrm{e}^{-\beta s_4 t} \ ,
\end{align}
which follows from renaming $t \leadsto -t$.
Analogously one shows
\begin{align}\label{eq:diag-S-loop+reconnect}
\raisebox{-0.45\height}{\scalebox{.25}{\includegraphics{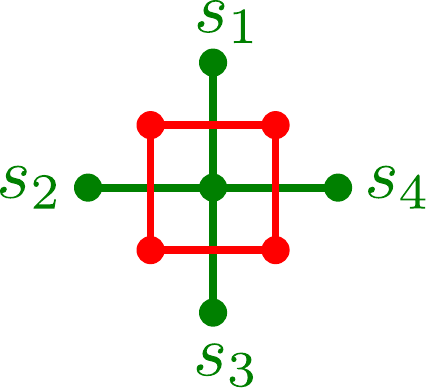}}}
~=~
\raisebox{-0.45\height}{\scalebox{.25}{\includegraphics{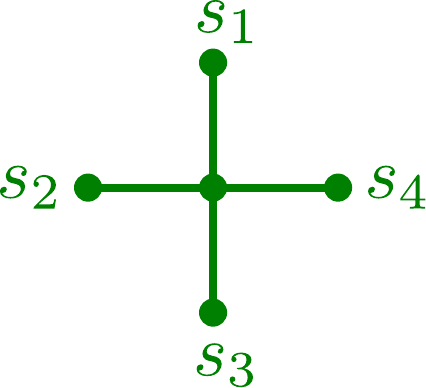}}}
\quad , \qquad
\raisebox{-0.45\height}{\scalebox{.25}{\includegraphics{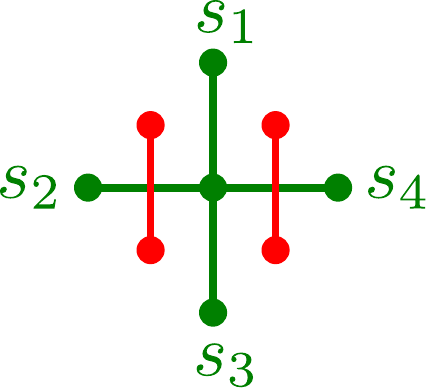}}}
~=~
\raisebox{-0.45\height}{\scalebox{.25}{\includegraphics{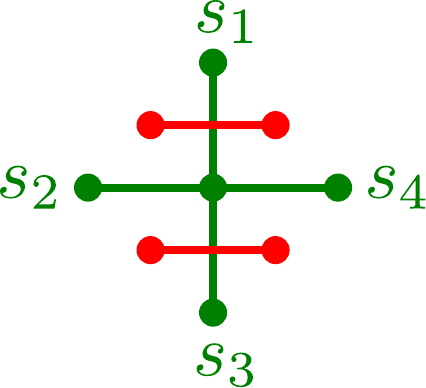}}}
\quad .
\end{align}
If we zoom out and do not display the details of the lattice, the above identities imply 
\begin{align}\label{eq:lattice-S-defect-zoom-out}
\raisebox{-0.45\height}{\scalebox{.25}{\includegraphics{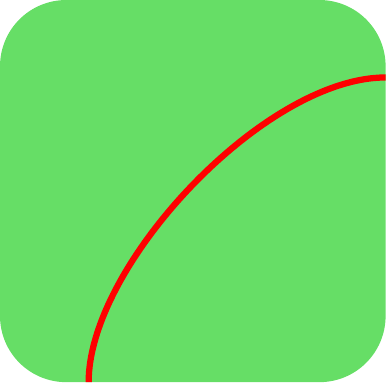}}}
\,=\,
\raisebox{-0.45\height}{\scalebox{.25}{\includegraphics{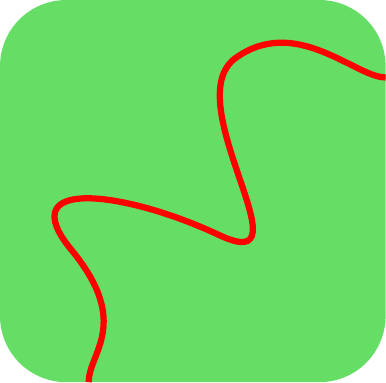}}}
~~ , ~~~
\raisebox{-0.45\height}{\scalebox{.25}{\includegraphics{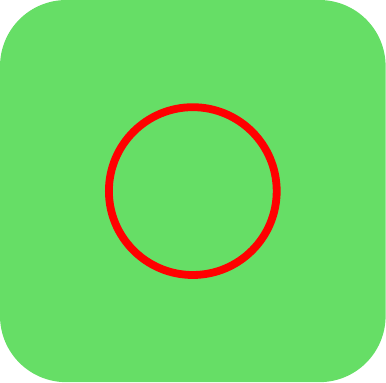}}}
\,=\,
\raisebox{-0.45\height}{\scalebox{.25}{\includegraphics{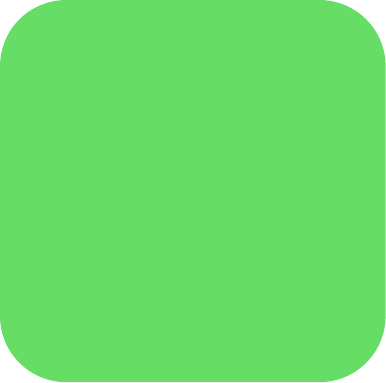}}}
~~ , ~~~
\raisebox{-0.45\height}{\scalebox{.25}{\includegraphics{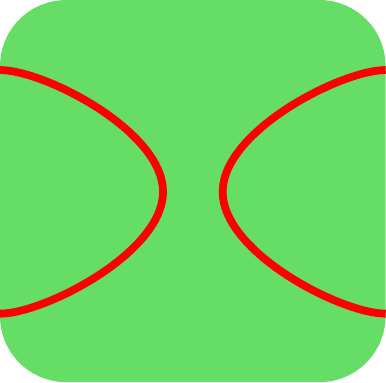}}}
\,=\,
\raisebox{-0.45\height}{\scalebox{.25}{\includegraphics{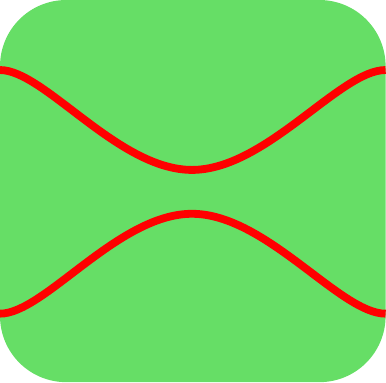}}}
~~.
\end{align}

The identities in \eqref{eq:diag-S-move-1-square} and \eqref{eq:diag-S-loop+reconnect} also hold for expectation values, provided that there are no operator insertions at the sites involved in the local configurations shown. Else for example in \eqref{eq:diag-S-move-1-square} one would also have to change the sign of the operator insertion at the site over which the path is moved. 
In the following we will always assume that there are no operator insertions at sites shown in the local diagrams.

The identity \eqref{eq:diag-S-move-1-square}\,(b) implies that expectation values do not change if the defect line $\gamma$ is deformed on $\widetilde\Lambda$ 
as in the first equality of \eqref{eq:lattice-S-defect-zoom-out}, provided
it stays away from operator insertions. 
In this sense it is a topological line defect, and we refer to it as the spin-flip defect $S$. The second and third identity in \eqref{eq:lattice-S-defect-zoom-out} imply that $S$ is \textit{(weakly) invertible}, in particular two parallel $S$ lines cancel.
These are examples of the \textit{bubble} and \textit{reconnection} condition, see Section~\ref{subsec:defect2categories}.

Apart from the spin flip defect, there is also the trivial defect, or invisible defect, $I$. 
Just like $S$ it is given by a path on $\widetilde\Lambda$, but it does not affect the weights in any way. 
We can write the above insight as a fusion rule for topological defect lines: $S \otimes S \cong I$.

\subsection{Gauging the $\mathds{Z}_2$-symmetry}

The $\mathds{Z}_2$-symmetry can be gauged, or, in other words, one can orbifold by it.
To this end, consider the superposition of line defects $A = I \oplus S$. 
On the lattice, the superposition is modelled by assigning a new degree of freedom $\alpha_e \in \{ \pm1\}$ to the edges $e$ of the path $\gamma$ and summing over the $\alpha_e$.
For $\alpha_e = 1$, the edge carries the defect $I$, and for $\alpha_e=-1$, it carries $S$. If an edge $\langle ij \rangle$ of $\Lambda$ is crossed by an edge $e$ of $\gamma$, the edge weight assigned to $\langle ij \rangle$ is $\textrm{e}^{\beta \alpha_e s_i s_j}$.
For two adjacent edges $e_1$, $e_2$ of the path $\gamma$ which meet at a vertex $a \in \widetilde\Lambda$, the variables have to agree: $\alpha_{e_1}=\alpha_{e_2}$. This condition means that on $\gamma$ either all $\alpha_e$ are $1$ or all are $-1$. 
The sum over the thus constrained $\alpha_e$ implements the superposition $A = I \oplus S$.
    
It will be convenient to orient $A$-defect lines. On top of that, we add additional structure, namely two types of three-fold junctions. For one type of junction, there are two incoming $A$-defects and one outgoing one, and vice versa for the other type, for example:
\begin{align}
\raisebox{-0.45\height}{\scalebox{.3}{\includegraphics{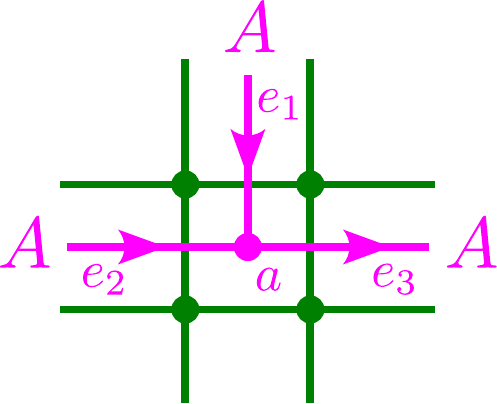}}}
\hspace{5em}
\raisebox{-0.45\height}{\scalebox{.3}{\includegraphics{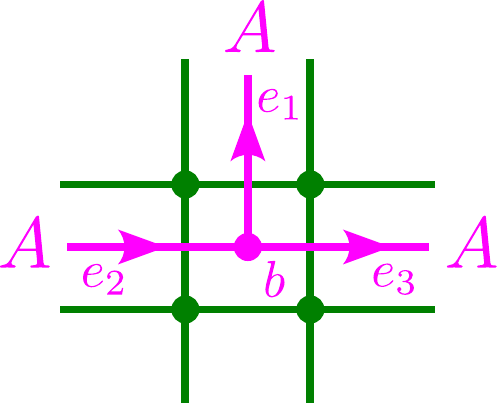}}}
\end{align}
The weight assigned to the vertex $a \in \widetilde\Lambda$ of the dual lattice is given by the delta-function $\delta(\alpha_{e_1}\alpha_{e_2}\alpha_{e_3},1)$, i.e.\ all three $\alpha_e$ touching $a$ have to multiply to $1$, and the weight assigned to $b$ is $\frac12 \delta(\alpha_{e_1}\alpha_{e_2}\alpha_{e_3},1)$.
That the product of the prefactors of the two $\delta(\cdots)$-weights is $\frac12$ implies that defect loops cancel (this condition will be used in the gauging procedure below):
\begin{align}
\raisebox{-0.45\height}{\scalebox{.25}{\includegraphics{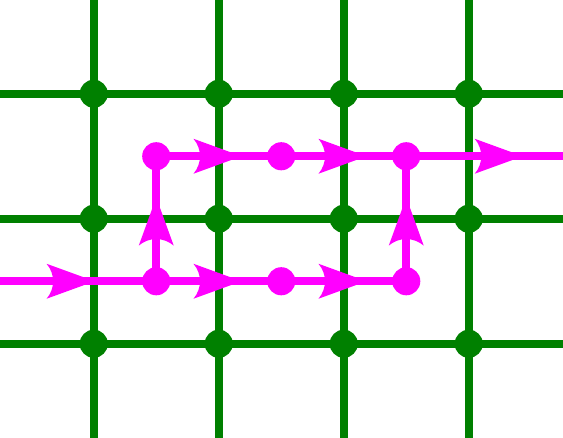}}}
~~=~~
\raisebox{-0.45\height}{\scalebox{.25}{\includegraphics{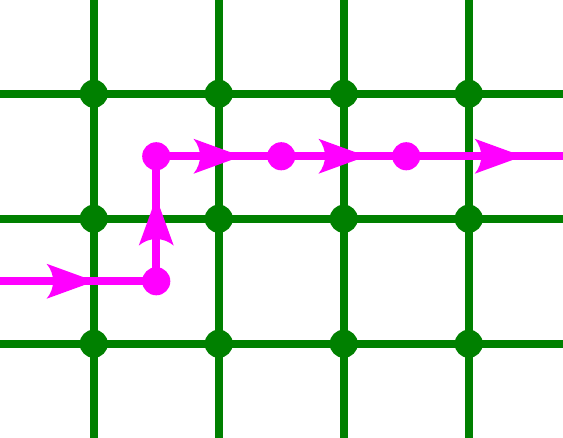}}}
\quad .
\end{align}
Indeed, if we keep the external $\alpha_e$'s fixed, the loop on the left-hand side has two allowed summands, while for the defect line on the right-hand side there is only one. 

The junctions are topological and associative in the sense that they can be moved past each other. We use this to define a four-valent junction
\begin{align}\label{eq:lattice-Adefect-4valent}
\raisebox{-0.45\height}{\scalebox{.3}{\includegraphics{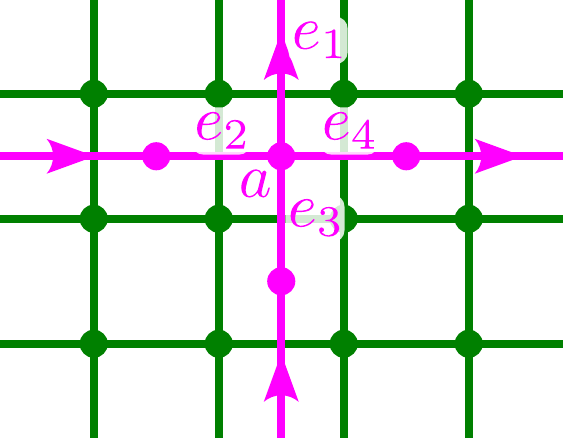}}}
~~:=~~
\raisebox{-0.45\height}{\scalebox{.3}{\includegraphics{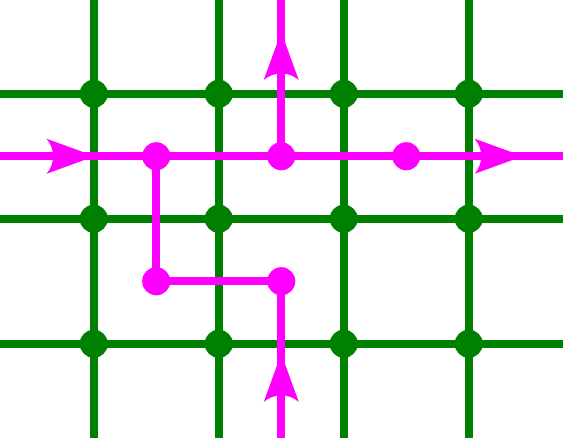}}}
~~=~~
\raisebox{-0.45\height}{\scalebox{.3}{\includegraphics{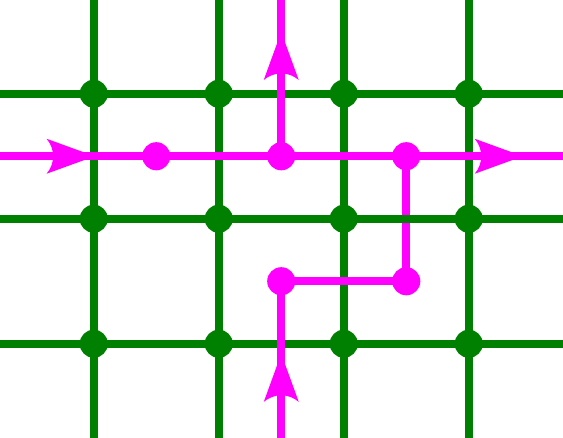}}}
\quad .
\end{align}
In formulas, the weight at site $a \in \tilde\Lambda$ is $\frac12 \delta(\alpha_{e_1}\alpha_{e_2}\alpha_{e_3}\alpha_{e_4},1)$.
In the language of 2d\,QFT in Section~\ref{sec:gauging} below, $A$ and its three-valent junctions correspond to a symmetric $\Delta$-separable Frobenius algebra (the unit and counit also have easy lattice descriptions, but we skip these).

To gauge the $\mathds{Z}_2$-symmetry, one places a network of oriented $A$-defect lines on $\tilde\Lambda$ which meet at three-valent junctions as above. The network has to be ``fine enough'' (see \cite{Carqueville:2023qrk} for more on that condition) -- the partition function is then independent of the precise choice of network. One may thus as well fill $\tilde\Lambda$ entirely, making use of the 4-valent junction \eqref{eq:lattice-Adefect-4valent}. 

Each valid (in the sense of non-zero weight) configuration of $\alpha_e$'s on the edges of $\tilde\Lambda$ can be interpreted as a flat $\mathds{Z}_2$-valued gauge field living on the edges of $\Lambda$, with the $\alpha_e$'s describing the transition functions. 
The sum over the $\alpha_e$'s implicit in evaluating the $A$-network is the sum over flat $\mathds{Z}_2$-gauge field configurations.

Let us refer to the original Ising model in Section~\ref{sec:Ising-and-Z2} by $\mathrm{Is}_\Lambda(\beta)$ and denote the $\mathds{Z}_2$-gauged version of the Ising model 
where each link of $\tilde\Lambda$ is occupied by $A$
by $\mathrm{Is}_\Lambda(\beta)/\mathds{Z}_2$. 
The partition function of the latter is
\begin{align}\label{eq:IsingZ2-PF}
  Z = \sum_s \sum_\alpha 
  \prod_v \tfrac12 \delta(\alpha_{e_1}\alpha_{e_2}\alpha_{e_3}\alpha_{e_4},1)
  \prod_{ \langle ij \rangle } \textrm{e}^{\beta \alpha_{\langle ij \rangle} s_i s_j} 
  \prod_i V~,
\end{align}
where in addition to \eqref{eq:Ising-PF}, $\alpha$ is summed over assignments $\alpha_e \in \{ \pm 1\}$ to edges $e$ of $\tilde\Lambda$, the product $\prod_v$ is over vertices of $\tilde\Lambda$, and $e_1,\dots,e_4$ denote the four edges meeting at $v$.

There is a canonical interface between $\mathrm{Is}_\Lambda(\beta)$ and $\mathrm{Is}_\Lambda(\beta)/\mathds{Z}_2$ which consists of just stopping the network. In fact, this interface is topological as one can easily check with the above rules for $A$-vertices. For example (we omit the orientations), 
\begin{align}\label{eq:interface}
\raisebox{-0.45\height}{\scalebox{.25}{\includegraphics{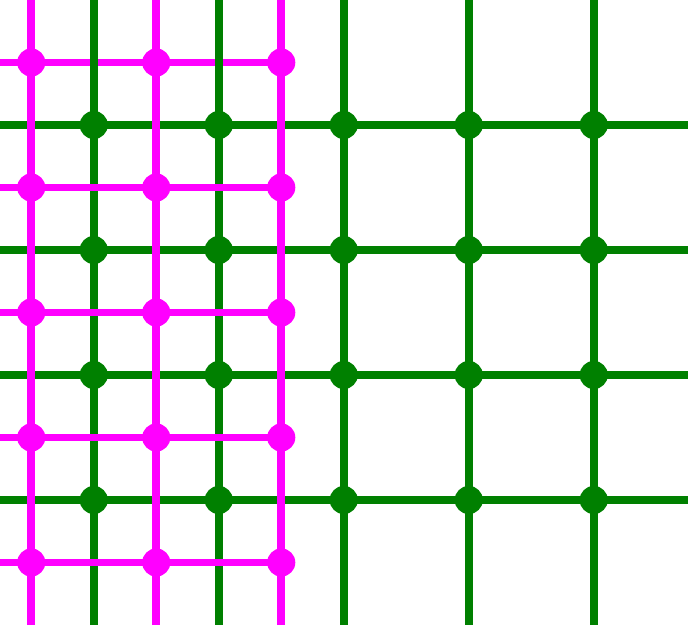}}}
~~=~~
\raisebox{-0.45\height}{\scalebox{.25}{\includegraphics{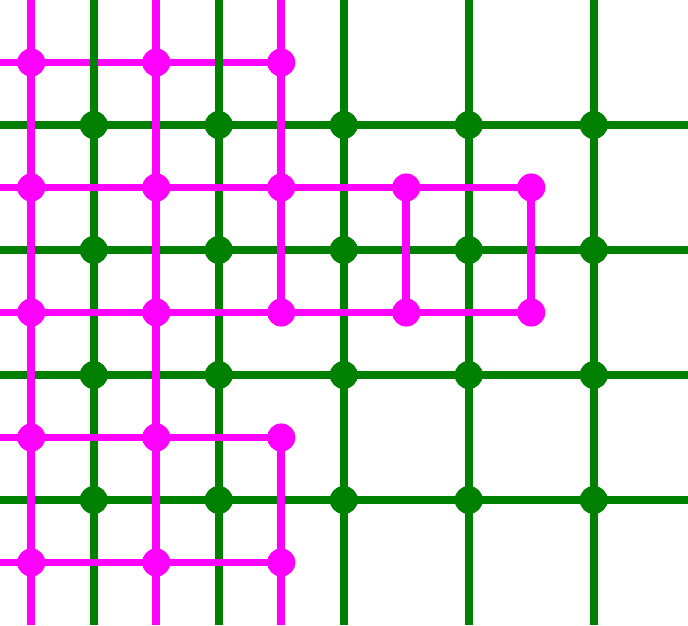}}}
\quad .
\end{align}
Let us denote this interface by $G$. This is our first example of a non-invertible topological defect (we use ``defect'' for both, an interface and an endo-defect). Namely, let $G^*$ be the corresponding defect for a mirrored version of \eqref{eq:interface} where the network is to the right of the defect. By simplifying the $A$-network using the conditions above, one finds $G^* \otimes G \cong A = I \oplus S$, which is not equivalent to the trivial defect $I$.
In terms of the conditions in Section~\ref{subsec:defect2categories} below, $G$ satisfies the invertible bubble conditions in \eqref{eq:weak-invertible-2cat-1st} and \eqref{eq:weak-invertible-2cat-2nd}, but not the reconnection conditions.

\subsection{The duality interface}

The $\mathds{Z}_2$-gauged Ising model $\mathrm{Is}_\Lambda(\beta)/\mathds{Z}_2$ is equivalent to $\mathrm{Is}_{\tilde\Lambda}(\tilde\beta)$, the ungauged Ising model at the dual temperature $\tilde\beta = -\frac12 \ln \tanh(\beta)$ 
on the dual lattice $\tilde\Lambda$. We will demonstrate this by giving an invertible topological interface $D$ between the two. 

The interface $D$ consists of diagonal edges $\langle ia \rangle$ between neighbouring sites $i \in \Lambda$ and $a \in \tilde\Lambda$. Apart from the degrees of freedom $s_i, \alpha_e \in \{\pm 1\}$ assigned to sites $i \in \Lambda$ and edges $e$ of $\tilde\Lambda$, there are also the spins $\sigma_a \in \{ \pm 1\}$ of $\mathrm{Is}_{\tilde\Lambda}(\tilde\beta)$ assigned to sites $a \in \tilde\Lambda$. The different weights are 
\begin{align}\label{eq:weights-duality-defect}
\raisebox{-0.45\height}{\scalebox{.35}{\includegraphics{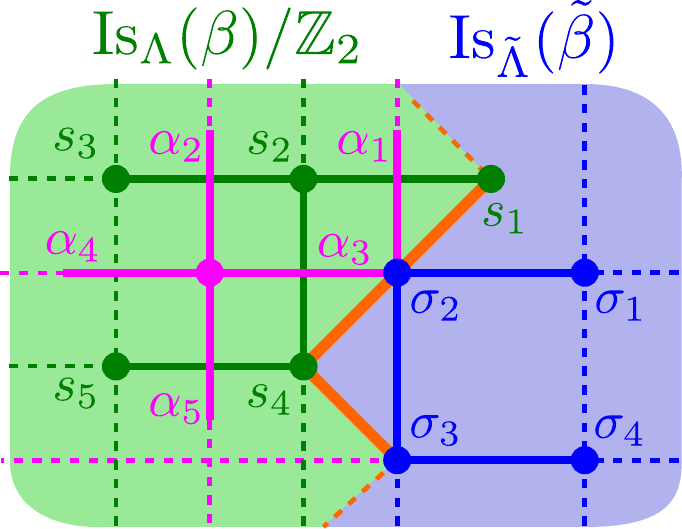}}}
\hspace{4em}
\raisebox{-0.45\height}{\scalebox{.35}{\includegraphics{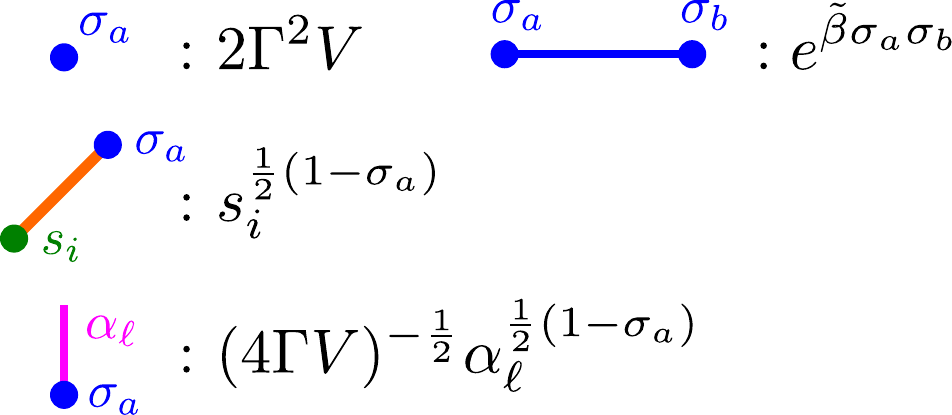}}}
\quad ,
\end{align}
where $\Gamma = \sqrt{ \sinh(\beta) \cosh(\beta) }$. 
For example, the total weight of the part of the above configuration drawn as solid lines is
\begin{align}
&\underbrace{ \tfrac12 \delta(\alpha_2\alpha_3\alpha_4\alpha_5,1) \, V^5 \,
	\textrm{e}^{\beta \alpha_1 s_1 s_2} \,
	\textrm{e}^{\beta \alpha_2 s_2 s_3} \,
	\textrm{e}^{\beta \alpha_3 s_2 s_4} \,
	\textrm{e}^{\beta \alpha_5 s_4 s_5}
}_{\mathrm{Is}_\Lambda(\beta)/\mathds{Z}_2}
\,\cdot\,
\underbrace{
	s_1^{\tau_2} \,
	s_4^{\tau_2} \,
	s_4^{\tau_3} \,
	(4 \Gamma V)^{-1} \, 
	\alpha_1^{\tau_2} \,
	\alpha_3^{\tau_2}
}_{\text{interface}}
\nonumber \\
& \quad \cdot  \,
\underbrace{
    (2 \Gamma^2 V)^4 \,
	\textrm{e}^{\tilde\beta \sigma_1 \sigma_2} \,
	\textrm{e}^{\tilde\beta \sigma_2 \sigma_3} \,
	\textrm{e}^{\tilde\beta \sigma_3 \sigma_4}
}_{\mathrm{Is}_{\tilde\Lambda}(\tilde\beta)}
\ ,
\end{align}
where $\tau_a = \frac12(1-\sigma_a)$.
That the interface is topological follows from the local identities
\begin{align}\label{eq:Ising-duality-topological}
\raisebox{-0.45\height}{\scalebox{.30}{\includegraphics{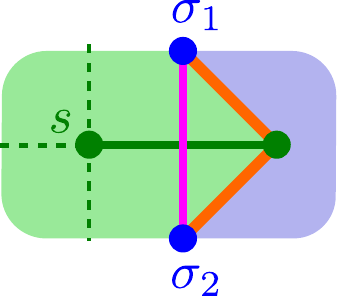}}}
~=~
\raisebox{-0.45\height}{\scalebox{.30}{\includegraphics{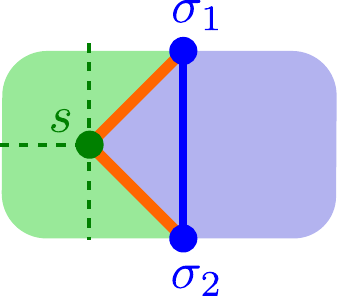}}}
\quad , \qquad
\raisebox{-0.45\height}{\scalebox{.30}{\includegraphics{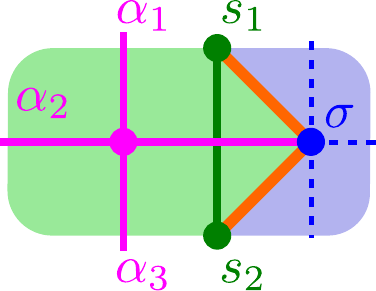}}}
~=~
\raisebox{-0.45\height}{\scalebox{.30}{\includegraphics{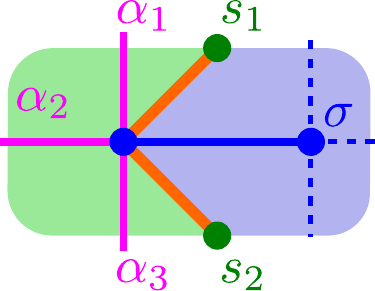}}}
\quad .
\end{align}
In these identities, the unlabelled degrees of freedom are summed over. 
The verification of the properties of~$D$ is based on the following two elementary identities: 
For $\varepsilon \in \{ \pm 1\}$ 
and $c_0(\beta) := \cosh(\beta)$, $c_1(\beta) := \sinh(\beta)$
one has
\begin{align}\label{eq:elementary-identities}
\textrm{e}^{\beta \varepsilon} = \sum_{k=0}^1 c_k(\beta) \varepsilon^k
\quad , \quad
c_{(1-\varepsilon)/2}(\beta) = \Gamma \textrm{e}^{\tilde\beta \varepsilon} ~.
\end{align}
For example, the computation leading to the first equality in \eqref{eq:Ising-duality-topological} is (denoting the unlabelled spin variable by $t$, the unlabelled edge variable by $\gamma$, and writing $\tau_a = \frac12(1-\sigma_a)$ as above)
\begin{align}
&  V^2 (2 \Gamma^2 V)^2 
\hspace{-.5em} \sum_{t,\gamma \in \{\pm 1\}} \hspace{-.5em}
\textrm{e}^{\beta \gamma st} t^{\tau_1} t^{\tau_2} 
(4 \Gamma V)^{-1} \gamma^{\tau_1} \gamma^{\tau_2}
\overset{\eqref{eq:elementary-identities}}= 
V^3 \Gamma^3 \sum_{k=0}^1  \sum_{t,\gamma} c_k(\beta) \gamma^{k+\tau_1+\tau_2} s^{k} t^{k+\tau_1+\tau_2} 
\nonumber \\
& \overset{(*)}=\, 4 V^3 \Gamma^3  c_{\frac12(1-\sigma_1 \sigma_2)}(\beta) s^{\tau_1+\tau_2}    
\overset{\eqref{eq:elementary-identities}}= 
V (2 \Gamma^2 V)^2 s^{\tau_1} s^{\tau_2}  \, \textrm{e}^{\tilde\beta \sigma_1 \sigma_2} ~.
\end{align}
In step $(*)$, the sum over $t$ forces $k = \tau_1+\tau_2 \mod 2$, and 
$\tau_1+\tau_2 = \frac12(1-\sigma_1 \sigma_2) \mod 2$.

Invertibility of $D$ is established by checking invertibility of the two bubbles, as well as the reconnection identity
\begin{align}\label{eq:Ising-bubbles+reconnect}
\raisebox{-0.45\height}{\scalebox{.30}{\includegraphics{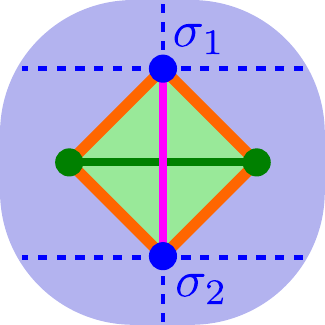}}}
\,&= 2V ~
\raisebox{-0.45\height}{\scalebox{.30}{\includegraphics{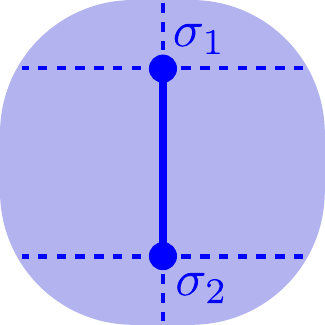}}}
\quad , \qquad
\raisebox{-0.45\height}{\scalebox{.30}{\includegraphics{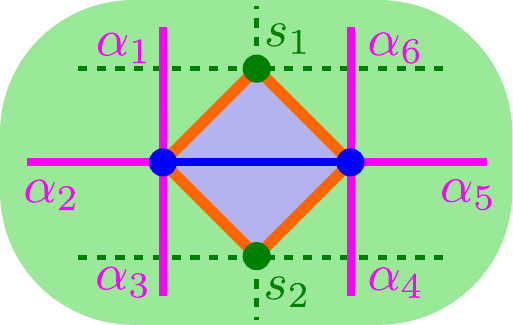}}}
\,= \frac{1}{2V}~
\raisebox{-0.45\height}{\scalebox{.30}{\includegraphics{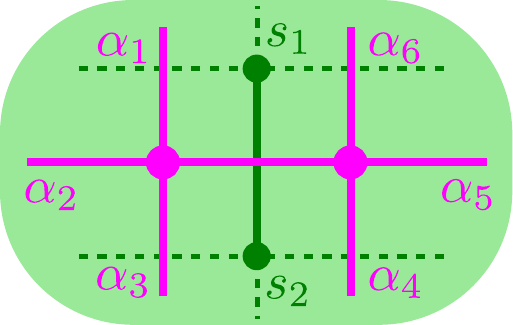}}}
\quad ,
\nonumber
\\[.6em]
\raisebox{-0.45\height}{\scalebox{.30}{\includegraphics{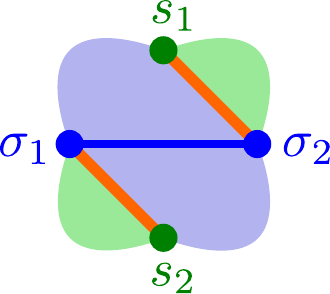}}}
\,&= 2V \,
\raisebox{-0.45\height}{\scalebox{.30}{\includegraphics{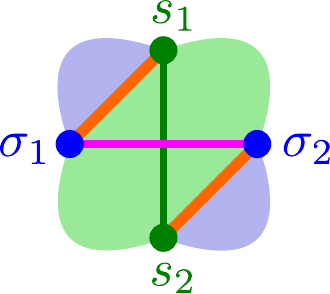}}}
~~.
\end{align}
In summary we discussed the topological defects
\begin{equation}\label{eq:Ising-defect-summary}
\begin{tikzcd}
\mathrm{Is}_{\Lambda}(\beta)
\arrow[bend right]{r}[swap]{G}
\arrow[loop left]{}{I,S,A}
&
\mathrm{Is}_\Lambda(\beta)/\mathds{Z}_2
\arrow[bend right]{l}[swap]{G^*}
\arrow[bend right]{r}[swap]{D^{-1}}
&
\mathrm{Is}_{\tilde\Lambda}(\tilde\beta)
\arrow[bend right]{l}[swap]{D}
\end{tikzcd}
\end{equation}
Of these, $I$, $S$, $D$ are invertible, but $G$ and $A$ are not. The composition $G^* \otimes D$ is the classical Kramers--Wannier duality relating $\mathrm{Is}_{\Lambda}(\beta)$ and $\mathrm{Is}_{\tilde\Lambda}(\tilde\beta)$ \cite{Kramers:1941}. 

\medskip

We conclude the discussion of the Ising lattice model with a quick comment on the relation to the continuum Ising CFT and its topological defects as described in more detail in Examples~\ref{ex:DQ-examples}\,(2) and~\ref{ex:surface-defect-interfaces}\,(2) below.
The critical point of the Ising model is achieved for $\beta = \tilde\beta$, which happens for $\beta_* = \frac12 \ln(1+\sqrt{2})$, or, equivalently, for $\Gamma = 1/\sqrt{2}$.
Note that this is precisely the value for which the vertex weights $V$ and $2 \Gamma^2 V$ on the two sides of~$D$ in \eqref{eq:weights-duality-defect} are equal. 
The value of $V$ affects the two bubbles in \eqref{eq:Ising-bubbles+reconnect}. 
In the continuum model this freedom can be understood as tensoring with an invertible 2d\,TFT, a so-called \textsl{Euler TFT}. 
The value of the Euler TFT for $\lambda \in \mathds{C}^\times$
on a surface $\Sigma$ is $\lambda^{\chi(\Sigma)}$, where $\chi(\Sigma)$ is the Euler character
(see also Footnote~\ref{fn:EulerTFT} in Section~\ref{sec:TFT}).
If $\Sigma$ is a disc, we have $\chi(\Sigma)=1$,
and so multiplying the gauged Ising CFT of the continuum limit by this Euler TFT changes the values of the bubbles corresponding to \eqref{eq:Ising-bubbles+reconnect} by $\lambda$ and $\lambda^{-1}$, respectively.
For the choice $V=1/\sqrt{2}$, the left and right quantum dimension of $G^* \otimes D$ is $\sqrt{2}$, as one would expect from the continuum Ising CFT. This value for $V$ agrees with the choice made in \cite{Aasen:2016dop}.

\section{2-dimensional quantum field theory}\label{sec:2dQFT}

Intuitively, an interface (or defect) $\delta$ between two 2-dimensional quantum field theories (QFTs) is a line which separates them, together with local interactions relating the fields to either side of~$\delta$, and those on~$\delta$.
Typically, amplitudes depend on the precise shape of the defect line. If this is not the case, i.e.\ if we can deform the defect line $\delta$
without changing the amplitude 
as long as it does not cross field insertions or defect lines, we call $\delta$ a topological defect.

Below we review how to assign a pivotal 2-category to 2d\,QFTs with topological defects. We then focus on the situation that a subset of endo-defects forms a pivotal fusion category and discuss some applications, as well as a connection to 3-dimensional topological QFTs.

\subsection{The 2-category of 2d\,QFTs and topological line defects}\label{sec:2cat-form-def}

Suppose we are given a 2d\,QFT $Q$ which is defined (at least) on discs $S$ embedded in $\mathds{R}^2$, possibly with smaller non-overlapping discs removed. $Q$ assigns a $\mathds{C}$-vector space $\mathcal{H}_r$ to a circle of radius $r$, called the \textit{state space} (this could be a Hilbert space or some other topological vector space, but we do not need this structure here). To $S$ the QFT assigns a multilinear map
\begin{align}
	Q(S) \colon \mathcal{H}_{r_1} \times \cdots \times \mathcal{H}_{r_n} \lra \mathcal{H}_R 
 \quad,\qquad
 S =
\raisebox{-0.45\height}{\scalebox{.45}{\includegraphics{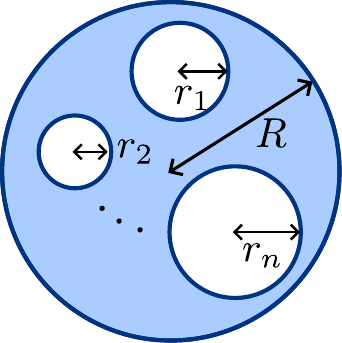}}}
~,
\end{align}
where $r_i$ are the radii of the inner discs removed from a larger disc of radius $R$. The maps $Q(S)$ are called \textit{amplitudes} and are required to be compatible with gluing smaller discs into larger ones.
The surface $S$ is equipped with the standard metric on $\mathds{R}^2$ and the amplitudes $Q(S)$ in general depend on the positions of the holes in $S$. 
This dependence is absent in the special case that $Q(S)$ is obtained from a TFT.

\begin{figure}[tb]
\centering

\raisebox{4em}{a)}	
\hspace{-2em}
\raisebox{-0.45\height}{\scalebox{.5}{\includegraphics{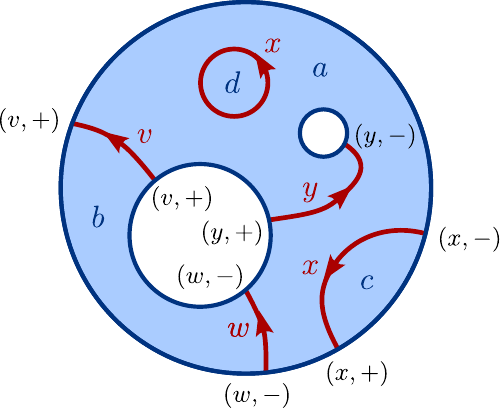}}}
\hspace{2em}
\raisebox{4em}{b)}	
\raisebox{-0.45\height}{\scalebox{.5}{\includegraphics{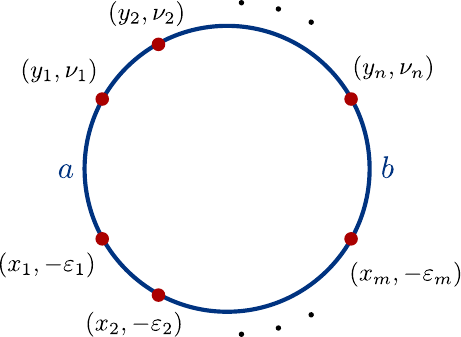}}}
\hspace{1em}
\raisebox{4em}{c)}	
\raisebox{-0.45\height}{\scalebox{.5}{\includegraphics{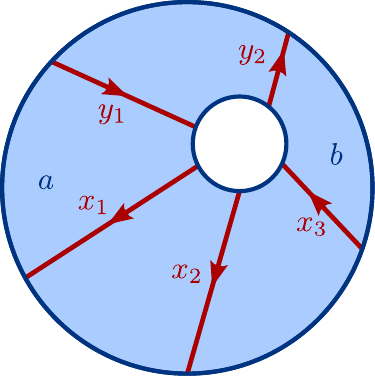}}}

\caption{a) Example of a surface $S$ with labels $a,b,c,d \in D_2$ for 2d patches and $v,w,x,y,z \in D_1$ for line components such that $s(v)=a$, $t(v)=b$, etc. b) Circle with marked points labelled by $\underline x^*$ and $\underline y$ to which $Q$ assigns the state space $\mathcal{H}_r(\underline x, \underline y)$. c) Example of a surface $S$ entering the condition for scale and translation invariant states.}
\label{fig:disc-with-discs-removed}

\end{figure}

We need a more general setting, where $S$ can have a finite number of embedded oriented lines, which are closed loops or intervals starting and ending transversally on the boundary of $S$. 
2-dimensional patches are labelled by elements of a set $D_2$ (2d bulk theories), and connected components of lines by elements of a set $D_1$ (defect conditions), see Figure~\ref{fig:disc-with-discs-removed}\,a). 
Two maps $s,t\colon D_1 \to D_2$ (source and target) describe which bulk theories are to the left and right of a given defect line labelled $x$, namely $t(x)$ and $s(x)$, respectively. For each label $a \in D_2$, one obtains a 2d\,QFT without defects by placing no lines on $S$ and labelling $S$ with $a$. In this sense one can think of $Q$ as describing a collection of 2d\,QFTs, indexed by $D_2$, together with line defects between them.

Let us call the rightmost point of each boundary circle of $S \subset \mathds{R}^2$ its \textit{preferred point}.
For example, the preferred point of the unit circle centred at the origin is  $(1,0) \in \mathds{R}^2$.
We require the defect conditions in $D_1$ to be topological defect lines for $Q$ in the following sense:\footnote{The existence of a topological defect between two bulk theories is a restrictive condition. 
For example, two CFTs can have a topological interface only if their central charges agree. 
If instead they are related by a non-trivial renormalisation group flow, they are joined by non-topological interfaces like those in \cite{Brunner:2007ur,Gaiotto:2012np}.}
Firstly, the state spaces $\mathcal{H}_r$ only depend on the order in which the defect lines meet the boundary, starting from the preferred point, and not on the precise places where they meet the boundary. 
Secondly, the amplitude $Q(S)$ is invariant under deformations of the embedded lines, as long as the lines do not touch each other and no line-endpoint is taken across the preferred point of a boundary circle. 

From this data, we can define a 2-category $\mathcal{D}_Q$ with adjoints \cite[Sect.\,2.4]{Davydov:2011kb}. 
The objects of $\mathcal{D}_Q$ are the set $D_2$. The 1-morphisms are compatible lists of labels with orientations $\underline x = ( (x_1,\varepsilon_1) , \dots , (x_n,\varepsilon_n) )$ where $x_i \in D_2$, $\varepsilon \in \{\pm 1\}$, and where at each step source and target agree. To this end, we extend the definition of $s,t\colon D_1 \to D_2$ to $s(x,+) = s(x)$, $s(x,-) = t(x)$, $t(x,+) = t(x)$, $t(x,-) = s(x)$ and require that $s(x_1,\varepsilon_1) = t(x_2,\varepsilon_2)$, etc. For $a,b \in D_2$, the list $\underline x$ is a 1-morphism $a \to b$ if $s(x_n,\varepsilon_n) = a$ and $t(x_1,\varepsilon_1) = b$. 
The unit 1-morphism $1_a\colon a \to a$ for $a \in D_2$ is the empty list $1_a = ()$. The composition of 1-morphisms is by concatenation of lists.
The description of 2-morphisms requires some preparation. 

For a list $\underline x$ as above, let the \textit{conjugate list} $\underline x^*$ be $((x_n,-\varepsilon_n),\dots,(x_1,-\varepsilon_1))$.
For two lists $\underline x$, $\underline y$ of defect conditions, write $\mathcal{H}_r(\underline x, \underline y)$ for the state space of a circle of radius $r$ with lower half intersected by defects in $\underline x^*$ and upper half by $\underline y$, see Figure~\ref{fig:disc-with-discs-removed}\,b). 
The vector space of \textit{scale and translation invariant states} $\mathcal{H}^\mathrm{inv}(\underline x, \underline y)$ consists of families $\{\psi_r\}_{r \in \mathds{R}_{>0}}$ where $\psi_r \in \mathcal{H}_r(\underline x, \underline y)$. The family has to satisfy $\psi_R = Q(S)(\psi_r)$ for all $R>r$ and all discs $S$ of radius $R$ with a disc of radius $r$ removed and boundary intersected by $\underline x^*$, 
$\underline y$ as above with straight defect lines placed on $S$, see Figure~\ref{fig:disc-with-discs-removed}\,c) for an example. 
One may think of $\mathcal{H}^\mathrm{inv}(\underline x, \underline y)$ as ``vacuum states'' in $\mathcal{H}_r(\underline x, \underline y)$. 
It may happen that $\mathcal{H}^\mathrm{inv}(\underline x, \underline y) = \{0\}$.

\begin{figure}[tb]
\centering

\raisebox{4em}{a)}	
~
\raisebox{-0.45\height}{\scalebox{.5}{\includegraphics{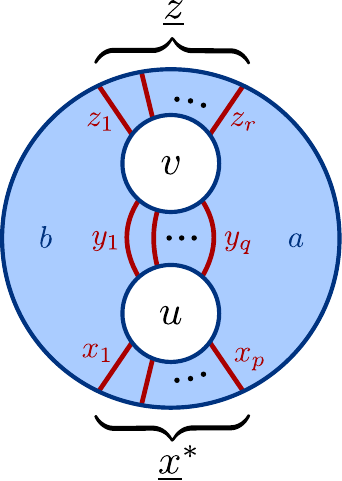}}}
\hspace{3em}
\raisebox{4em}{b)}	
~
\raisebox{-0.45\height}{\scalebox{.5}{\includegraphics{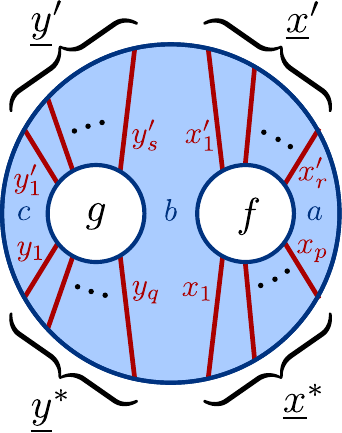}}}
\hspace{3em}
\raisebox{4em}{c)}
~	
\raisebox{-0.45\height}{\scalebox{.5}{\includegraphics{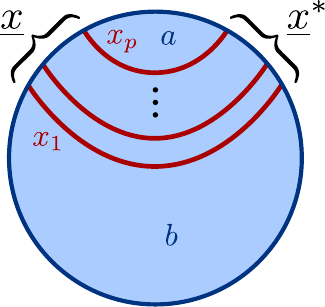}}}
	
	\caption{Surfaces defining (a) the vertical composition $\underline x \xrightarrow{u} \underline y \xrightarrow{v} \underline z$ and (b) the horizontal composition $(\underline y \xrightarrow{g} \underline y') \otimes (\underline x \xrightarrow{f} \underline x')$, as well as (c) the adjunction 2-morphism $1_b \to \underline x \otimes \underline x^*$.}
	\label{fig:horiz-vert-compos}
\end{figure}

For 1-morphisms $\underline x , \underline y\colon a \to b$ in $\mathcal{D}_Q$, the space of 2-morphisms $\underline x \to \underline y$ is defined to be $\mathcal{H}^\mathrm{inv}(\underline x, \underline y)$.
The horizontal and vertical composition of 2-morphisms are given by evaluating $Q$ on the surfaces shown in Figure~\ref{fig:horiz-vert-compos}, see \cite[Sect.\,2.4]{Davydov:2011kb} for details.
Note that working with families of states, rather than with a subspace of
$\mathcal{H}_r(\underline x, \underline y)$ for a fixed $r$, is necessary to be able to define these compositions.

States $u \in \mathcal{H}^\mathrm{inv}(\underline x, \underline y)$ can be interpreted as topological point defects sitting at the junction of line defects, a point of view we will often use below:
\begin{align}\label{eq:defect-junction}
\raisebox{-0.45\height}{\scalebox{.45}{\includegraphics{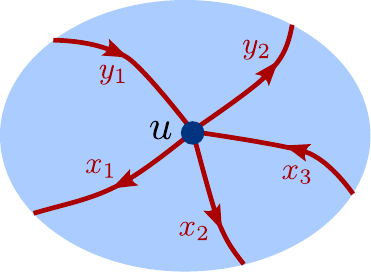}}}
~~:=~~
\raisebox{-0.45\height}{\scalebox{.45}{\includegraphics{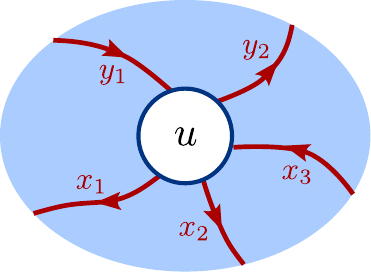}}}
\end{align}

The 1-morphism $\underline x^*$ is a simultaneous left and right adjoint of $\underline x$ in $\mathcal{D}_Q$. The various unit and counit 2-morphisms exhibiting $\underline x^*$ as a two-sided adjoint are given by evaluating $Q$ on suitable discs, see Figure~\ref{fig:horiz-vert-compos}\,c) for an example. The adjunctions satisfy a compatibility with the composition of 1-morphisms so that $\mathcal{D}_Q$ is a \textit{pivotal 2-category}, see e.g.\ \cite[Sect.\,2.2]{Carqueville:2016nqk}. 
The above discussion is summarised as:

\begin{proposition}
    A collection of 2d\,QFTs $Q$ 
    with bulk theories~$D_2$ and topological defects~$D_1$ gives rise to a pivotal 2-category $\mathcal{D}_Q$.
\end{proposition}

For each $\underline x\colon a \to b$, the adjunction morphisms define the left quantum dimension $\dim_\mathrm{l}(\underline x) \in \mathcal{H}^{\mathrm{inv}}(1_a,1_a)$ and the right quantum dimension $\dim_\mathrm{r}(\underline x) \in \mathcal{H}^{\mathrm{inv}}(1_b,1_b)$ by evaluating $Q$ on a disc with an embedded  anti-clockwise (resp.\ clockwise) defect circle labelled $\underline x$
(see also \eqref{eq:diml-dimr} below).
For example, if $\underline x$ consists only of the element $ (x,+)$, 
\begin{align}\label{eq:loop-of-interface}
\dim_\mathrm{l}(\underline x) = Q\Bigg(
\raisebox{-0.45\height}{\scalebox{.45}{\includegraphics{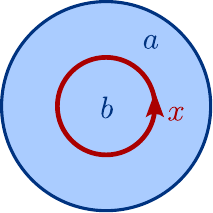}}}
\Bigg)
\quad , \quad
\dim_\mathrm{r}(\underline x) = Q\Bigg( \raisebox{-0.45\height}{\scalebox{.45}{\includegraphics{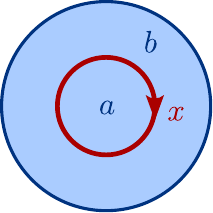}}}
\Bigg)
\quad , 
\end{align}
which do not have to be equal even for $a=b$. 
The quantum dimensions are related by $\dim_\mathrm{l}(\underline x) = \dim_\mathrm{r}(\underline x^*)$.
One may interpret the existence of adjoint 1-morphisms as a finiteness conditions on the topological defects that are allowed in the present formalism. In particular, all quantum dimensions will be finite.
Below we will often implicitly complete $\mathcal{D}_Q$ with respect to finite direct sums. 

\begin{examples}\label{ex:DQ-examples}
\begin{enumerate}
\item In \cite{Davydov:2011kb} (see also \cite{lp0602047, Mule1}) a state sum description of 2d defect TFTs is given for which the resulting 2-category $\mathcal{D}_Q$ is equivalent to $\textrm{ssFrob}$, the 2-category of separable symmetric Frobenius algebras, finite-dimensional bimodules, and bimodule intertwiners.  

\item 
For unitary 2-dimensional conformal field theories of central charge $c=\tfrac12$, $\mathcal{D}_Q$ has one object up to equivalence, namely the Ising CFT.
Its simple 1-morphisms are the three elementary topological endo-defects $\one$, $\sigma$, $\varepsilon$ of the Ising CFT \cite{Petkova:2000ip,Frohlich:2004ef}. Their composition rules are $\sigma \otimes \sigma \cong \one \oplus \varepsilon$ and $\varepsilon \otimes \varepsilon \cong \one$, and the quantum dimensions are $\dim_{\mathrm{l/r}}(\one) = \dim_{\mathrm{l/r}}(\varepsilon) = 1$, $\dim_{\mathrm{l/r}}(\sigma) = \sqrt{2}$.

At $c = \tfrac45$, there are the two unitary CFTs and $\mathcal{D}_Q$ has two objects up to equivalence, given by the tetracritical Ising CFT and the three-state Potts CFT, which are linked by a non-invertible topological interface. The full 2-category is worked out in \cite{Frohlich:2006ch}.

\item In the two examples above, the morphism categories of $\mathcal{D}_Q$ are finitely semisimple. An example where the morphism categories are neither finite nor semisimple is provided by
topologically twisted 2d $\mathcal N=2$ supersymmetric Landau--Ginzburg models. These are 2d\,TFTs with defects described by matrix factorisations, see \cite{Brunner:2007qu,Carqueville:2012st,CMM}. 
Here it can happen that the loop values \eqref{eq:loop-of-interface} are non-zero and nilpotent (and hence non-invertible), as is the case e.g.\ for the boundary-bulk maps of \cite[Sect.\,9]{Carqueville:2012st}. 
\end{enumerate}	
\end{examples}

\subsection{Pivotal fusion categories as symmetries}\label{sec:pivotal-sym}

We now restrict to endo-defects of a given 2d bulk theory $a \in D_2$ of $Q$, and we write $\mathcal{E}_Q(a)$ for the endomorphism category $a \to a$ of the 2-category $\mathcal{D}_Q$. 
If $a$ is clear from the context, we write $\mathcal{E}_Q$ or just $\mathcal{E}$. 
This is a pivotal monoidal category, with tensor product given by the composition of 1-morphisms and duals by the two-sided adjoints in $\mathcal{D}_Q$.

\newcommand{\finGr}{E}

By a group symmetry of a QFT one usually means that a group $G$ acts on the fields of the theory in a way that leaves the correlators invariant. 
Suppose we would like to know if~$Q$ has a given finite group $\finGr$ as symmetry. 
If $G(Q)$ denotes the maximal symmetry group of $Q$, this amounts to giving a group homomorphism
\begin{align}\label{eq:group-sym}
\rho\colon \finGr \lra G(Q) \ .
\end{align}
There may be no such homomorphism, or more than one, and $\rho$ may or may not be injective or surjective. Typically, $G(Q)$ may not be known explicitly, but one may still be able to describe subgroups in this way.

In two dimensions, when describing symmetries by topological defects, the role of the finite group in \eqref{eq:group-sym} is taken over by a \textit{pivotal fusion category}. 
A fusion category is a finitely semisimple $\mathds{C}$-linear monoidal category with duals and with simple tensor unit \cite[Sect.\,4.1]{EGNO}. In particular, there is a finite number of simple objects, and all objects can be written as finite direct sums of these. Let $\mathcal{F}$ be a pivotal fusion category. The generalisation of \eqref{eq:group-sym} is \cite{Fuchs:2002cm,Frohlich:2006ch,Bhardwaj:2017xup,Chang:2018iay}:
\begin{align}\label{eq:piv-fusion-cat-sym}
\begin{minipage}{30em}
	A 2d\,QFT $Q$ \textit{has symmetry $\mathcal{F}$} if there is a pivotal monoidal $\mathds{C}$-linear functor $R \colon \mathcal{F} \to \mathcal{E}_Q$. Then~$R$ is a \textit{realisation} 
    of $\mathcal{F}$ on $Q$.
\end{minipage}	
\end{align}
One also says ``$Q$ has topological (or categorical) symmetry $\mathcal{F}$''.
The functor $R$ is structure, not a property, and the full category $\mathcal{E}_Q$ of topological defects is typically not known explicitly.

The situation in \eqref{eq:piv-fusion-cat-sym} differs from \eqref{eq:group-sym} in two ways: $\mathcal{F}$ may contain objects that are not invertible with respect to $\otimes$, and $R$ maps morphisms of $\mathcal{F}$ to topological defect junctions as in \eqref{eq:defect-junction}, while \eqref{eq:group-sym} gives no information about junctions.

Note that considering fusion category symmetries is a very special case, just as demanding $\finGr$ in \eqref{eq:group-sym} to be a finite group is very restrictive. Non-finite and/or non-semisimple pivotal monoidal categories can equally be considered as topological symmetries (see Example~\ref{ex:DQ-examples}\,(3)).
Using only fusion categories as source for $R$ does in general not allow one to probe all of $\mathcal{E}_Q$.

\begin{example}\label{ex:su(2)k-fusion-cat}
Let $Q$ be the Wess--Zumino--Witten model for $\textrm{SU}(2)$ at level $k \in \mathds{Z}_{>0}$. This is a unitary 2d\,CFT with charge conjugation modular invariant partition function. The full set of (semisimple) topological defects can be described for $k=1$ \cite{Fuchs:2007tx}, but is not known for $k>1$. However, the topological defects which commute with the holomorphic and antiholomorphic currents form the pivotal fusion category $\mathcal{C}(\mathfrak{su}(2),k)$ which has $k+1$ simple objects $\underline 0,\underline 1, \dots, \underline k$ \cite{Petkova:2000ip,Fuchs:2002cm}. Some fusion rules are 
$\underline 0 \otimes \underline \lambda \cong \underline \lambda$, 
$\underline k \otimes \underline \lambda \cong \underline {k-\lambda}$, 
$\underline 1 \otimes \underline \lambda \cong \underline {\lambda-1} \oplus \underline {\lambda+1}$, where the last rule holds for $0<\lambda<k$. 
The object $\underline 0$ is the tensor unit, $\underline k$ is invertible under $\otimes$, and $\underline \lambda$ with $0<\lambda<k$ is not invertible. Altogether, the WZW model for $\textrm{SU}(2)$ at level $k$ has pivotal fusion category symmetry $R\colon \mathcal{C}(\mathfrak{su}(2),k) \to \mathcal{E}_Q$.
\end{example}

The description \eqref{eq:piv-fusion-cat-sym} of topological symmetries has the advantage of dividing the study of such symmetries into a model independent part, namely investigating the properties of a given pivotal category, and the model dependent part of finding pivotal functors $R$.

The functor $R$ is called \textit{essentially surjective} if each $x \in \mathcal{E}_Q$ is isomorphic to $R(u)$ for some $u \in \mathcal{F}$. This is the counterpart of $\rho$ in \eqref{eq:piv-fusion-cat-sym} being surjective. 
$R$ is \textit{faithful} if for each pair of objects $u,v \in \mathcal{F}$, the linear map $\Hom_{\mathcal{F}}(u,v) \to \Hom_{\mathcal{E}_Q}(R(u),R(v))$, 
$f \mapsto R(f)$ is injective. $R$ is called \textit{full} if these linear maps are surjective. 

Any monoidal $\mathds{C}$-linear functor out of a fusion category (with non-zero target)
is automatically faithful,
in particular $R$ in \eqref{eq:piv-fusion-cat-sym} is faithful. But $R$ need not be full. This is the case for example if two non-isomorphic simple objects in $\mathcal{F}$ are mapped to isomorphic objects in $\mathcal{E}_Q$. This cannot happen if $R$ is full, and so $R$ being full (and in any case faithful) is the counterpart of $\rho$ in \eqref{eq:piv-fusion-cat-sym} being injective. 

\medskip

\noindent
\textbf{Graded vector spaces.} 
    Let $\finGr$ be a finite group and let $\mathrm{Vect}_\finGr$ be the category of finite-dimensional $\finGr$-graded $\mathds{C}$-vector spaces. 
    It has simple objects $\mathds{C}_g$, $g \in \finGr$, a 1-dimensional vector space concentrated in degree $g$. 
    The tensor product is $\mathds{C}_g \otimes \mathds{C}_h \cong \mathds{C}_{gh}$. 
	Turning $\mathrm{Vect}_\finGr$ into a pivotal fusion category involves two pieces of data beyond the group structure of $\finGr$: the associator is determined (up to equivalence) by a class $[\omega]$ in $H^3(\finGr,\textrm{U}(1))$, the third group cohomology group of $\finGr$ with values in $\textrm{U}(1)$ \cite[Sect.\,2.3]{EGNO}, and the pivotal structure is fixed by a group homomorphism $d\colon \finGr \to \textrm{U}(1)$, such that $\dim_{\mathrm{l}}(\mathds{C}_g) = d(g)$ and $\dim_{\mathrm{r}}(\mathds{C}_g) = d(g^{-1})$. 
	Denote the resulting pivotal fusion category by $\mathrm{Vect}_\finGr^{\omega,d}$.
	
	Whether or not a (full or non-full) pivotal monoidal functor $R\colon \mathrm{Vect}_\finGr^{\omega,d} \to \mathcal{E}_Q$ exists depends not just on $\finGr$, but also on $\omega$ and $d$. In this sense, the statement that $Q$ has symmetry $\mathrm{Vect}_\finGr^{\omega,d}$ is finer information than just saying that $Q$ has group symmetry $\finGr$. 
    The 3-cocycle $\omega$ and the map $d$ appear as follows in amplitudes of $Q$:
\begin{align}
Q\Bigg(
\raisebox{-0.45\height}{\scalebox{.45}{\includegraphics{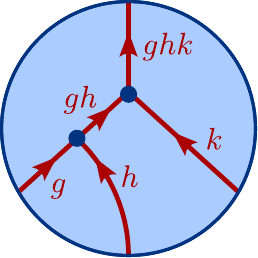}}}
\Bigg)
= \omega(g,h,k) \cdot
Q\Bigg(
\raisebox{-0.45\height}{\scalebox{.45}{\includegraphics{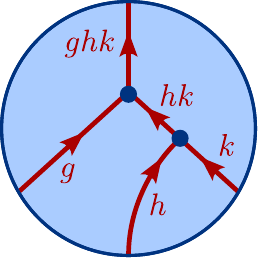}}}
\Bigg)
~ , \quad
Q\Bigg(
\raisebox{-0.45\height}{\scalebox{.45}{\includegraphics{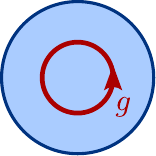}}}
\Bigg)
= d(g) \cdot
Q\Bigg(
\raisebox{-0.45\height}{\scalebox{.45}{\includegraphics{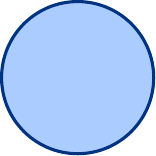}}}
\Bigg)
~, 
\end{align}
where the label $g$ stands for $R(\mathds{C}_g)$, etc.

\medskip

\noindent
\textbf{The case $\finGr = \mathds{Z}_2$.} 
	We have $H^3(\mathds{Z}_2,\textrm{U}(1)) \cong \mathds{Z}_2$ and there are two group homomorphisms $\mathds{Z}_2 \to \textrm{U}(1)$. There are thus four pivotal fusion categories with underlying fusion rules given by the group $\mathds{Z}_2$:
\begin{align}\label{eq:Z2-pivotal-cats}
\mathrm{Vect}_{\mathds{Z}_2} 
~,~~
\mathrm{Vect}_{\mathds{Z}_2}^\omega
~,~~
\mathrm{SVect}_{\mathds{Z}_2}
~,~~
\mathrm{SVect}_{\mathds{Z}_2}^\omega
~.
\end{align}
Here, $\omega$ denotes the non-trivial element in $H^3(\mathds{Z}_2,\textrm{U}(1))$ and $\mathrm{SVect}$ is the category of finite-dimensional super vector spaces, which is distinguished from $\mathrm{Vect}_{\mathds{Z}_2}$ by the quantum dimension being given by the super dimension $\mathrm{sdim}(\mathds{C}^{m|n}) = m-n$. 

Hence there are four inequivalent ways in which a 2d\,QFT can have a $\mathds{Z}_2$-symmetry. 
For example, the first symmetry in \eqref{eq:Z2-pivotal-cats} is realised in the Ising CFT and in the $\textrm{SU}(2)$-WZW models with even level $k$, and the second in those with $k$ odd (see \cite[Sect.\,4]{Runkel:2008gr} for this and other Lie groups). The remaining two cases can only be realised in non-unitary theories, for example in a Bershadsky--Polyakov model \cite[Sect.\,7.5]{Runkel:2022fzi} and in the Virasoro minimal model $M(3,5)$, respectively. Note that here we are considering oriented parity-even QFTs; 
in spin theories $\mathrm{SVect}_{\mathds{Z}_2}$ is realised e.g.\ in the free fermion CFT.

\medskip

\noindent
\textbf{Rank-2 pivotal fusion categories.} 
The number of isomorphism classes of simple objects in a fusion category is called its \textit{rank}. There are six pivotal fusion categories of rank 2 \cite{Ostrik:2002}: The four examples \eqref{eq:Z2-pivotal-cats} and two more examples with fusion rule $\phi \otimes \phi \cong \one \oplus \phi$, called Fibonacci categories.
It is a non-trivial result that there are no rank-2 fusion categories such that $\phi \otimes \phi \cong \one \oplus n \phi$ for $n \geqslant 2$.

\subsection{Ground state degeneracy in 2d\,TFT}\label{sec:groundstatedeg}

A simple example of how topological symmetry can impose constraints on QFTs is provided by TFTs with line defects. 
These are described as symmetric monoidal functors $\zz$ from a suitably stratified and decorated bordism category to $\mathrm{Vect}$, see Section~\ref{sec:TFT} and \cite{Davydov:2011kb,Carqueville:2016nqk}. 
We consider the case that all 2d patches carry the same label, i.e.\ there is a single bulk theory.
Since $\zz$ is defined on discs as in Section~\ref{sec:2cat-form-def}, we obtain a pivotal monoidal category $\mathcal{E}_{\zz}$.

\begin{proposition}\label{prop:2dTFT-fibre-functor}
If $\dim \zz(S^1) = 1$ and $\zz(S^2) \neq 0$, 
the Hom functor $\Hom_{\mathcal{E}_{\zz}}(\one,-)\colon \mathcal{E}_{\zz} \to \mathrm{Vect}$ is pivotal monoidal.
\end{proposition}
	
\begin{proof}
Abbreviate $F = \Hom_{\mathcal{E}_{\zz}}(\one,-)$. By definition of $\mathcal{E}_{\zz}$, $F( \underline y ) = \zz(S^1(\underline y))$, 
where $S^1(\underline y)$ is the unit circle decorated as in Figure~\ref{fig:disc-with-discs-removed}\,b) for $\underline x = ()$. The coherence morphism $\phi_{\underline x, \underline y}\colon F(\underline x) \otimes F(\underline y) \to F(\underline x \otimes \underline y)$ is the horizontal composition $\Hom_{\mathcal{E}_{\zz}}(\one,\underline y) \otimes \Hom_{\mathcal{E}_{\zz}}(\one,\underline y) \to \Hom_{\mathcal{E}_{\zz}}(\one,\underline x \otimes \underline y)$ from Figure~\ref{fig:horiz-vert-compos}\,b).

So far we have used none of our assumptions on $\zz$. These enter in the proof that $\phi_{\underline x, \underline y}$ is invertible. Let $\lambda := \zz(S^2) \in \mathds{C}^\times$. Since $\zz(S^1)=\mathds{C}$, we have
\begin{align}\label{eq:tft-cylinder-identity}
\mathcal{Z}\Big(\,
\raisebox{-0.4\height}{\scalebox{.45}{\includegraphics{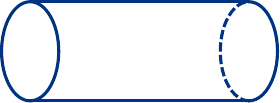}}}
\,\Big)
\,=\,
\frac{1}{\lambda} \,
\mathcal{Z}\Big(\,
\raisebox{-0.4\height}{\scalebox{.45}{\includegraphics{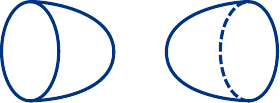}}}
\,\Big)
~.
\end{align}
This implies the following identities, which exhibit~$\lambda \, \zz$ 
applied to the dual pair-of-pants as inverse to $\phi_{\underline x, \underline y}$:
\begin{align}
\lambda\,
\mathcal{Z}\!\left(\,
\raisebox{-0.45\height}{\scalebox{.45}{\includegraphics{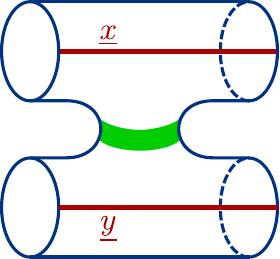}}}
\,\right)
=
\mathcal{Z}\!\left(\,
\raisebox{-0.45\height}{\scalebox{.45}{\includegraphics{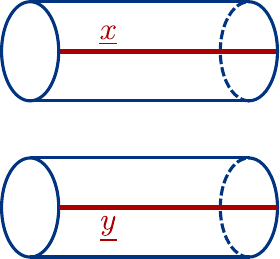}}}
\,\right)
~,~~
\lambda\,
\mathcal{Z}\Big(\,
\raisebox{-0.4\height}{\scalebox{.45}{\includegraphics{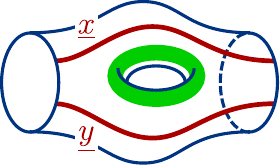}}}
\,\Big)
=
\mathcal{Z}\Big(\,
\raisebox{-0.4\height}{\scalebox{.45}{\includegraphics{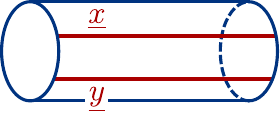}}}
\,\Big)
~.
\nonumber
\\[-1em]
\end{align}
The cylinder marked in green is the part of the bordism where \eqref{eq:tft-cylinder-identity} is applied.

The linear map $\phi_0\colon \mathds{C} \to F(\one)$ which takes $1$ to $\id$ is an isomorphism since $\zz(S^1) = \mathds{C}$. This shows that $F$ is a (strong) monoidal functor.
It is then a matter of algebra to check that if the Hom functor 
is monoidal, it is automatically pivotal; we skip the details.	
\end{proof}

Let now $\mathcal{F}$ be a pivotal fusion category and suppose that~$\zz$ has a realisation~$R$ of the symmetry $\mathcal{F}$. 
Under the assumptions in Proposition~\ref{prop:2dTFT-fibre-functor} we obtain a pivotal monoidal functor 
$\mathcal{F} \xrightarrow{R} \mathcal{E}_{\zz}\xrightarrow{\Hom(\one,-)} \mathrm{Vect}$.
Since $\mathcal{F}$ is semisimple, the composite functor $\mathcal{F} \to \mathrm{Vect}$ is exact, i.e.\ it is a \textit{(pivotal) fibre functor}. This in turn implies that $\mathcal{F} \cong \mathrm{Rep}(H)$ as (pivotal) monoidal categories for some (pivotal) Hopf algebra $H$ \cite[Sect.\,5.3]{EGNO}. 
Furthermore, if $R$ is full then already $\mathcal{F} \cong \mathrm{Vect}$ (as only then a fully faithful fibre functor exists).
See \cite{Thorngren:2019iar} for the relation between topological symmetries of TFTs and fibre functors.

If $\mathcal{F}$ does not allow for a pivotal monoidal functor $\mathcal{F} \to \mathrm{Vect}$, then it cannot be the topological symmetry of a TFT with unique vacuum state. For example, of the six rank-2 pivotal fusion categories, only one, namely $\mathrm{Vect}_{\mathds{Z}_2}$, allows for a pivotal monoidal functor to $\mathrm{Vect}$.
A related argument has been used for example to constrain the vacuum multiplicity in massive renormalisation group flows \cite{Chang:2018iay},
and in \cite{Komargodski:2020mxz} non-invertible topological symmetries have been used to study renormalisation group flows in 2d quantum chromodynamics.
Note that a group symmetry as in \eqref{eq:group-sym} is not strong enough to obtain such restrictions, as $\mathrm{Vect}_\finGr$ always allows for a pivotal fibre functor. 

\medskip

An $\mathcal{F}$-\textit{module category} is an additive $\mathds{C}$-linear category $\mathcal{M}$ together with a monoidal functor $\mathcal{F} \to \mathcal{E}nd(\mathcal{M})$, the endofunctors of $\mathcal{M}$. We focus on the case that $\mathcal{M}$ is finitely semisimple, i.e.\ $\mathcal{M} \cong \mathrm{Vect}^{\oplus r}$ as $\mathds{C}$-linear categories for some rank $r>0$. For example, since $ \mathcal{E}nd(\mathrm{Vect}) \cong \mathrm{Vect}$, the existence of a monoidal functor $\mathcal{F} \to \mathrm{Vect}$ is equivalent to $\mathcal{F}$ having a rank-1 module category.

If one applies the state sum construction of 2d\,TFTs with defects from \cite{Davydov:2011kb} to the algebra $A = \mathds{C}^r$ with element-wise product and addition, the category $\mathcal{E}$ of endo-defects (given by $A$-$A$-bimodules) is equivalent to $\mathcal{E}nd(\mathrm{Vect}^{\oplus r})$ as a monoidal category. Furthermore, in this case $\zz(S^1) = \mathds{C}^r$. 
This shows that if $\mathcal{F}$ has a rank-$r$ module category, $\mathcal{F}$ can be realised as topological symmetries of a 2d\,TFT with $\dim \zz(S^1)=r$ \cite{Huang:2021zvu}.

\begin{example}\label{ex:su2-k-module-cats}
The indecomposable module categories (i.e.\ those not equivalent to a direct sum) over $\mathcal{C}(\mathfrak{su}(2),k)$ from 
Example~\ref{ex:su(2)k-fusion-cat} are classified in \cite{Ostrik:2001} and follow an ADE pattern. 
Namely, for each $k$, $\mathcal{M}_{\mathrm{A}_{k+1}} = \mathcal{C}(\mathfrak{su}(2),k)$ is a module over itself of rank $k+1$. For even $k$, there is also a module category $\mathcal{M}_{\mathrm{D}_{k/2+2}}$ of rank $k/2+2$. Additionally, for $k=10,16,28$ there are module categories $\mathcal{M}_{\mathrm{E}_6}$, $\mathcal{M}_{\mathrm{E}_7}$, $\mathcal{M}_{\mathrm{E}_8}$ of ranks $6$, $7$, $8$, respectively.
\end{example}

\subsection{Gauging topological symmetries}\label{sec:gauging}

One can use topological defects in a 2d\,QFT to produce a new 2d\,QFT in a procedure called \textit{generalised orbifold}, or \textit{internal state sum}, or \textit{gauging} \cite{Fuchs:2002cm,Frohlich:2009gb,Carqueville:2012dk}.
The procedure exists in any dimension and is the same for QFTs as it is for TFTs, see \cite{Carqueville:2023qrk} for the latter.

In our setting, we start from a 2d\,QFT on oriented surfaces where all 2d\,patches carry the same bulk theory. 
The datum $\mathcal A$ we need to fix consists of a topological endo-defect $A \in \mathcal{E}_Q$, which will typically not be simple, and two three-fold topological junctions $\mu \in \Hom_{\mathcal{E}_Q}(A \otimes A,A)$, $\Delta \in \Hom_{\mathcal{E}_Q}(A, A \otimes A)$.
For the algebraic description it is convenient to also require two topological 
``1-fold junctions'' $\eta \in \Hom_{\mathcal{E}_Q}(\one,A)$, $\varepsilon \in \Hom_{\mathcal{E}_Q}(A,\one)$.
The conditions are that $A$ forms a $\Delta$-separable symmetric Frobenius algebra in $\mathcal{E}_Q$ with product $\mu$, coproduct $\Delta$, unit $\eta$ and counit $\varepsilon$ (see e.g.~\cite{Fuchs:2008fa} for more on Frobenius algebras).
These conditions imply that $\eta$, $\varepsilon$ are unique if they exist.
In terms of defect lines, the conditions are
\begin{align}\label{eq:fusion+bubble}
Q\Bigg(\!\!\raisebox{-0.45\height}{\scalebox{.45}{\includegraphics{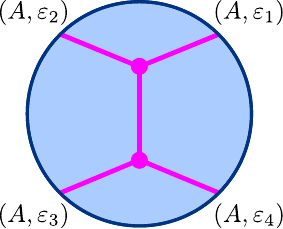}}}
\!\!\Bigg)
=
Q\Bigg(\!\!\raisebox{-0.45\height}{\scalebox{.45}{\includegraphics{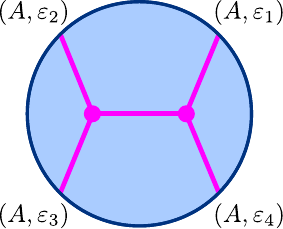}}}
\!\!\Bigg)
~,~~
Q\Bigg(\,\raisebox{-0.45\height}{\scalebox{.45}{\includegraphics{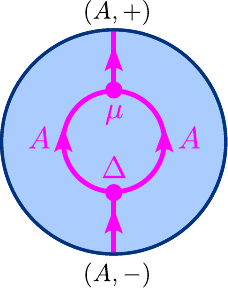}}}
\,\Bigg)
=
Q\Bigg(\,\raisebox{-0.45\height}{\scalebox{.45}{\includegraphics{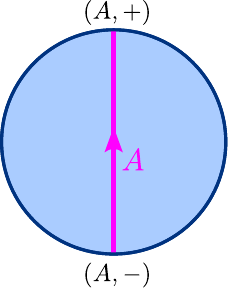}}}
\,\Bigg)
~,
\end{align}
together with (co)unitality which we do not draw. The first equality has to hold for all allowed choices of line orientations, and amounts to (co)associativity, symmetry, and the Frobenius property. The second diagram is $\Delta$-separability. 
We call $\mathcal{A} = (A, \mu, \Delta)$ an \textit{orbifold datum} or a \textit{gaugeable symmetry}.

The gauged theory $Q/\mathcal{A}$ is defined on an oriented surface $\Sigma$ by triangulating $\Sigma$ and choosing orientations for the edges of the triangulation such that the boundary of each triangle is not an oriented loop. 
Then one places an oriented 3-valent $\mathcal{A}$-graph on the dual of the triangulation and evaluates with $Q$, see Figure~\ref{fig:triang-flip-Nak}\,a).
The conditions in \eqref{eq:fusion+bubble} guarantee that the result is independent of the choice of triangulation and edge orientation (if~$\Sigma$ is closed; for non-empty boundary there is an additional step that amounts to splitting an idempotent \cite[Sect.\,4.2]{Carqueville:2023qrk}).

\begin{figure}[bt]
\centering

\raisebox{5em}{a)}	
~
\raisebox{-0.45\height}{\scalebox{.5}{\includegraphics{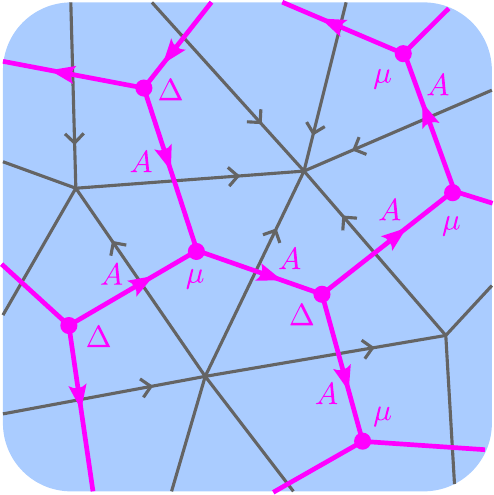}}}
\hspace{3em}
\raisebox{5em}{b)}	
\begin{minipage}{10em}
\begin{align*}
&Q\Bigg(\raisebox{-0.45\height}{\scalebox{.5}{\includegraphics{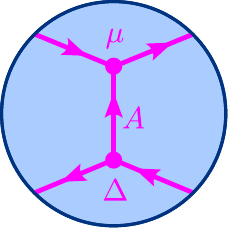}}}
\Bigg)
\\
&= 
Q\Bigg(\raisebox{-0.45\height}{\scalebox{.5}{\includegraphics{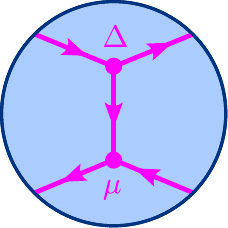}}}\Bigg)
\end{align*}
\end{minipage}
\hspace{1em}
\raisebox{5em}{c)}
$N_\mathcal{A} 
=\hspace{-0.5em} \raisebox{-0.5\height}{\scalebox{.5}{\includegraphics{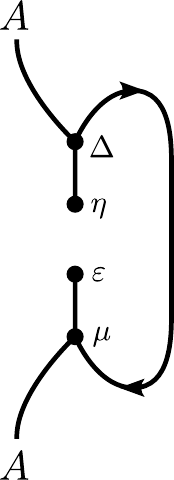}}}$

\caption{a) Part of a triangulation with oriented edges and the corresponding dual $A$-network. b) Invariance under orientation flip of the internal edge requires $\mathcal{A}$ to be symmetric. c) The Nakayama automorphism $N_\mathcal{A}$ of $\mathcal{A}$ as a string diagram.}
\label{fig:triang-flip-Nak}

\end{figure}

The requirement for $\mathcal{A}$ to be symmetric is needed to flip the orientation of an $A$-defect in the situation shown in Figure~\ref{fig:triang-flip-Nak}\,b). This condition is best understood by introducing the \textit{Nakayama automorphism} $N_\mathcal{A}\colon A \to A$, shown as a string diagram in Figure~\ref{fig:triang-flip-Nak}\,c), see~\cite{Fuchs:2008fa}.
Note that this needs all the structure we have at our disposal: 
$A$, all structure morphisms of $\mathcal{A}$, as well as the tensor product and the pivotal structure of $\mathcal{E}_Q$.
$N_\mathcal{A}$ is automatically a Frobenius algebra automorphism of $\mathcal{A}$, and by definition $\mathcal{A}$ is symmetric if $N_\mathcal{A} = \id_A$.

If $N_\mathcal{A} \neq \id_A$, the above gauging construction is not well-defined as it depends on the choices made. However, if $(N_\mathcal{A})^r = \id_A$ for some $r \geqslant 1$, the gauged theory is well-defined on surfaces equipped with an $r$-spin structure (corresponding to the $r$-fold cover of $\textrm{SO}(2)$), with $r=2$ being the usual spin structure. 
If $N_\mathcal{A}$ has infinite order, one needs the universal cover $\mathds{R} \to \textrm{SO}(2)$ which is equivalent to a framing, see e.g.~\cite{Runkel:2018feb}. 
The $r$-spin structure is encoded in the triangulation in Figure~\ref{fig:triang-flip-Nak}\,a) by placing indices $s_e \in \mathds{Z}_r$ on each edge $e$, subject to a condition at each vertex, see
 \cite{Runkel:2018feb,Runkel:2022fzi}. The indices $s_e$ control the powers of $N_\mathcal{A}$ inserted on the dual $A$-network. A related construction of 2d\,QFTs on 2-spin surfaces is to couple the QFT to the 
2d Arf spin\,TFT, see \cite{Karch:2019lnn,Hsieh:2020uwb}.
 
\medskip

Let $\mathcal{F}$ be a pivotal fusion category, and $Q$ a 2d\,QFT with symmetry $\mathcal{F}$ and realisation
$R\colon \mathcal{F} \to \mathcal{E}_Q$. 
If we know a $\Delta$-separable Frobenius algebra $\mathcal{A}$ in $\mathcal{F}$ with $N_\mathcal{A}$ of order $r>0$, then $R(\mathcal{A})$ is a gaugeable topological symmetry for~$Q$.
The theory $Q/R(\mathcal{A})$ is well-defined on $r$-spin surfaces, with $r=1$ being the oriented case, $r=2$ the usual spin case, and the infinite-order case requiring framed surfaces.

\begin{examples}
	\begin{enumerate}
\item In the four pivotal fusion categories \eqref{eq:Z2-pivotal-cats} consider the object $A=\mathds{C}_0 \oplus \mathds{C}_1$, the sum of the two simple objects. $A$ carries a unique-up-to-isomorphism $\Delta$-separable Frobenius algebra structure $\mathcal{A}$ in the cases 
$\mathrm{Vect}_{\mathds{Z}_2}$ and $\mathrm{SVect}_{\mathds{Z}_2}$, and so the $\mathds{Z}_2$-symmetry can be gauged. 
For $\mathrm{Vect}_{\mathds{Z}_2}$, $N_\mathcal{A} =\id_A$ and the gauged theory $Q/R(\mathcal{A})$ is again defined on oriented surfaces, while for $\mathrm{SVect}_{\mathds{Z}_2}$, $N_\mathcal{A}$ has order 2 and the gauged theory is defined on spin surfaces. 
In the remaining two cases, $A$ does not allow for a $\Delta$-separable Frobenius algebra structure -- the $\mathds{Z}_2$-symmetry is anomalous and cannot be gauged. 

Note that if one only considers fusion categories and ignores the pivotal structure, one can still conclude which $\mathds{Z}_2$ can be gauged and which cannot, but one cannot tell what tangential structure is needed to make the gauged theory well-defined.

\item For Morita equivalent algebras~$\mathcal{A}$ as above, 
the gauged theories are equivalent \cite{Frohlich:2006ch} (up to invertible 2d\,TFTs), and so it is enough to find representatives of different Morita classes. 
The latter correspond to module categories \cite{Ostrik:2001} (to avoid the invertible 2d\,TFTs one has to work with module traces \cite{Schm} which also capture the Frobenius structure of $\mathcal{A}$). 
In the case of $\mathcal{C}(\mathfrak{su}(2),k)$ as in Example~\ref{ex:su2-k-module-cats}, possible choices for algebras are 	
for $\mathcal{M}_{\mathrm{A}_{k+1}}$: $\mathcal{A}_{\mathrm{A}_{k+1}} = \underline 0$,
for $\mathcal{M}_{\mathrm{D}_{k/2+2}}$: 
$\mathcal{A}_{\mathrm{D}_{k/2+2}} = \underline 0 \oplus \underline k$,
and for $\mathcal{M}_{\mathrm{E}_6}$: $\mathcal{A}_{\mathrm{E}_6} = \underline 0 \oplus \underline 6$, $\mathcal{M}_{\mathrm{E}_7}$: $\mathcal{A}_{\mathrm{E}_7} = \underline 0 \oplus \underline 8 \oplus \underline{16}$, 
$\mathcal{M}_{\mathrm{E}_8}$: $\mathcal{A}_{\mathrm{E}_8} = \underline 0 \oplus \underline{10} \oplus \underline{18} \oplus \underline{28}$.
In all cases, the Nakayama automorphism is the identity. 
$\mathcal{A}_{\mathrm{E}_{6,7,8}}$ involve non-invertible topological defects as simple summands, so these are examples of gauging a non-invertible topological symmetry.
\end{enumerate}
\end{examples}

\subsection{2d\,CFT via surface defects in 3d\,TFT}\label{sec:2dCFT-3dTFT}

A rational full 2d\,CFT can be described as the boundary theory of a 3d\,TFT $\zz$ of Reshetikhin--Turaev type. If the holomorphic fields of the 2d\,CFT contain a rational vertex operator algebra $\mathcal{V}$, the modular fusion category defining the 3d\,TFT $\zz$ is $\mathcal{C} = \mathrm{Rep}(\mathcal{V})$, the category of $\mathcal{V}$-modules. 
The state spaces of the TFT are identified with the spaces of conformal blocks for $\mathcal{V}$. More details can be found in \cite{Fuchs:2023ngi}.

\begin{figure}[bt]
\centering

\raisebox{5em}{a)}	
\hspace{-1em}
\raisebox{-0.45\height}{\scalebox{.4}{\includegraphics{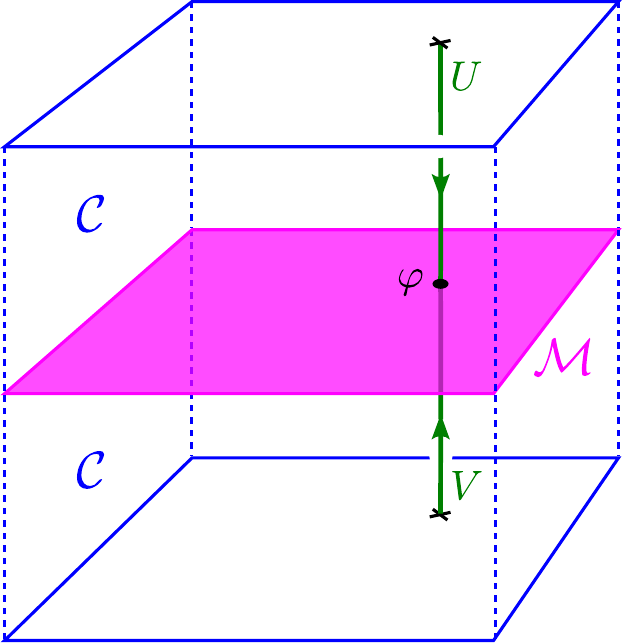}}}
\hspace{1.2em}
\raisebox{5em}{b)}	
\hspace{-1em}
\raisebox{-0.45\height}{\scalebox{.4}{\includegraphics{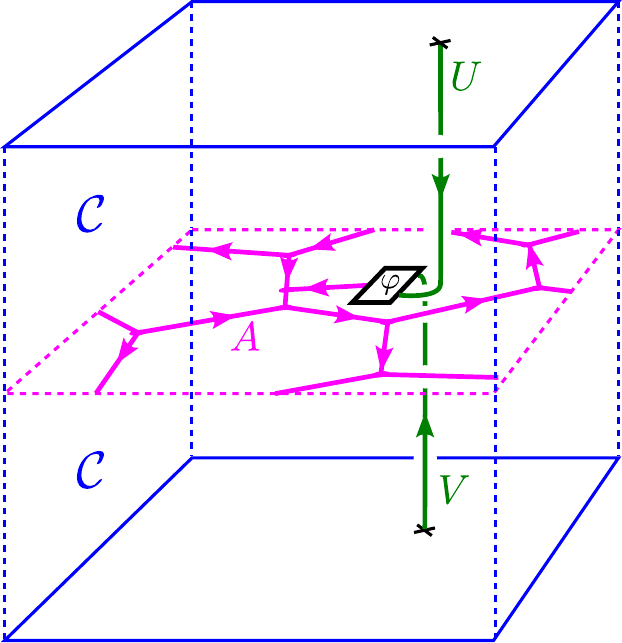}}}
\hspace{1.2em}
\raisebox{5em}{c)}	
\hspace{-1em}
\raisebox{-0.1\height}{\scalebox{.4}{\includegraphics{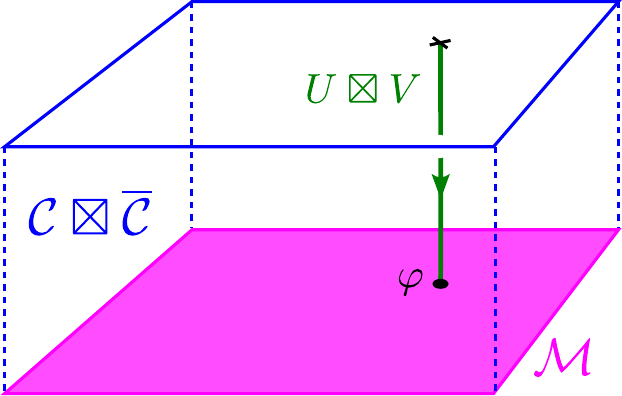}}}

	\caption{a) A surface defect $\mathcal{M}$ in the Reshetikhin--Turaev TFT for the modular fusion category $\mathcal{C}$, and line defects $U,V$ starting at the holomorphic and antiholomorphic boundary and meeting at the topological junction $\varphi$. b) Description of the surface defect via the gaugeable topological symmetry $\mathcal{A} = (A,\mu,\Delta)$.
 c) The folded picture of a) gives a symmetry TFT description of 2d\,CFT correlators.}
	\label{fig:2dCFT-3dTFT}
	
\end{figure}

The TFT $\zz$ is defined on bordisms with line defects labelled by objects of $\mathcal{C}$, and surface defects labelled by $\Delta$-separable symmetric Frobenius algebras, or, in a Morita invariant way by $\mathcal{C}$-module categories with module trace, see~\cite{Kapustin:2010if,Fuchs:2012dt,Schm,Carqueville:2017ono,Koppen:2021kry,Carqueville:2023aak}. 
More precisely, these label surface defects which allow for an interface with the trivial defect with invertible bubble, and which therefore are obtained by gauging line defects analogously to the procedure in Section~\ref{sec:gauging}, but now carried out on the surface defect \cite{Kapustin:2010if,Fuchs:2012dt}, see Figure~\ref{fig:2dCFT-3dTFT}\,b) and Section~\ref{subsec:defect3categories} below. 
2d\,CFTs are classified by surface defects in the following sense \cite{Fjelstad:2006aw,Kong:2009inh}:

\begin{theorem}\label{thm:class-2dCFT}
2d\,CFTs with unique vacuum, non-degenerate two-point function and whose bulk fields contain a holomorphic and antiholomorphic copy of $\mathcal{V}$ are in one-to-one correspondence to indecomposable semisimple $\mathcal{C}$-module categories with module trace.
\end{theorem}

As an example, in this approach the multiplicity of bulk fields of the CFT which transform in the irreducible representations $U,V$ of the holomorphic and antiholomorphic copy of $\mathcal{V}$, respectively, is given by the space of topological junctions between the line defects $U$ and $V$ and the surface defect $\mathcal{M}$, see Figure~\ref{fig:2dCFT-3dTFT}\,a). 
In the description by gauging $\mathcal{A}$ (a $\Delta$-separable symmetric Frobenius algebra in~$\mathcal C$), these junctions are the image of the idempotent $P$ on $\Hom_{\mathcal{C}}(U \otimes V,A)$ given by
\begin{align}
P(\varphi) ~=~
\raisebox{-0.45\height}{\scalebox{.45}{\includegraphics{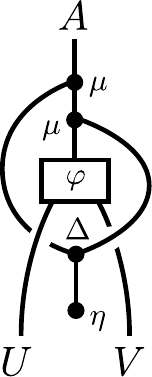}}}
\quad , \quad
\varphi \in \Hom_{\mathcal{C}}(U \otimes V,A)
\ .
\end{align}
As another example, line defects embedded in the surface defect $\mathcal{M}$ describe topological defects of the 2d\,CFT which are transparent to fields in the holomorphic and antiholomorphic copy of $\mathcal{V}$. In the gauged description, these are given by $\mathcal{A}$-$\mathcal{A}$-bimodules in $\mathcal{C}$, see \cite{Fuchs:2002cm,Fuchs:2023ngi}.

\begin{example}\label{ex:surface-defect-interfaces}
\begin{enumerate}
    \item 
If $\mathcal{M}$ is the trivial surface defect $\mathcal{I}$, the only allowed bulk fields are charge conjugate pairs $(U,U^*)$ with multiplicity one. Topological line defects are in this case just the line defects of the 3d\,TFT, i.e.\ correspond to objects of $\mathcal{C}$. 
\item 
Let $\mathcal{C}$ be the Ising modular category, which has simple objects $\one$, $\sigma$, $\varepsilon$, with $\sigma \otimes \sigma \cong \one \oplus \varepsilon$, $\varepsilon \otimes \varepsilon \cong \one$ as in Example~\ref{ex:DQ-examples}\,(2).
The object $A = \one \oplus \varepsilon$ can be equipped with the structure $\mathcal{A}$ of a $\Delta$-separable symmetric Frobenius algebra 
and defines a surface defect $\mathcal{M}$. Thinking of $A$ as a left $\mathcal{A}$-module ${}_\mathcal{A}A$ turns it into an interface between $\mathcal{M}$ and the trivial surface defect $\mathcal{I}$. 
One can also turn $\sigma$ into a left $\mathcal{A}$-module ${}_\mathcal{A}\sigma$, and so obtains another interface between $\mathcal{M}$ and $\mathcal{I}$. One checks that ${}_\mathcal{A}\sigma$ is invertible (i.e.\ $({}_\mathcal{A}\sigma)^* \otimes_\mathcal{A} {}_\mathcal{A}\sigma \cong \one$ and ${}_\mathcal{A}\sigma \otimes ({}_\mathcal{A}\sigma)^* \cong {}_\mathcal{A} A_\mathcal{A}$), while ${}_\mathcal{A}A$ is not (see e.g.~\cite[Sect.\,2.1]{Frohlich:2006ch} for more on the relative tensor product 
$\otimes_\mathcal{A}$).
Note that $({}_\mathcal{A}A)^* \otimes_\mathcal{A} {}_\mathcal{A}\sigma \cong \sigma$.

The module ${}_\mathcal{A}\sigma$ serves as a Morita context between $\one$ and $\mathcal{A}$. 
In terms of defects, the line defect ${}_\mathcal{A}\sigma$ provides an equivalence between the surface defects $\mathcal{I}$ and $\mathcal{M}$ (see also Section~\ref{subsec:defect2categories}). In fact, the Ising modular category has only one indecomposable semisimple module category, namely itself, and hence only one Morita class of simple algebras. Correspondingly, $\mathcal{I}$ is up-to-equivalence the only surface defect.

The lines labelled $\one,\varepsilon,A$ are the counterparts of the lattice defects $I,S,A$ from \eqref{eq:Ising-defect-summary}, while ${}_\mathcal{A}A$ and  ${}_\mathcal{A}\sigma$ correspond to the interfaces $G$ and $D$.
\end{enumerate}
\end{example}

\medskip

Rather than using the TFT for $\mathcal{C}$ with two boundaries with conformal structure and a topological defect separating them as in Figure~\ref{fig:2dCFT-3dTFT}\,a), one can use a folded picture with the TFT for $\mathcal{C} \boxtimes \overline{\mathcal{C}}$ 
(where $\overline{\mathcal{C}}$ has reversed braiding and ribbon twist) as in Figure~\ref{fig:2dCFT-3dTFT}\,c). 
Now only one boundary is non-topological and carries both the holomorphic and antiholomorphic conformal blocks, while the other boundary is topological and its boundary conditions classify the different full 2d\,CFTs as in Theorem~\ref{thm:class-2dCFT}.
This is an example of the symmetry TFT picture discussed in Section~\ref{sec:symTFT}.

\section{Topological quantum field theory}
\label{sec:TFT}

This section is about defects in topological quantum field theories as well as the higher categories they give rise to. 
Section~\ref{subsec:DefectTFTs} discusses the functorial approach to (extended and non-extended) defect TFTs, while Sections~\ref{subsec:defect2categories} and~\ref{subsec:defect3categories} cover weak invertibility, condensation, and the gauging of topological symmetries in the 2- and 3-dimensional case, respectively. 
For background on 2- and 3-categories we refer to \cite{JohnsonYauBook, BMS}.

\subsection{TFTs with defects}
\label{subsec:DefectTFTs}

A \textsl{topological quantum field theory with defects} (or \textsl{defect TFT}) is a (possibly higher) symmetric monoidal functor~$\zz$ whose domain is a category of stratified and labelled bordisms. 
The stratification provides a decomposition of bordisms into finitely many connected submanifolds (or \textsl{strata}) of arbitrary codimension, and a label is assigned to every stratum which we think of as ``the physical degrees of freedom'' of the defect supported on that stratum, or as local interactions between 
``fields'' associated to adjacent higher-dimensional strata. 
If a stratum~$\sigma$ ends on the boundary of a bordism, then it does so transversally and~$\partial \sigma$ carries the same label as~$\sigma$.  
By definition, evaluation with~$\zz$ only depends on the isotopy class of the \textsl{defects}, i.e.\ the labelled strata. 
This is illustrated by the following 2-dimensional example: 
\begin{equation}
\zz \left(
	\begin{tikzpicture}[thick,scale=2.83,color=black, >=stealth, baseline=-1.14cm]
	\coordinate (p1) at (-0.55,0);
	\coordinate (p2) at (-0.2,0);
	\coordinate (p3) at (0.2,0);
	\coordinate (p4) at (0.55,0);
	\coordinate (p5) at (0.175,-0.8);
	\coordinate (p6) at (-0.175,-0.8);
	\coordinate (s) at (0.175,0.07);
	\coordinate (d) at (0,-0.5);
	%
	%
	\fill [orange!20!white, opacity=0.8] 
	(p1) .. controls +(0,0.1) and +(0,0.1) ..  (p2)
	-- (p2) .. controls +(0,-0.35) and +(0,-0.35) ..  (p3)
	-- (p3) .. controls +(0,0.1) and +(0,0.1) ..  (p4)
	-- (p4) .. controls +(0,-0.5) and +(0,0.5) ..  (p5)
	-- (p5) .. controls +(0,0.1) and +(0,0.1) ..  (p6)
	-- (p6) .. controls +(0,0.5) and +(0,-0.5) ..  (p1)
	;
	%
	\fill [orange!30!white, opacity=0.8] 
	(p1) .. controls +(0,-0.1) and +(0,-0.1) ..  (p2)
	-- (p2) .. controls +(0,-0.35) and +(0,-0.35) ..  (p3)
	-- (p3) .. controls +(0,-0.1) and +(0,-0.1) ..  (p4)
	-- (p4) .. controls +(0,-0.5) and +(0,0.5) ..  (p5)
	-- (p5) .. controls +(0,-0.1) and +(0,-0.1) ..  (p6)
	-- (p6) .. controls +(0,0.5) and +(0,-0.5) ..  (p1)
	;
	%
	\draw[brown!50!red, very thick, postaction={decorate}, decoration={markings,mark=at position .5 with {\arrow[draw=brown!50!red]{>}}}] (-0.52,-0.2).. controls +(0.01,-0.03) and +(-0.01,-0.06) .. (-0.15,-0.18);
	\draw[brown!50!red, very thick, opacity=0.3] (-0.52,-0.2).. controls +(0.01,0.07) and +(-0.01,0.04) .. (-0.155,-0.18);
	\fill[brown!50!red] (-0.32,-0.27) circle (0pt) node {\scalebox{.4}{$X_6$}};
	\draw[blue!50!red, very thick, postaction={decorate}, decoration={markings,mark=at position .6 with {\arrow[draw=blue!50!red]{>}}}] (0,-0.45) .. controls +(0,0.2) and +(0,-0.3) .. (0.35,-0.073); 
	\fill[blue!50!red] (0.27,-0.32) circle (0pt) node {\scalebox{.4}{$X_2$}};
	\draw[brown!50!black, very thick, postaction={decorate}, decoration={markings,mark=at position .5 with {\arrow[draw=brown!50!black]{>}}}] (0.13,-0.75) .. controls +(0,0.2) and +(0.1,0) .. (0,-0.45);
	\fill[brown!50!black] (0.16,-0.55) circle (0pt) node {\scalebox{.4}{$X_4$}};
	\draw[brown!50!blue, very thick, postaction={decorate}, decoration={markings,mark=at position .5 with {\arrow[draw=brown!50!blue]{>}}}] (-0.13,-0.75) .. controls +(0,0.2) and +(-0.1,0) .. (0,-0.45);
	\fill[brown!50!blue] (-0.16,-0.55) circle (0pt) node {\scalebox{.4}{$X_5$}};
	\fill[black] (0,-0.45) circle (0.55pt) node[above] {};
	\fill[black] (0,-0.44) circle (0pt) node[below] {{\tiny$\varphi$}};
	%
	%
	\draw[thick] (p2) .. controls +(0,-0.35) and +(0,-0.35) ..  (p3); 
	\draw[thick] (p4) .. controls +(0,-0.5) and +(0,0.5) ..  (p5); 
	\draw[thick] (p6) .. controls +(0,0.5) and +(0,-0.5) ..  (p1); 
	\draw[very thick, red!80!black] (p1) .. controls +(0,-0.1) and +(0,-0.1) ..  (p2); 
	\draw[very thick, red!80!black] (p3) .. controls +(0,-0.1) and +(0,-0.1) ..  (p4) ; 
	\draw[very thick, red!80!black] (p5) .. controls +(0,-0.1) and +(0,-0.1) ..  (p6); 
	\draw[very thick, red!80!black] (p1) .. controls +(0,0.1) and +(0,0.1) ..  (p2); 
	\draw[very thick, red!80!black] (p3) .. controls +(0,0.1) and +(0,0.1) ..  (p4); 
	\draw[very thick, red!80!black] (p5) .. controls +(0,0.1) and +(0,0.1) ..  (p6); 
	\fill[red!80!black] (-0.47,-0.12) circle (0pt) node {\scalebox{.4}{$u_1$}};
	\fill[red!80!black] (-0.2,-0.4) circle (0pt) node {\scalebox{.4}{$u_2$}};
	\fill[red!80!black] (0,-0.6) circle (0pt) node {\scalebox{.4}{$u_3$}};
	\fill[blue!50!red] (0.35,-0.076) circle (0.55pt) node {};
	\fill[brown!50!blue] (-0.13,-0.75) circle (0.55pt) node {};
	\fill[brown!50!black] (0.13,-0.75) circle (0.55pt) node {};
	\end{tikzpicture}
\right)
= 
\zz \left(
\begin{tikzpicture}[thick,scale=2.83,color=black, >=stealth, baseline=-1.14cm]
\coordinate (p1) at (-0.55,0);
\coordinate (p2) at (-0.2,0);
\coordinate (p3) at (0.2,0);
\coordinate (p4) at (0.55,0);
\coordinate (p5) at (0.175,-0.8);
\coordinate (p6) at (-0.175,-0.8);
\coordinate (s) at (0.175,0.07);
\coordinate (d) at (0,-0.5);
%
%
\fill [orange!20!white, opacity=0.8] 
(p1) .. controls +(0,0.1) and +(0,0.1) ..  (p2)
-- (p2) .. controls +(0,-0.35) and +(0,-0.35) ..  (p3)
-- (p3) .. controls +(0,0.1) and +(0,0.1) ..  (p4)
-- (p4) .. controls +(0,-0.5) and +(0,0.5) ..  (p5)
-- (p5) .. controls +(0,0.1) and +(0,0.1) ..  (p6)
-- (p6) .. controls +(0,0.5) and +(0,-0.5) ..  (p1)
;
%
\fill [orange!30!white, opacity=0.8] 
(p1) .. controls +(0,-0.1) and +(0,-0.1) ..  (p2)
-- (p2) .. controls +(0,-0.35) and +(0,-0.35) ..  (p3)
-- (p3) .. controls +(0,-0.1) and +(0,-0.1) ..  (p4)
-- (p4) .. controls +(0,-0.5) and +(0,0.5) ..  (p5)
-- (p5) .. controls +(0,-0.1) and +(0,-0.1) ..  (p6)
-- (p6) .. controls +(0,0.5) and +(0,-0.5) ..  (p1)
;
%
\draw[brown!50!red, very thick, postaction={decorate}, decoration={markings,mark=at position .5 with {\arrow[draw=brown!50!red]{>}}}] (-0.52,-0.2).. controls +(0.01,-0.1) and +(-0.01,-0.3) .. (-0.15,-0.18);
\draw[brown!50!red, very thick, opacity=0.3] (-0.52,-0.2).. controls +(0.01,0.07) and +(-0.01,0.04) .. (-0.155,-0.18);
\fill[brown!50!red] (-0.32,-0.29) circle (0pt) node {\scalebox{.4}{$X_6$}};
\draw[blue!50!red, very thick, postaction={decorate}, decoration={markings,mark=at position .8 with {\arrow[draw=blue!50!red]{>}}}] (0,-0.35) .. controls +(0,0.2) and +(0,-0.5) .. (0.35,-0.073); 
\fill[blue!50!red] (0.28,-0.2) circle (0pt) node {\scalebox{.4}{$X_2$}};
\draw[brown!50!black, very thick, postaction={decorate}, decoration={markings,mark=at position .5 with {\arrow[draw=brown!50!black]{>}}}] (0.13,-0.75) .. controls +(0,0.2) and +(0.1,0) .. (0,-0.35);
\fill[brown!50!black] (0.16,-0.53) circle (0pt) node {\scalebox{.4}{$X_4$}};
\draw[brown!50!blue, very thick, postaction={decorate}, decoration={markings,mark=at position .5 with {\arrow[draw=brown!50!blue]{>}}}] (-0.13,-0.75) .. controls +(0.2,0.2) and +(-0.4,0) .. (0,-0.35);
\fill[brown!50!blue] (-0.17,-0.53) circle (0pt) node {\scalebox{.4}{$X_5$}};
\fill[black] (0,-0.35) circle (0.55pt) node[above] {};
\fill[black] (0,-0.34) circle (0pt) node[below] {{\tiny$\varphi$}};
%
%
\draw[thick] (p2) .. controls +(0,-0.35) and +(0,-0.35) ..  (p3); 
\draw[thick] (p4) .. controls +(0,-0.5) and +(0,0.5) ..  (p5); 
\draw[thick] (p6) .. controls +(0,0.5) and +(0,-0.5) ..  (p1); 
\draw[very thick, red!80!black] (p1) .. controls +(0,-0.1) and +(0,-0.1) ..  (p2); 
\draw[very thick, red!80!black] (p3) .. controls +(0,-0.1) and +(0,-0.1) ..  (p4) ; 
\draw[very thick, red!80!black] (p5) .. controls +(0,-0.1) and +(0,-0.1) ..  (p6); 
\draw[very thick, red!80!black] (p1) .. controls +(0,0.1) and +(0,0.1) ..  (p2); 
\draw[very thick, red!80!black] (p3) .. controls +(0,0.1) and +(0,0.1) ..  (p4); 
\draw[very thick, red!80!black] (p5) .. controls +(0,0.1) and +(0,0.1) ..  (p6); 
\fill[red!80!black] (-0.47,-0.12) circle (0pt) node {\scalebox{.4}{$u_1$}};
\fill[red!80!black] (-0.25,-0.42) circle (0pt) node {\scalebox{.4}{$u_2$}};
\fill[red!80!black] (0,-0.6) circle (0pt) node {\scalebox{.4}{$u_3$}};
\fill[blue!50!red] (0.35,-0.076) circle (0.55pt) node {};
\fill[brown!50!blue] (-0.13,-0.75) circle (0.55pt) node {};
\fill[brown!50!black] (0.13,-0.75) circle (0.55pt) node {};
\draw[brown!50!blue, very thick, postaction={decorate}, decoration={markings,mark=at position .5 with {\arrow[draw=brown!50!blue]{>}}}] (-0.13,-0.75) .. controls +(0.2,0.2) and +(-0.4,0) .. (0,-0.35);
\end{tikzpicture}
\right)
. 
\end{equation}

An $n$-dimensional defect TFT~$\zz$ may or may not be extended. 
In the latter case, the domain of~$\zz$ is a symmetric monoidal 1-category $\Borddef$ whose objects and morphisms are stratified closed $(n-1)$-manifolds and $n$-bordism classes whose strata are labelled by elements of a given set of defect data~$\mathds{D}$, respectively. 
The target of~$\zz$ is usually taken to be the category of (super) vector spaces. 
The detailed definition of such a defect TFT $\zz\colon \Borddef \to \Vect$ where all strata come with orientations is given in \cite[Sect.\,2]{Carqueville:2017aoe}. 
We will return to such non-extended TFTs shortly. 

On the other hand, fully extended defect TFTs are the subject of \cite[Sect.\,4.3]{l0905.0465}, where a classification of such symmetric monoidal $n$-functors $\Bord_n^{\vec{X}} \to \mathcal D$ is proposed in terms of the ``cobordism hypothesis with singularities''. 
Here~$\vec{X}$ is a very general choice of tangential structure (like framings, orientations, spin, or principal $G$-bundles) for defects, and the target $n$-category~$\mathcal D$ is another choice (which is very different, say, for state sum models and twisted sigma models). 

Some aspects of the cobordism hypothesis are still under development, and it has only been applied systematically in low dimension and for specific tangential structures. 
More details for the special case of defects that are boundary conditions are provided in \cite[App.\,A]{FreedTeleman2020}. 
2-dimensional oriented extended defect TFTs~$\zz^{\textrm{ext}}$ are described in \cite[Sect.\,3]{BCFR}, and in this case the cobordism hypothesis amounts
to an interpretation of defect bordisms in terms of the 3-dimensional graphical calculus of its target 2-category~$\mathcal D$ as developed in \cite{BMS}. 
For example 
\begin{align}
& \zz^{\textrm{ext}}
\left(
\begin{tikzpicture}[thick,scale=2,color=black, baseline=-1.9cm, yscale=-1]
\coordinate (p1) at (0,0);
\coordinate (p2) at (0.5,0.2);
\coordinate (p3) at (-0.5,0.4);
\coordinate (p4) at (2.5,0);
\coordinate (p4new) at (3,0);
\coordinate (p5) at (1.5,0.2);
\coordinate (p6) at (2.0,0.4);
\coordinate (p6new) at (2.0,0.4);
\coordinate (s) at ($(0,1.5)$);
\coordinate (d) at ($(3.8,0)$);
\coordinate (h) at ($(0,2.5)$);
\coordinate (t) at ($(0.05,0)$);
\coordinate (u) at ($(0,-1.1)$);
\coordinate (X) at ($(p6)+(s)+(-1,0)$);
\coordinate (Xd) at ($(p6)+(s)+(-0.25,0)$);
%
\fill[red!80!black] ($(p1)+(t)$) circle (0pt) node[left] {{\tiny $(v,+)$}};
\fill[red!80!black] ($(p3)+(t)$) circle (0pt) node[left] {{\tiny $(v,-)$}};
\fill[red!80!black] ($(p4new)-0.5*(t)$) circle (0pt) node[right] {{\tiny $(u,+)$}};
\fill[red!80!black] ($(p4new)+(s)-0.5*(t)$) circle (0pt) node[right] {{\tiny $(u,+)$}};
\fill[red!80!black] ($(p1)+(s)+(t)$) circle (0pt) node[left] {{\tiny $(v,+)$}};
\fill[red!80!black] ($(p3)+(s)+(t)$) circle (0pt) node[left] {{\tiny $(v,-)$}};
\fill[red!80!black] ($(p6new)+(s)-0.5*(t)$) circle (0pt) node[right] {{\tiny $(v,-)$}};
%
\fill [orange!20!white, opacity=0.8] 
(p1) .. controls +(0.3,0) and +(0,-0.1) ..  (p2)
-- (p2) .. controls +(0,0.9) and +(0,0.9) ..  (p5)
-- (p5) .. controls +(0,-0.1) and +(-0.5,0) ..  (p4) -- (p4)
-- (p4new) -- ($(p4new)+(s)$) -- ($(p1)+(s)$);
\fill[orange!45!white, opacity=0.7] ($(p4)+(s)+(-0.2,0)$) .. controls +(0,-1) and +(0,1) .. ($(p4)+(0.3,0)$) -- (p4new) -- ($(p4new)+(s)$);
%
\draw[very thick, red!80!black] ($(p4new)+(s)$) -- ($(p1)+(s)$);
\draw[thin] (p1) -- ($(p1)+(s)$);
%
\fill [orange!30!white, opacity=0.8] 
(p3) .. controls +(0.5,0) and +(0,0.1) ..  (p2)
-- (p2) .. controls +(0,0.9) and +(0,0.9) ..  (p5)
-- (p5) .. controls +(0,0.1) and +(-0.3,0) .. (p6)
-- (p6) -- ($(p6)+(s)$) -- ($(p3)+(s)$);
%
\draw[thin] (p6new) -- ($(p6new)+(s)$);
\draw (p2) .. controls +(0,0.9) and +(0,0.9) ..  (p5);
\draw[thin] (p3) -- ($(p3)+(s)$);
\draw[thin] (p4new) -- ($(p4new)+(s)$);
%
\draw[very thick, blue!50!black, postaction={decorate}, decoration={markings,mark=at position .59 with {\arrow[draw=blue!50!black]{>}}}] ($(p4)+(s)+(-0.2,0)$) .. controls +(0,-1) and +(0,1) .. ($(p4)+(0.3,0)$);
\fill[blue!50!black] ($(p6)+(0.75,0.3)$) circle (0pt) node {{\small$Y$}};
\fill[orange!40!white, opacity=0.7] (0.2,1) circle (0.3);
\draw[very thick, blue!50!black] (0.2,1) circle (0.3);
\fill[blue!50!black] (-0.23,0.95) circle (0pt) node {{\small$Z$}};
\draw[<-, very thick, blue!50!black] (-0.1,1.001) -- (-0.1,0.999) node[above] {}; 
\fill[green!50!black] (0.5,1) circle (1.25pt) node[right] {{\small$\chi$}};
\fill[orange!10!white, opacity=0.7] (X) .. controls +(0,-0.75) and +(0,-0.75) .. (Xd);
\draw[redirected, very thick, blue!50!black] (X) .. controls +(0,-0.75) and +(0,-0.75) .. (Xd);
\fill[blue!50!black] ($(p6)+(s)+(-0.6,-0.72)$) circle (0pt) node {{\small$X$}};
\fill[green!50!black] (0.25,1.7) circle (1.25pt) node[right] {{\small$\xi$}};
%
\draw[very thick, red!80!black] (p1) .. controls +(0.3,0) and +(0,-0.1) ..  (p2); 
\draw[very thick, red!80!black] (p2) .. controls +(0,0.1) and +(0.5,0) ..  (p3); 
\draw[very thick, red!80!black] (p6) .. controls +(-0.3,0) and +(0,0.1) ..  (p5); 
\draw[very thick, red!80!black] (p1) .. controls +(0.3,0) and +(0,-0.1) ..  (p2); 
\draw[very thick, red!80!black] (p5) .. controls +(0,-0.1) and +(-0.5,0) ..  (p4); 
\draw[very thick, red!80!black] ($(p6new)+(s)$) -- ($(p3)+(s)$);
\draw[very thick, red!80!black] (p6) -- (p6new);
\draw[very thick, red!80!black] (p4) -- (p4new);
\fill[blue!50!black] (X) circle (1.25pt) node[below] {{\scriptsize $(X,-)$}};
\fill[blue!50!black] (Xd) circle (1.25pt) node[below] {{\scriptsize $(X,+)$}};
\fill[blue!50!black] ($(p4)+(s)+(-0.2,0)$) circle (1.25pt) node[below] {{\scriptsize$(Y,+)$}};
\fill[blue!50!black] ($(p4)+(0.3,0)$) circle (1.25pt) node[above] {{\scriptsize$(Y,+)$}};
%
%
\fill[red!80!black] (0.2,1) circle (0pt) node {{\scriptsize $w$}};
\fill[red!80!black] ($(p6)+(s)+(-0.6,-0.33)$) circle (0pt) node {{\scriptsize $v'$}};
\fill[red!80!black] ($(p6)+(0.6,0.7)$) circle (0pt) node[right] {{\scriptsize $u$}};
\fill[red!80!black] ($(p6)+(0.1,0.3)$) circle (0pt) node[right] {{\scriptsize $v$}};
\fill[red!80!black] ($(p6)+(-0.5,0.5)$) circle (0pt) node[right] {{\scriptsize $v$}};
\fill[red!80!black] ($(p6)+(0.1,-0.33)$) circle (0pt) node[right] {{\tiny $v$}};
\fill[red!80!black] ($(p6)+(-1.9,-0.32)$) circle (0pt) node[right] {{\tiny $v$}};
\fill[red!80!black] ($(p6)+(0.05,1.05)$) circle (0pt) node[right] {{\tiny $v$}};
\fill[red!80!black] ($(p6)+(-0.2,1.45)$) circle (0pt) node[right] {{\tiny $v$}};
\fill[red!80!black] ($(p6)+(-1.5,1.45)$) circle (0pt) node[right] {{\tiny $v$}};
\fill[red!80!black] ($(p6)+(-0.75,1.43)$) circle (0pt) node[right] {{\tiny $v'$}};
\fill[red!80!black] ($(p6)+(0.6,1.05)$) circle (0pt) node[right] {{\tiny $u$}};
\fill[red!80!black] ($(p6)+(0.8,-0.34)$) circle (0pt) node[right] {{\tiny $u$}};
\fill[red!80!black] ($(p6new)-0.5*(t)$) circle (0pt) node[right] {{\tiny $(v,-)$}};
\end{tikzpicture}
\right)
\qquad 
	\begin{tikzpicture}[thick,scale=1.0,color=blue!50!black, baseline=0.0cm, 
	style={x={(-0.4cm,-0.3cm)},y={(0.8cm,0cm)},z={(0cm,0.9cm)}}]
	\draw[
	color=gray, 
	opacity=0.3, 
	semithick,
	dashed
	] 
	(1,0,0) -- (0,0,0) -- (0,1,0)
	(0,0,0) -- (0,0,1);
	\fill [blue!40,opacity=0.1] (0,0,0) -- (1,0,0) -- (1,1,0) -- (0,1,0) -- (0,1,1) -- (0,0,1) -- (1,0,1) -- (1,0,0);
	%
	%
	\draw[
	color=gray, 
	opacity=0.4, 
	semithick
	] 
	(0,1,1) -- (0,1,0) -- (1,1,0) -- (1,1,1) -- (0,1,1) -- (0,0,1) -- (1,0,1) -- (1,0,0) -- (1,1,0)
	(1,0,1) -- (1,1,1);
	%
	\draw[
	color=gray, 
	opacity=0.9, 
	thick, 
	->
	] 
	(1, 1.15, 0) -- node[pos=0.5, color=black, below] {\footnotesize $\btimes$} (0, 1.15, 0); 
	\draw[
	color=gray, 
	opacity=0.9, 
	thick, 
	->
	] 
	(1.15, 1, 0) -- node[pos=0.5, color=black, below, sloped] {\footnotesize $\otimes$} (1.15, 0, 0); 
	\draw[
	color=gray, 
	opacity=0.9, 
	thick, 
	->
	] 
	(1, -0.15, 0) -- node[pos=0.5, color=black, above, sloped] {\footnotesize $\circ$} (1, -0.15, 1); 
	\end{tikzpicture}
\nonumber
\\
& \quad 
= \Big[ 1_{1_{\zz^{\textrm{ext}}(v)}}  \btimes \Big(  \textrm{tr}_{\textrm{l}}\big( \zz^{\textrm{ext}}(\chi) \big) \circ {\zz^{\textrm{ext}}(\xi)} \Big) \Big] 
\otimes \widetilde{\textrm{coev}}_{\widetilde{\textrm{ev}}_{\zz^{\textrm{ext}}(v)}} \otimes \Big[ 1_{\zz^{\textrm{ext}}(Y)} \btimes \textrm{ev}_{\zz^{\textrm{ext}}(X)} \Big]
\label{eq:cobordism-is-graphical-calc}
\end{align}
where the cube indicates our conventions for the compositions $\btimes, \otimes, \circ$ in~$\mathcal D$, see Sections~\ref{subsec:defect2categories} and~\ref{subsec:defect3categories} for more details. 
It is expected that more generally the cobordism hypothesis with singularities is the statement that the evaluation of an extended TFT on a defect bordism amounts to interpreting the bordism in the appropriate graphical calculus. 

\medskip 

For fully extended defect TFTs, higher categories feature as their targets~$\mathcal D$ as above, while for non-extended TFTs $\zz\colon \Borddef \to \Vect$ one naturally expects to extract an $n$-category~$\mathcal D_\zz$ which encodes the local action of~$\zz$, i.e.\ how it evaluates on stratified and labelled $n$-balls. 
For $n=2$, this was made rigorous in \cite{Davydov:2011kb} as explained in Section~\ref{sec:2cat-form-def}, while the case $n=3$ is treated in \cite{CMS}. 
Either way, objects of~$\mathcal D_\zz$ (or objects in the image of extended TFTs) are interpreted as $n$-dimensional bulk TFTs, while $k$-morphisms~$\Phi$ supported on some submanifold~$M$ of codimension $k\geqslant 1$ correspond to defects. 
By definition, $\textrm{End}(\Phi)$ is a monoidal $(n-k-1)$-category whose objects label $(n-k-1)$-dimensional strata in~$M$, and $l$-morphisms of $\textrm{End}(\Phi)$ label $(n-k-l-1)$-strata in~$M$, with $l=n-k-1$ corresponding to point defects. 
A special situation is when $\Phi = 1_{{1_{\dots}}_{1_a}}$ is the identity $k$-morphism on the identity $(k-1)$-morphism, etc., on some object~$a$, in which case~$M$ is just embedded in an $a$-labelled $n$-stratum, with no other adjacent higher-dimensional strata. 
Consistent with this geometric setting, $\textrm{End}(1_{{1_{\dots}}_{1_a}})$ comes with a categorification of commutation; 
for example for $n=3$ and $k=2$ the 1-category $\textrm{End}(1_a)$ of ``free-floating'' line defects has a braided structure.

\subsection{Weak inverses and condensations in defect 2-categories}
\label{subsec:defect2categories}

In Example~\ref{ex:DQ-examples}\,(1) we already met the 2-category ssFrob of separable symmetric Frobenius $\mathds{C}$-algebras, bimodules, and bimodule maps. It carries a pivotal structure which describes 2-dimensional \textsl{oriented} state sum models and their defects. 
\textsl{Framed} state sum models on the other hand lead to the 2-category of all (not necessarily symmetric) separable Frobenius algebras \cite{spthesis, PstragowskiJournalVersion}. 

Since separable $\mathds{C}$-algebras are always semisimple, 2-dimensional state sum models are made only from finitely semisimple data. 
While this generalises to state sum models in higher dimensions, there are many 2d defect TFTs which have much richer non-semisimple structure. 
In particular, B-twisted sigma models \cite{cw1007.2679} and Landau--Ginzburg models \cite{Carqueville:2012st} give rise to pivotal 2-categories that are quite far from being finitely
semisimple. 
This is illustrated by the fact that generic line defect categories for Landau--Ginzburg models admit infinite families of indecomposable objects, see e.g.\ \cite{kst0511155}. 
We also note that in all of the above examples the pivotal 2-categories admit a natural symmetric monoidal structure in which every object is fully dualisable (as shown in \cite{spthesis, BanksOnRozanskyWitten, CMM}, respectively), hence they also describe fully extended defect TFTs. 

\medskip 

With examples like the above in mind, we shall next discuss various notions of equivalence between bulk TFTs or their defects. 
Recall first that two objects $X,Y$ in a 1-category are \textsl{isomorphic}, $X\cong Y$, if there exist morphisms $\varphi\colon X\to Y$ and $\psi\colon Y\to X$ such that $\psi\circ\varphi = 1_X$ and $\varphi\circ\psi=1_Y$, or graphically
\begin{equation}
\begin{tikzpicture}[very thick,scale=0.7,color=blue!50!black, baseline=0.8cm]
\coordinate (d1) at (-2,0);
\coordinate (d2) at (0,0);
\coordinate (u1) at (-2,3);
\coordinate (u2) at (0,3);
\coordinate (s) at ($(-0.5,-0.2)$);
\coordinate (b1) at ($(d1)+(s)$);
\coordinate (b2) at ($(d2)+(s)$);
\coordinate (t1) at ($(u1)+(s)$);
\coordinate (t2) at ($(u2)+(s)$);
%
\draw[very thick] ($(d1)+(1,0)+(s)$) -- ($(d1)+(1,3)+(s)$);
\fill ($(d1)+(1,2)+(s)$) circle (3.5pt) node[left] {{\scriptsize $\psi$}};
\fill ($(d1)+(1,1)+(s)$) circle (3.5pt) node[left] {{\scriptsize $\varphi$}};
\fill ($(d1)+(1.65,2.75)+(s)$) circle (0pt) node[left] {{\scriptsize $X$}};
\fill ($(d1)+(1.65,1.5)+(s)$) circle (0pt) node[left] {{\scriptsize $Y$}};
\fill ($(d1)+(1.65,0.25)+(s)$) circle (0pt) node[left] {{\scriptsize $X$}};
\end{tikzpicture}
=
\;
\begin{tikzpicture}[very thick,scale=0.7,color=blue!50!black, baseline=0.8cm]
\coordinate (d1) at (-2,0);
\coordinate (d2) at (0,0);
\coordinate (u1) at (-2,3);
\coordinate (u2) at (0,3);
\coordinate (s) at ($(-0.5,-0.2)$);
\coordinate (b1) at ($(d1)+(s)$);
\coordinate (b2) at ($(d2)+(s)$);
\coordinate (t1) at ($(u1)+(s)$);
\coordinate (t2) at ($(u2)+(s)$);
%
\draw[very thick] ($(d1)+(1,0)+(s)$) -- ($(d1)+(1,3)+(s)$);
\fill ($(d1)+(1.65,0.25)+(s)$) circle (0pt) node[left] {{\scriptsize $X$}};
\end{tikzpicture}
\, , 
\qquad 
\begin{tikzpicture}[very thick,scale=0.7,color=blue!50!black, baseline=0.8cm]
\coordinate (d1) at (-2,0);
\coordinate (d2) at (0,0);
\coordinate (u1) at (-2,3);
\coordinate (u2) at (0,3);
\coordinate (s) at ($(-0.5,-0.2)$);
\coordinate (b1) at ($(d1)+(s)$);
\coordinate (b2) at ($(d2)+(s)$);
\coordinate (t1) at ($(u1)+(s)$);
\coordinate (t2) at ($(u2)+(s)$);
%
\draw[very thick] ($(d1)+(1,0)+(s)$) -- ($(d1)+(1,3)+(s)$);
\fill ($(d1)+(1,2)+(s)$) circle (3.5pt) node[left] {{\scriptsize $\varphi$}};
\fill ($(d1)+(1,1)+(s)$) circle (3.5pt) node[left] {{\scriptsize $\psi$}};
\fill ($(d1)+(1.65,2.75)+(s)$) circle (0pt) node[left] {{\scriptsize $Y$}};
\fill ($(d1)+(1.65,1.5)+(s)$) circle (0pt) node[left] {{\scriptsize $X$}};
\fill ($(d1)+(1.65,0.25)+(s)$) circle (0pt) node[left] {{\scriptsize $Y$}};
\end{tikzpicture}
=
\;
\begin{tikzpicture}[very thick,scale=0.7,color=blue!50!black, baseline=0.8cm]
\coordinate (d1) at (-2,0);
\coordinate (d2) at (0,0);
\coordinate (u1) at (-2,3);
\coordinate (u2) at (0,3);
\coordinate (s) at ($(-0.5,-0.2)$);
\coordinate (b1) at ($(d1)+(s)$);
\coordinate (b2) at ($(d2)+(s)$);
\coordinate (t1) at ($(u1)+(s)$);
\coordinate (t2) at ($(u2)+(s)$);
%
\draw[very thick] ($(d1)+(1,0)+(s)$) -- ($(d1)+(1,3)+(s)$);
\fill ($(d1)+(1.65,0.25)+(s)$) circle (0pt) node[left] {{\scriptsize $Y$}};
\end{tikzpicture}
\, . 
\end{equation}
Note that if one imposes only the first condition,
the composition $e := \varphi \circ \psi \colon Y \to Y$ becomes an \textsl{idempotent}: $e \circ e=e$.

Let now~$\mathcal B$ be a 2-category. 
Two objects $u,v \in \mathcal B$ are \textsl{equivalent}, $u\cong v$, if there exist 1-morphisms $X\colon u \to v$ and $Y\colon v \to u$ as well as 2-morphisms $\varphi\colon Y\otimes X \to 1_u$, $\psi\colon 1_u \to Y\otimes X$, $\zeta \colon X\otimes Y \to 1_v$ and $\xi\colon 1_v\to X\otimes Y$ such that 
\begin{align}
\begin{tikzpicture}[very thick,scale=1.0,color=blue!50!black, baseline=0.9cm]
\coordinate (X) at (0.5,0);
\coordinate (Xd) at (-0.5,0);
\coordinate (X2) at (0.5,2);
\coordinate (Xd2) at (-0.5,2);
\coordinate (d1) at (-1,0);
\coordinate (d2) at (+1,0);
\coordinate (u1) at (-1,2);
\coordinate (u2) at (+1,2);
%
\fill [orange!40!white, opacity=0.7] (d1) -- (d2) -- (u2) -- (u1); 
\fill [orange!60!white, opacity=0.7] (X) .. controls +(0,1) and +(0,1) .. (Xd) -- (X); 
\fill [orange!60!white, opacity=0.7] (X2) .. controls +(0,-1) and +(0,-1) .. (Xd2) -- (X2); 
%
\draw[thin] (d1) -- (d2) -- (u2) -- (u1) -- (d1); 
%
\draw (X) .. controls +(0,1) and +(0,1) .. (Xd);
\draw (X2) .. controls +(0,-1) and +(0,-1) .. (Xd2);
%
\fill ($(X)+(-0.1,0)$) circle (0pt) node[above right] {{\scriptsize $X$}};
\fill ($(Xd)+(0.1,0)$) circle (0pt) node[above left] {{\scriptsize $Y$}};
\fill ($(X2)+(-0.1,0)$) circle (0pt) node[below right] {{\scriptsize $X$}};
\fill ($(Xd2)+(0.1,0)$) circle (0pt) node[below left] {{\scriptsize $Y$}};
\fill[red!80!black] (0,0.2) circle (0pt) node {{\scriptsize $v$}};
\fill[red!80!black] (0,1.8) circle (0pt) node {{\scriptsize $v$}};
\fill[red!80!black] (0.75,1) circle (0pt) node {{\scriptsize $u$}};
\fill (0,0.73) circle (2.5pt) node[below] {{\scriptsize $\varphi$}};
\fill (0,1.26) circle (2.5pt) node[above] {{\scriptsize $\psi$}};
\end{tikzpicture}
& =
\begin{tikzpicture}[very thick,scale=1.0,color=blue!50!black, baseline=0.9cm]
\coordinate (X) at (0.5,0);
\coordinate (Xd) at (-0.5,0);
\coordinate (X2) at (0.5,2);
\coordinate (Xd2) at (-0.5,2);
\coordinate (d1) at (-1,0);
\coordinate (d2) at (+1,0);
\coordinate (u1) at (-1,2);
\coordinate (u2) at (+1,2);
%
\fill [orange!40!white, opacity=0.7] (d1) -- (d2) -- (u2) -- (u1); 
\fill [orange!60!white, opacity=0.7] (X) --(Xd) -- (Xd2) -- (X2); 
%
\draw[thin] (d1) -- (d2) -- (u2) -- (u1) -- (d1); 
%
\draw (X) -- (X2);
\draw (Xd) -- (Xd2);
%
\fill ($(X)+(-0.1,0)$) circle (0pt) node[above right] {{\scriptsize $X$}};
\fill ($(Xd)+(0.1,0)$) circle (0pt) node[above left] {{\scriptsize $Y$}};
\fill[red!80!black] (0,1) circle (0pt) node {{\scriptsize $v$}};
\fill[red!80!black] (0.75,1) circle (0pt) node {{\scriptsize $u$}};
\fill[red!80!black] (-0.75,1) circle (0pt) node {{\scriptsize $u$}};
\end{tikzpicture}
= 1_{Y\otimes X}
\, , 
\qquad 
\begin{tikzpicture}[very thick,scale=1.0,color=blue!50!black, baseline=0.9cm]
\coordinate (X) at (0.5,0);
\coordinate (Xd) at (-0.5,0);
\coordinate (X2) at (0.5,2);
\coordinate (Xd2) at (-0.5,2);
\coordinate (d1) at (-1,0);
\coordinate (d2) at (+1,0);
\coordinate (u1) at (-1,2);
\coordinate (u2) at (+1,2);
%
\fill [orange!40!white, opacity=0.7] (d1) -- (d2) -- (u2) -- (u1); 
\fill [orange!60!white, opacity=0.7] (0,1) circle (0.6);
%
\draw[thin] (d1) -- (d2) -- (u2) -- (u1) -- (d1); 
\fill (-0.8,0.2) circle (0pt) node[red!80!black] {{\scriptsize$u$}};
\fill (0,1) circle (0pt) node[red!80!black] {{\scriptsize$v$}};
\fill (0,1.6) circle (2.5pt) node[above] {{\scriptsize $\varphi$}};
\fill (0,0.4) circle (2.5pt) node[below] {{\scriptsize $\psi$}};
\fill (+0.77,1) circle (0pt) node {{\scriptsize $X$}};
\fill (-0.75,1) circle (0pt) node {{\scriptsize $Y$}};
\draw (0,1) circle (0.6);
\end{tikzpicture} 
=
\begin{tikzpicture}[very thick,scale=1.0,color=blue!50!black, baseline=0.9cm]
\coordinate (X) at (0.5,0);
\coordinate (Xd) at (-0.5,0);
\coordinate (X2) at (0.5,2);
\coordinate (Xd2) at (-0.5,2);
\coordinate (d1) at (-1,0);
\coordinate (d2) at (+1,0);
\coordinate (u1) at (-1,2);
\coordinate (u2) at (+1,2);
%
\fill [orange!40!white, opacity=0.7] (d1) -- (d2) -- (u2) -- (u1); 
%
\draw[thin] (d1) -- (d2) -- (u2) -- (u1) -- (d1); 
\fill (-0.8,0.2) circle (0pt) node[red!80!black] {{\scriptsize$u$}};
%
\end{tikzpicture} 
= 1_{1_u}
\, , 
\label{eq:weak-invertible-2cat-1st}
\\
\begin{tikzpicture}[very thick,scale=1.0,color=blue!50!black, baseline=0.9cm]
\coordinate (X) at (0.5,0);
\coordinate (Xd) at (-0.5,0);
\coordinate (X2) at (0.5,2);
\coordinate (Xd2) at (-0.5,2);
\coordinate (d1) at (-1,0);
\coordinate (d2) at (+1,0);
\coordinate (u1) at (-1,2);
\coordinate (u2) at (+1,2);
%
\fill [orange!40!white, opacity=0.7] (d1) -- (d2) -- (u2) -- (u1); 
\fill [orange!60!white, opacity=0.7] (d1) -- (Xd) .. controls +(0,1) and +(0,1) .. (X) -- (d2) -- (u2) -- (X2) .. controls +(0,-1) and +(0,-1) .. (Xd2) -- (u1); 
%
\draw[thin] (d1) -- (d2) -- (u2) -- (u1) -- (d1); 
%
\draw (X) .. controls +(0,1) and +(0,1) .. (Xd);
\draw (X2) .. controls +(0,-1) and +(0,-1) .. (Xd2);
%
\fill ($(X)+(-0.1,0)$) circle (0pt) node[above right] {{\scriptsize $Y$}};
\fill ($(Xd)+(0.1,0)$) circle (0pt) node[above left] {{\scriptsize $X$}};
\fill ($(X2)+(-0.1,0)$) circle (0pt) node[below right] {{\scriptsize $Y$}};
\fill ($(Xd2)+(0.1,0)$) circle (0pt) node[below left] {{\scriptsize $X$}};
\fill[red!80!black] (0,0.2) circle (0pt) node {{\scriptsize $u$}};
\fill[red!80!black] (0,1.8) circle (0pt) node {{\scriptsize $u$}};
\fill[red!80!black] (0.75,1) circle (0pt) node {{\scriptsize $v$}};
\fill (0,0.73) circle (2.5pt) node[below] {{\scriptsize $\zeta$}};
\fill (0,1.26) circle (2.5pt) node[above] {{\scriptsize $\xi$}};
\end{tikzpicture}
& =
\begin{tikzpicture}[very thick,scale=1.0,color=blue!50!black, baseline=0.9cm]
\coordinate (X) at (0.5,0);
\coordinate (Xd) at (-0.5,0);
\coordinate (X2) at (0.5,2);
\coordinate (Xd2) at (-0.5,2);
\coordinate (d1) at (-1,0);
\coordinate (d2) at (+1,0);
\coordinate (u1) at (-1,2);
\coordinate (u2) at (+1,2);
%
\fill [orange!40!white, opacity=0.7] (d1) -- (d2) -- (u2) -- (u1); 
\fill [orange!60!white, opacity=0.7] (X) --(d2) -- (u2) -- (X2); 
\fill [orange!60!white, opacity=0.7] (Xd) --(d1) -- (u1) -- (Xd2); 
%
\draw[thin] (d1) -- (d2) -- (u2) -- (u1) -- (d1); 
%
\draw (X) -- (X2);
\draw (Xd) -- (Xd2);
%
\fill ($(X)+(-0.1,0)$) circle (0pt) node[above right] {{\scriptsize $Y$}};
\fill ($(Xd)+(0.1,0)$) circle (0pt) node[above left] {{\scriptsize $X$}};
\fill[red!80!black] (0,1) circle (0pt) node {{\scriptsize $u$}};
\fill[red!80!black] (0.75,1) circle (0pt) node {{\scriptsize $v$}};
\fill[red!80!black] (-0.75,1) circle (0pt) node {{\scriptsize $v$}};
\end{tikzpicture}
= 1_{X\otimes Y}
\, , 
\qquad 
\begin{tikzpicture}[very thick,scale=1.0,color=blue!50!black, baseline=0.9cm]
\coordinate (X) at (0.5,0);
\coordinate (Xd) at (-0.5,0);
\coordinate (X2) at (0.5,2);
\coordinate (Xd2) at (-0.5,2);
\coordinate (d1) at (-1,0);
\coordinate (d2) at (+1,0);
\coordinate (u1) at (-1,2);
\coordinate (u2) at (+1,2);
%
\fill [orange!40!white, opacity=0.7] (d1) -- (d2) -- (u2) -- (u1); 
\fill [orange!60!white, opacity=0.7] (d1) -- (d2) -- (u2) -- (u1); 
\fill [white] (0,1) circle (0.6);
\fill [orange!40!white, opacity=0.7] (0,1) circle (0.6);
%
\draw[thin] (d1) -- (d2) -- (u2) -- (u1) -- (d1); 
\fill (-0.8,0.2) circle (0pt) node[red!80!black] {{\scriptsize$v$}};
\fill (0,1) circle (0pt) node[red!80!black] {{\scriptsize$u$}};
\fill (0,1.6) circle (2.5pt) node[above] {{\scriptsize $\zeta$}};
\fill (0,0.4) circle (2.5pt) node[below] {{\scriptsize $\xi$}};
\fill (+0.77,1) circle (0pt) node {{\scriptsize $Y$}};
\fill (-0.75,1) circle (0pt) node {{\scriptsize $X$}};
\draw (0,1) circle (0.6);
\end{tikzpicture} 
=
\begin{tikzpicture}[very thick,scale=1.0,color=blue!50!black, baseline=0.9cm]
\coordinate (X) at (0.5,0);
\coordinate (Xd) at (-0.5,0);
\coordinate (X2) at (0.5,2);
\coordinate (Xd2) at (-0.5,2);
\coordinate (d1) at (-1,0);
\coordinate (d2) at (+1,0);
\coordinate (u1) at (-1,2);
\coordinate (u2) at (+1,2);
%
\fill [orange!40!white, opacity=0.7] (d1) -- (d2) -- (u2) -- (u1); 
\fill [orange!60!white, opacity=0.7] (d1) -- (d2) -- (u2) -- (u1); 
%
\draw[thin] (d1) -- (d2) -- (u2) -- (u1) -- (d1); 
\fill (-0.8,0.2) circle (0pt) node[red!80!black] {{\scriptsize$v$}};
%
\end{tikzpicture} 
= 1_{1_v}
\, . 
\label{eq:weak-invertible-2cat-2nd}
\end{align}
Hence~$X$ and~$Y$ are \textsl{weakly inverse} to one another, in the sense that their compositions are isomorphic (but not necessarily equal) to identity 1-morphisms. 

A weaker notion is that of a \textsl{(2-)condensation of~$u$ onto~$v$}, which consists only of morphisms $X,Y,\zeta,\xi$ as above, and only the last of the four conditions has to hold, namely the \textsl{bubble condition} $\zeta\circ\xi=1_{1_v}$ in \eqref{eq:weak-invertible-2cat-2nd}. 
In this case $A:= Y\otimes X\colon u\to u$ with (co)multiplication $\mu:=1_Y\otimes\zeta\otimes 1_X$, $\Delta:= 1_Y\otimes \xi\otimes 1_X$ is a $\Delta$-separable Frobenius algebra~$\mathcal{A}$.
While $v$ need not be equivalent to $u$ in $\mathcal B$, it (or rather the trivial algebra $1_v\colon v\to v$ on $v$) is equivalent to~$\mathcal{A}$ in the \textsl{condensation completion} $\overline{\mathcal{B}} = \textrm{Kar}(\mathcal B)$ of~$\mathcal B$, see \cite[Prop.\,4.4]{Carqueville:2012dk} (where~$\overline{\mathcal{B}}$ was called ``equivariant completion'') and \cite[Sect.\,2]{GaiottoJohnsonFreyd}. 
By definition, objects of~$\overline{\mathcal{B}}$ are $\Delta$-separable Frobenius algebras in~$\mathcal B$, 1-morphisms are bimodules, and 2-morphisms are bimodule maps. 

As observed in \cite{GaiottoJohnsonFreyd}, the condensation completion~$\overline{\mathcal{B}}$ is a categorification of the idempotent completion of 1-categories, where the identity $e\circ e=e$ for idempotents is replaced by the multiplication $\mu\colon A\otimes A \to A$ of a $\Delta$-separable Frobenius algebra. 
Categorifications are not unique, and we may refer to this particular one as a ``framed categorification''. 

If the 2-category~$\mathcal B$ comes with a pivotal structure (as it does when it arises from an oriented defect TFT), then there is a natural oriented version of the above framed categorification of idempotent completion. 
For this we will need the left and right quantum dimension of a 1-morphism $X\colon u\to v$,
\begin{equation}\label{eq:diml-dimr}
\dim_{\textrm{l}}(X) := 
\textrm{ev}_X \circ \widetilde{\textrm{coev}}_X
=
\begin{tikzpicture}[very thick,scale=1.0,color=blue!50!black, baseline=0.9cm]
\coordinate (X) at (0.5,0);
\coordinate (Xd) at (-0.5,0);
\coordinate (X2) at (0.5,2);
\coordinate (Xd2) at (-0.5,2);
\coordinate (d1) at (-1,0);
\coordinate (d2) at (+1,0);
\coordinate (u1) at (-1,2);
\coordinate (u2) at (+1,2);
%
\fill [orange!40!white, opacity=0.7] (d1) -- (d2) -- (u2) -- (u1); 
\fill [orange!60!white, opacity=0.7] (0,1) circle (0.6);
%
\draw[thin] (d1) -- (d2) -- (u2) -- (u1) -- (d1); 
\fill (-0.8,0.2) circle (0pt) node[red!80!black] {{\scriptsize$u$}};
\fill (0,1) circle (0pt) node[red!80!black] {{\scriptsize$v$}};
\fill (+0.77,1) circle (0pt) node {{\scriptsize $X$}};
\draw (0,1) circle (0.6);
\draw[<-, very thick] (0.100,0.4) -- (-0.101,0.4) node[above] {}; 
\draw[<-, very thick] (-0.100,1.6) -- (0.101,1.6) node[below] {}; 
\end{tikzpicture} 
\, , \quad 
\dim_{\textrm{r}}(X) := 
\widetilde{\textrm{ev}}_X \circ \textrm{coev}_X
=
\begin{tikzpicture}[very thick,scale=1.0,color=blue!50!black, baseline=0.9cm]
\coordinate (X) at (0.5,0);
\coordinate (Xd) at (-0.5,0);
\coordinate (X2) at (0.5,2);
\coordinate (Xd2) at (-0.5,2);
\coordinate (d1) at (-1,0);
\coordinate (d2) at (+1,0);
\coordinate (u1) at (-1,2);
\coordinate (u2) at (+1,2);
%
\fill [orange!40!white, opacity=0.7] (d1) -- (d2) -- (u2) -- (u1); 
\fill [orange!60!white, opacity=0.7] (d1) -- (d2) -- (u2) -- (u1); 
\fill [white] (0,1) circle (0.6);
\fill [orange!40!white, opacity=0.7] (0,1) circle (0.6);
%
\draw[thin] (d1) -- (d2) -- (u2) -- (u1) -- (d1); 
\fill (-0.8,0.2) circle (0pt) node[red!80!black] {{\scriptsize$v$}};
\fill (0,1) circle (0pt) node[red!80!black] {{\scriptsize$u$}};
\draw[->, very thick] (0.100,0.4) -- (-0.101,0.4) node[above] {}; 
\draw[->, very thick] (-0.100,1.6) -- (0.101,1.6) node[below] {}; 
\fill (-0.75,1) circle (0pt) node {{\scriptsize $X$}};
\draw (0,1) circle (0.6);
\end{tikzpicture} 
~.
\end{equation}
Suppose that $\dim_\mathrm{r}(X)$ is invertible. 
Then by setting $Y=X^*$, $\zeta = \widetilde{\textrm{ev}}_X$, $\xi=\textrm{coev}_X \circ \dim_\mathrm{r}(X)^{-1}$ above, one finds that
\begin{equation}
\label{eq:2d-orb-datum}
\mathcal A = 
\left( 
u \, , \quad 
A := X^* \otimes X \colon u \lra u 
\, , \quad 
\mu := 
\begin{tikzpicture}[very thick,scale=0.9,color=blue!50!black, baseline=.7cm]
\fill [orange!40!white, opacity=0.7] (1.25,0) -- (3.75,0) -- (3.75,1.655) -- (1.25,1.665);
\fill [orange!60!white, opacity=0.7] (3.28,0) -- (3,0) .. controls +(0,1) and +(0,1) .. (2,0) -- (1.55,0) -- (1.55,0) .. controls +(0,0.75) and +(0,-0.25) .. (2.3,1.25) -- (2.3,1.655) -- (2.7,1.655) -- (2.7,1.25) .. controls +(0,-0.25) and +(0,0.75) .. (3.45,0);
\draw[line width=0pt] 
(2.77,0.13) node[line width=0pt] (D) {{\tiny$X^*$}}
(2.16,0.13) node[line width=0pt] (s) {{\tiny$X\vphantom{X^*}$}}; 
\draw[redirected] (3,0) .. controls +(0,1) and +(0,1) .. (2,0);
\draw[line width=0pt] (2.5,0.4) node[red!80!black] (D) {{\tiny$u$}};
\draw[line width=0pt] (2.5,1.1) node[red!80!black] (D) {{\tiny$v$}};
\draw[line width=0pt] (3.4,1.2) node[red!80!black] (D) {{\tiny$u$}};
\draw[line width=0pt] (1.6,1.2) node[red!80!black] (D) {{\tiny$u$}};
\draw[line width=0pt] 
(3.28,0.13) node[line width=0pt] (re) {{\tiny$X\vphantom{X^*}$}}
(1.77,0.13) node[line width=0pt] (li) {{\tiny$X^*$}}; 
\draw[line width=0pt] 
(2.85,1.55) node[line width=0pt] (ore) {{\tiny$X\vphantom{X^*}$}}
(2.1,1.55) node[line width=0pt] (oli) {{\tiny$X^*$}}; 
\draw (1.55,0) .. controls +(0,0.75) and +(0,-0.25) .. (2.3,1.25);
\draw (2.3,1.25) -- (2.3,1.655);
\draw (3.45,0) .. controls +(0,0.75) and +(0,-0.25) .. (2.7,1.25);
\draw (2.7,1.25) -- (2.7,1.655);
\end{tikzpicture}
\, , \quad 
\Delta := 
\begin{tikzpicture}[very thick,scale=0.9,color=blue!50!black, baseline=-0.8cm, rotate=180]
\fill [orange!40!white, opacity=0.7] (1.25,0) -- (3.75,0) -- (3.75,1.655) -- (1.25,1.665);
\fill [orange!60!white, opacity=0.7] (3.28,0) -- (3,0) .. controls +(0,1) and +(0,1) .. (2,0) -- (1.55,0) -- (1.55,0) .. controls +(0,0.75) and +(0,-0.25) .. (2.3,1.25) -- (2.3,1.655) -- (2.7,1.655) -- (2.7,1.25) .. controls +(0,-0.25) and +(0,0.75) .. (3.45,0);
\draw[line width=0pt] 
(2.8,0.13) node[line width=0pt] (D) {{\tiny$X\vphantom{X^*}$}}
(2.2,0.13) node[line width=0pt] (s) {{\tiny$X^*\vphantom{X^*}$}}; 
\draw[redirected] (3,0) .. controls +(0,1) and +(0,1) .. (2,0);
\draw[line width=0pt] (2.5,0.4) node[red!80!black] (D) {{\tiny$u$}};
\draw[line width=0pt] (2.5,1.1) node[red!80!black] (D) {{\tiny$v$}};
\draw[line width=0pt] (3.4,1.2) node[red!80!black] (D) {{\tiny$u$}};
\draw[line width=0pt] (1.6,1.2) node[red!80!black] (D) {{\tiny$u$}};
\draw[line width=0pt] 
(3.2,0.13) node[line width=0pt] (re) {{\tiny$X^*\vphantom{X^*}$}}
(1.7,0.13) node[line width=0pt] (li) {{\tiny$X\vphantom{X^*}$}}; 
\draw[line width=0pt] 
(2.9,1.55) node[line width=0pt] (ore) {{\tiny$X^*\vphantom{X^*}$}}
(2.15,1.55) node[line width=0pt] (oli) {{\tiny$X\vphantom{X^*}$}}; 
\draw (2.5,1.4) node (D) {{\tiny$\boldsymbol{\star}$}}; 
\draw (1.55,0) .. controls +(0,0.75) and +(0,-0.25) .. (2.3,1.25);
\draw (2.3,1.25) -- (2.3,1.655);
\draw (3.45,0) .. controls +(0,0.75) and +(0,-0.25) .. (2.7,1.25);
\draw (2.7,1.25) -- (2.7,1.655);
\end{tikzpicture}
\right) 
\end{equation}
where $\boldsymbol{\star} := \dim_{\textrm{r}}(X)^{-1}$, is now \textsl{symmetric} in addition to being a $\Delta$-separable Frobenius algebra. 
Moreover, the algebras $\mathcal A \colon u \to u$ and $1_v \colon v \to v$ are equivalent in the \textsl{orbifold completion}~$\mathcal B_{\textrm{orb}}$ of~$\mathcal B$, see \cite[Sect.\,5]{Carqueville:2012dk}. 
By definition, $\mathcal B_{\textrm{orb}}$ has $\Delta$-separable symmetric Frobenius algebras as objects, while 1- and 2-morphisms are bimodules and their maps. 

\medskip

Let us apply the above considerations to the case that
$\mathcal B = \mathcal D_\zz$ is obtained from an (oriented) 2d defect TFT~$\zz$. Write $\zz_u$ for the closed TFT associated to an object $u \in \mathcal{B}$. 
We stress that even if there are line defects $X \colon u \to v$ and $Y\colon v \to u$ which are weakly inverse to each other, this does \textsl{not} imply that $\zz_u(\Sigma) = \zz_v(\Sigma)$ for all closed surfaces~$\Sigma$. 
This equality however does follow if $Y = X^*$ and the 2-morphisms in \eqref{eq:weak-invertible-2cat-1st} and \eqref{eq:weak-invertible-2cat-2nd} are given by the duality data of the pivotal structure. 
In particular, in this case both dimensions in \eqref{eq:diml-dimr} are $1$. 
If we only demand $\dim_{\textrm{r}}(X) = 1_{1_v}$, then we find that~$\mathcal A$ as in~\eqref{eq:2d-orb-datum} is equivalent to the trivial Frobenius algebra $1_v\colon v\to v$ in the orbifold completion $\mathcal B_{\textrm{orb}}$. 
$\mathcal A$ is a gaugeable symmetry in the sense of Section~\ref{sec:gauging}. 
The equivalence 
$\mathcal A \cong 1_v$ in $\mathcal B_{\textrm{orb}}$ translates into an equivalence $\mathcal Z_u/\mathcal A \cong \mathcal Z_v$ of closed TFTs. 

If both quantum dimensions in~\eqref{eq:diml-dimr} are identities,\footnote{If $\mathcal B = \mathcal D_{\mathcal Z }$ for a defect TFT~$\mathcal Z$ and if the dimensions in \eqref{eq:diml-dimr} differ from identities by non-zero scalars, one can set them to unity by stacking with 2d Euler TFTs. 
Euler TFTs are examples of invertible topological field theories. They exist in any dimension and assign the Euler character (with suitable boundary term) to bordisms, see \cite{Quinnlectures} and e.g.\ \cite[Ex.\,2.14]{Carqueville:2017aoe}.
Note that the Euler character of an odd dimensional closed manifold is zero. In particular, stacking an Euler TFT on a defect cannot change the value of an odd-dimensional defect sphere. 
On the other hand, the Euler character of a ball is $1$ in any dimension, and so one can always change the value of a defect sphere by stacking the ball it bounds with an Euler TFT. 
Starting from dimension 4, there are more oriented non-extended invertible TFTs than Euler theories, a classification is given in \cite[Thm.\,3.14\,\&\,Ex.\,3.17]{Reutter:2022}.\label{fn:EulerTFT}} 
then we have another orbifold datum $\mathcal A' = (v, A':= X\otimes X^*,\dots)$ such that $\mathcal A' \cong 1_u$ in $\mathcal B_{\textrm{orb}}$, in addition to $\mathcal A \cong 1_v$ as above. 
In this symmetric situation we call~$u$ and~$v$ \textsl{orbifold equivalent}. 
This indeed gives an equivalence relation for objects in~$\mathcal B$. 
If $\mathcal B = \mathcal D_{\mathcal Z }$, we then in particular have $\mathcal Z_u/\mathcal A \cong \mathcal Z_v$ and $\mathcal Z_v/\mathcal A' \cong \mathcal Z_u$.

\subsection{Weak inverses and condensations in defect 3-categories}
\label{subsec:defect3categories}

The ``largest'' rigorously constructed 3d defect TFT is that of Reshetikhin--Turaev type, see \cite{Kapustin:2010if, Fuchs:2012dt, Carqueville:2017ono, Koppen:2021kry, Carqueville:2023aak}, closely related to quantised Chern--Simons theory. 
This in particular includes state sum models such as Dijkgraaf--Witten models for finite groups, and more generally Turaev--Viro--Barrett--Westbury models for spherical fusion categories. 
A non-semisimple example is believed to be given by Rozansky--Witten models \cite{RW1996, KRS, KR0909.3643}, 
which are topologically twisted 3-dimensional sigma models. 

The 3-category~$\mathcal D_\zz$ associated to any given 3d defect TFT~$\zz$ has the structure of a ``Gray category with duals'' \cite{CMS,BMS}. 
This is a categorification of the result on pivotal 2-categories in two dimensions (cf.\ Section~\ref{sec:2cat-form-def}) and essentially means that 1-morphisms~$u$ in~$\mathcal D_\zz$ have coherently isomorphic left and right adjoints~$u^\#$, 2-morphisms~$X$ have coherently isomorphic left and right adjoints~$X^*$, while $(-)^\#$ and~$(-)^*$ enjoy a compatibility condition and originate from orientation reversal. 
For Reshetikhin--Turaev models, this structure is constructed in \cite{Carqueville:2023aak}, while the state of the art for Rozansky--Witten models is summarised in  \cite{KR0909.3643, KRS} in general, and in \cite{BCFR} for affine target manifolds. 

In the remainder of this section we discuss various notions of equivalence in a given 3-category~$\mathcal T$. 
We employ the graphical calculus of \cite{BMS}, and as indicated by the cube in \eqref{eq:cobordism-is-graphical-calc} we read vertical composition~$\circ$ of 3-morphisms (``operator product of point insertions'') from bottom to top, horizontal composition~$\otimes$ of 2-morphisms (``fusion of line defects'') from right to left, and composition~$\btimes$ of 1-morphisms (``stacking of surface defects'') from front to back. 
The graphical calculus is illustrated by the following example: 
\begin{equation}
\label{eq:3dCalcExample}
\hspace{-2em}
\begin{tikzpicture}[thick,scale=2.321,color=blue!50!black, baseline=0.3cm, >=stealth, 
style={x={(-0.6cm,-0.4cm)},y={(1cm,-0.2cm)},z={(0cm,0.9cm)}}]
\pgfmathsetmacro{\yy}{0.2}
\coordinate (P) at (0.5, \yy, 0);
\coordinate (R) at (0.625, 0.5 + \yy/2, 0);
\coordinate (L) at (0.5, 0, 0);
\coordinate (R1) at (0.25, 1, 0);
\coordinate (R2) at (0.5, 1, 0);
\coordinate (R3) at (0.75, 1, 0);
\coordinate (Pt) at (0.5, \yy, 1);
\coordinate (Rt) at (0.375, 0.5 + \yy/2, 1);
\coordinate (Lt) at (0.5, 0, 1);
\coordinate (R1t) at (0.25, 1, 1);
\coordinate (R2t) at (0.5, 1, 1);
\coordinate (R3t) at (0.75, 1, 1);
\coordinate (alpha) at (0.5, 0.5, 0.5);
%
%
\fill [orange!30!white, opacity=0.8] (L) -- (P) -- (alpha) -- (Pt) -- (Lt);
\fill [orange!30!white, opacity=0.8] (Pt) -- (Rt) -- (alpha);
\fill [orange!30!white, opacity=0.8] (Rt) -- (R1t) -- (R1) -- (P) -- (alpha);
\fill [orange!30!white, opacity=0.8] (Rt) -- (R2t) -- (R2) -- (R) -- (alpha);
\draw[string=green!30!black, ultra thick] (alpha) -- (Rt);
\fill[color=blue!50!black] (0.5,0.83,0.77) circle (0pt) node[left] (0up) { {\scriptsize$X_4$} };
\fill [orange!30!white, opacity=0.8] (Pt) -- (R3t) -- (R3) -- (R) -- (alpha);
\fill [orange!30!white, opacity=0.8] (P) -- (R) -- (alpha);
%
\draw[string=blue!50!black, ultra thick] (P) -- (alpha);
\draw[string=blue!50!black, ultra thick] (R) -- (alpha);
\draw[string=blue!50!black, ultra thick] (alpha) -- (Pt);
%
\fill[color=blue!50!black] (alpha) circle (1.2pt) node[left] (0up) { {\scriptsize$\varphi$} };
%
\fill[red!80!black] (0.5,0.35,-0.17) circle (0pt) node {{\scriptsize $\hspace{7em}u$}};
\fill[red!80!black] (0.5,0.35,-0.02) circle (0pt) node {{\scriptsize $\hspace{1em}v$}};
\fill[red!80!black] (0.5,0.35,0.4) circle (0pt) node {{\scriptsize $\hspace{-4em}w$}};
\fill[red!80!black] (0.5,0.35,-0.07) circle (0pt) node {{\scriptsize $\hspace{9.5em}s$}};
\fill[red!80!black] (0.5,0.35,0.05) circle (0pt) node {{\scriptsize $\hspace{12em}r$}};
\fill[red!80!black] (0.5,0.35,0.95) circle (0pt) node {{\scriptsize $\hspace{3em}t$}};
\fill[black] (0.5,0.35,-0.15) circle (0pt) node {{\scriptsize $\hspace{0em}a$}};
\fill[black] (0.5,0.35,-0.19) circle (0pt) node {{\scriptsize $\hspace{9.3em}b$}};
\fill[black] (0.5,0.35,-0.08) circle (0pt) node {{\scriptsize $\hspace{11.8em}c$}};
\fill[black] (0.5,0.35,1.05) circle (0pt) node {{\scriptsize $\hspace{10em}d$}};
%
\fill[color=blue!50!black] (0.5,0.35,0.25) circle (0pt) node[left] (0up) { {\scriptsize$X_1\,$} };
\fill[color=blue!50!black] (0.5,0.8,0.21) circle (0pt) node[left] (0up) { {\scriptsize$\hspace{-1em}X_2$} };
\fill[color=blue!50!black] (0.5,0.4,0.71) circle (0pt) node[left] (0up) { {\scriptsize$X_3\,$} };
%
\draw [black,opacity=1, very thin] (Pt) -- (Lt) -- (L) -- (P);
\draw [black,opacity=1, very thin] (Pt) -- (Rt);
\draw [black,opacity=1, very thin] (Rt) -- (R1t) -- (R1) -- (P);
\draw [black,opacity=1, very thin] (Rt) -- (R2t) -- (R2) -- (R);
\draw [black,opacity=1, very thin] (Pt) -- (R3t) -- (R3) -- (R);
\draw [black,opacity=1, very thin] (P) -- (R);
\end{tikzpicture} 
\quad 
\textrm{represents a 3-morphism } \varphi\colon X_1 \otimes \big(1_r \btimes X_2 \big) \lra X_3 \otimes \big( X_4 \btimes 1_u \big) 
\end{equation}
for 1-morphisms $u\colon a\to b$, $v\colon a \to c$, 2-morphisms $X_1\colon r \btimes v\to w$, $X_2\colon s \btimes u\to v$ etc. 

Let~$\mathcal T$ be a 3-category. 
The categorification of the 2-dimensional notion of equivalence in~\eqref{eq:weak-invertible-2cat-1st} and~\eqref{eq:weak-invertible-2cat-2nd} involves more complexity, but no new ideas are needed: 
Two objects $a,b\in\mathcal T$ are \textsl{equivalent}, $a\cong b$, if there are
\begin{equation}
\label{eq:table}
	\begin{tabular}{l|l|ll}
		1-morphisms & 2-morphisms & 3-morphisms &
		\\
		\midrule[1.pt]
		$u\colon a \lra b$
		&
		$X\colon v\btimes u \lra 1_a$\phantom{D}
		& 
		$\varphi\colon Y\otimes X \lra 1_{v\btimes u}$ 
		&
		$\zeta\colon X\otimes Y\lra 1_u$
		\\
		$v\colon b\lra a$
		& 
		$Y\colon 1_a \lra v\btimes u$
		& 
		$\psi\colon 1_{v\btimes u} \lra Y\otimes X$
		& 
		$\xi\colon 1_u\lra X\otimes Y$
		\\
		\midrule[0.4pt]
		{}
		&
		$X'\colon u\btimes v \lra 1_b$
		& 
		$\varphi'\colon Y'\otimes X' \lra 1_{u\btimes v}$\phantom{D} 
		&
		$\zeta'\colon X'\otimes Y'\lra 1_v$
		\\
		{}
		& 
		$Y'\colon 1_b \lra u\btimes v$
		& 
		$\psi'\colon 1_{u\btimes v} \lra Y'\otimes X'$
		& 
		$\xi'\colon 1_v\lra X'\otimes Y'$	
	\end{tabular}
\end{equation}
such that
\begin{align}
\hspace{-0.5em}
\begin{tikzpicture}[thick,scale=0.75,color=black, baseline=1.15cm]
\coordinate (p1) at (0,0);
\coordinate (p2) at (0.5,0.2);
\coordinate (p3) at (-0.5,0.4);
\coordinate (p4) at (2.5,0);
\coordinate (p5) at (1.5,0.2);
\coordinate (p6) at (2.0,0.4);
\coordinate (s) at ($(0,1.5)$);
\coordinate (s2) at ($(0,-3)$);
\coordinate (d) at ($(4,0)$);
\coordinate (h) at ($(0,2.5)$);
\coordinate (t) at ($(0.1,0)$);
\coordinate (u) at ($(0,-1.1)$);
%
%
%
\fill[black] (2.2,0.2) node[line width=0pt] (ore) {{\tiny$b$}};
%
\fill [orange!40!white, opacity=0.8] 
(p3) .. controls +(0.5,0) and +(0,0.1) ..  (p2)
-- (p2) .. controls +(0,0.9) and +(0,0.9) ..  (p5)
-- (p5) .. controls +(0,0.1) and +(-0.3,0) .. (p6)
-- (p6) -- ($(p6)+(s)$) -- ($(p3)+(s)$);
%
%
\draw[thin] (p2) .. controls +(0,0.1) and +(0.5,0) ..  (p3); 
\draw[thin] (p6) .. controls +(-0.3,0) and +(0,0.1) ..  (p5); 
\draw[thin] (p6) -- ($(p6)+(s)$);
%
\fill [orange!20!white, opacity=0.8] 
(p1) .. controls +(0.3,0) and +(0,-0.1) ..  (p2)
-- (p2) .. controls +(0,0.9) and +(0,0.9) ..  (p5)
-- (p5) .. controls +(0,-0.1) and +(-0.5,0) ..  (p4) 
-- (p4) -- ($(p4)+(s)$) -- ($(p1)+(s)$);
%
\draw[thin] (p1) .. controls +(0.3,0) and +(0,-0.1) ..  (p2); 
\draw[thin] (p5) .. controls +(0,-0.1) and +(-0.5,0) ..  (p4); 
%
\draw[very thick, color=blue!50!black] (p2) .. controls +(0,0.9) and +(0,0.9) ..  (p5);
\draw[thin] (p1) -- ($(p1)+(s)$);
\draw[thin] (p3) -- ($(p3)+(s)$);
\draw[thin] (p4) -- ($(p4)+(s)$);
%
%
%
\fill [orange!40!white, opacity=0.8] 
($(p3)-(s2)$) .. controls +(0.5,0) and +(0,0.1) ..  ($(p2)-(s2)$)
-- ($(p2)-(s2)$) .. controls +(0,-0.9) and +(0,-0.9) ..  ($(p5)-(s2)$)
-- ($(p5)-(s2)$) .. controls +(0,0.1) and +(-0.3,0) .. ($(p6)-(s2)$)
-- ($(p6)-(s2)$) -- ($(p6)-(s)-(s2)$) -- ($(p3)-(s)-(s2)$);
%
\draw[thin] ($(p2)-(s2)$) .. controls +(0,0.1) and +(0.5,0) .. ($(p3)-(s2)$); 
\draw[thin] ($(p6)-(s2)$) .. controls +(-0.3,0) and +(0,0.1) .. ($(p5)-(s2)$); 
\draw[thin] ($(p6)-(s2)$) -- ($(p6)-(s)-(s2)$);
%
\fill [orange!20!white, opacity=0.8] 
($(p1)-(s2)$) .. controls +(0.3,0) and +(0,-0.1) ..  ($(p2)-(s2)$)
-- ($(p2)-(s2)$) .. controls +(0,-0.9) and +(0,-0.9) ..  ($(p5)-(s2)$)
-- ($(p5)-(s2)$) .. controls +(0,-0.1) and +(-0.5,0) ..  ($(p4)-(s2)$) 
-- ($(p4)-(s2)$) -- ($(p4)-(s)-(s2)$) -- ($(p1)-(s)-(s2)$);
%
\draw[thin] ($(p1)-(s2)$) .. controls +(0.3,0) and +(0,-0.1) ..  ($(p2)-(s2)$); 
\draw[thin] ($(p5)-(s2)$) .. controls +(0,-0.1) and +(-0.5,0) ..  ($(p4)-(s2)$); 
%
\draw[very thick, color=blue!50!black] ($(p2)-(s2)$) .. controls +(0,-0.9) and +(0,-0.9) ..  ($(p5)-(s2)$);
\draw[thin] ($(p1)-(s2)$) -- ($(p1)-(s)-(s2)$);
\draw[thin] ($(p3)-(s2)$) -- ($(p3)-(s)-(s2)$);
\draw[thin] ($(p4)-(s2)$) -- ($(p4)-(s)-(s2)$);
\fill[red!80!black] (0.3,1.7) circle (0pt) node {{\scriptsize $u$}};
\fill[red!80!black] (-0.3,2.5) circle (0pt) node {{\scriptsize $v$}};
\fill[red!80!black] (1.8,3.2) circle (0pt) node {{\scriptsize $v$}};
\fill[color=blue!50!black] (1,0.87) circle (3pt) node[above] (0up) { {\scriptsize$\varphi$} };
\fill[color=blue!50!black] (1,2.53) circle (3pt) node[below] (0up) { {\scriptsize$\psi$} };
\fill[blue!50!black] (1.55,2.65) node[line width=0pt] (ore) {{\tiny$X$}};
\fill[blue!50!black] (0.45,2.65) node[line width=0pt] (ore) {{\tiny$Y$}};
\fill[blue!50!black] (1.6,0.65) node[line width=0pt] (ore) {{\tiny$X$}};
\fill[blue!50!black] (0.45,0.65) node[line width=0pt] (ore) {{\tiny$Y$}};
\fill[black] (1,0.2) node[line width=0pt] (ore) {{\tiny$a\vphantom{b}$}};
\fill[black] (-0.3,0.2) node[line width=0pt] (ore) {{\tiny$b$}};
\end{tikzpicture}
& =
\begin{tikzpicture}[thick,scale=0.75,color=black, baseline=0cm]
\coordinate (p1) at (0,0);
\coordinate (p2) at (0.5,0.2);
\coordinate (p3) at (-0.5,0.4);
\coordinate (p4) at (2.5,0);
\coordinate (p5) at (1.5,0.2);
\coordinate (p6) at (2.0,0.4);
\coordinate (s) at ($(0,1.5)$);
\coordinate (s2) at ($(0,0)$);
\coordinate (d) at ($(0,0)$);
\coordinate (h) at ($(0,2.5)$);
\coordinate (t) at ($(0.1,0)$);
\coordinate (u) at ($(0,-1.1)$);
\fill[black] (2.3,-1.3) node[line width=0pt] (ore) {{\tiny$b$}};
%
\fill [orange!40!white, opacity=0.8] 
($(p2)+(d)-(s)-(s2)$) .. controls +(0,0.1) and +(0.5,0) ..  ($(p3)+(d)-(s)-(s2)$)
-- ($(p3)+(d)-(s)-(s2)$) -- ($(p3)+(d)+(s)$) 
-- ($(p3)+(d)+(s)$) .. controls +(0.5,0) and +(0,0.1) ..($(p2)+(d)+(s)$); 
%
\draw[thin] ($(p2)+(d)-(s)-(s2)$) .. controls +(0,0.1) and +(0.5,0) ..  ($(p3)+(d)-(s)-(s2)$); 
\draw[thin] ($(p2)+(d)+(s)$) .. controls +(0,0.1) and +(0.5,0) ..  ($(p3)+(d)+(s)$); 
%
\fill [orange!20!white, opacity=0.8] 
($(p1)+(d)-(s)-(s2)$) .. controls +(0.3,0) and +(0,-0.1) .. ($(p2)+(d)-(s)-(s2)$)
-- ($(p2)+(d)-(s)-(s2)$) -- ($(p2)+(d)+(s)$)
-- ($(p2)+(d)+(s)$) .. controls +(0,-0.1) and +(0.3,0) .. ($(p1)+(d)+(s)$);
%
\draw[thin] ($(p1)+(d)-(s)-(s2)$) .. controls +(0.3,0) and +(0,-0.1) ..  ($(p2)+(d)-(s)-(s2)$); 
\draw[thin] ($(p1)+(d)+(s)$) .. controls +(0.3,0) and +(0,-0.1) ..  ($(p2)+(d)+(s)$); 
%
\draw[thin] ($(p1)+(d)-(s)-(s2)$) -- ($(p1)+(d)+(s)$);
\draw[thin] ($(p3)+(d)-(s)-(s2)$) -- ($(p3)+(d)+(s)$);
\draw[very thick, color=blue!50!black] ($(p2)+(d)-(s)-(s2)$) -- ($(p2)+(d)+(s)$);
%
%
\fill [orange!40!white, opacity=0.8] 
($(p5)-(d)-(s)-(s2)$) .. controls +(0,0.1) and +(-0.3,0) .. ($(p6)-(d)-(s)-(s2)$)
-- ($(p6)-(d)-(s)-(s2)$) -- ($(p6)-(d)+(s)$)
-- ($(p6)-(d)+(s)$) .. controls +(-0.3,0) and +(0,0.1) .. ($(p5)-(d)+(s)$);
%
%
\draw[thin] ($(p6)-(d)-(s)-(s2)$) -- ($(p6)-(d)+(s)$);
%
\draw[thin] ($(p5)-(d)-(s)-(s2)$) .. controls +(0,0.1) and +(-0.3,0) ..  ($(p6)-(d)-(s)-(s2)$); 
\draw[thin] ($(p5)-(d)+(s)$) .. controls +(0,0.1) and +(-0.3,0) ..  ($(p6)-(d)+(s)$); 
colouring front: 
\fill [orange!20!white, opacity=0.8] 
($(p5)-(d)-(s)-(s2)$) .. controls +(0,-0.1) and +(-0.5,0) ..  ($(p4)-(d)-(s)-(s2)$)
-- ($(p4)-(d)-(s)-(s2)$) -- ($(p4)-(d)+(s)$) 
-- ($(p4)-(d)+(s)$) .. controls +(-0.5,0) and +(0,-0.1) ..($(p5)-(d)+(s)$); 
%
\draw[thin] ($(p4)-(d)+(s)$) .. controls +(-0.5,0) and +(0,-0.1) ..($(p5)-(d)+(s)$);
\draw[thin] ($(p4)-(d)-(s)-(s2)$) .. controls +(-0.5,0) and +(0,-0.1) ..($(p5)-(d)-(s)-(s2)$);
%
\draw[thin] ($(p4)-(d)-(s)-(s2)$) -- ($(p4)-(d)+(s)$);
\draw[very thick, color=blue!50!black] ($(p5)-(d)-(s)-(s2)$) -- ($(p5)-(d)+(s)$);
\fill[red!80!black] (0.25,1.0) circle (0pt) node {{\scriptsize $u$}};
\fill[red!80!black] (-0.25,1.7) circle (0pt) node {{\scriptsize $v$}};
\fill[red!80!black] (1.8,1.75) circle (0pt) node {{\scriptsize $v$}};
\fill[red!80!black] (2.25,1.0) circle (0pt) node {{\scriptsize $u$}};
\fill[blue!50!black] (1.3,0.3) node[line width=0pt] (ore) {{\tiny$X$}};
\fill[blue!50!black] (0.7,0.3) node[line width=0pt] (ore) {{\tiny$Y$}};
\fill[black] (1,-1.3) node[line width=0pt] (ore) {{\tiny$a\vphantom{b}$}};
\fill[black] (-0.3,-1.3) node[line width=0pt] (ore) {{\tiny$b$}};
\end{tikzpicture}
= 1_{1_{Y\otimes X}}
\, , 
&
\begin{tikzpicture}[thick,scale=0.75,color=black, baseline=0cm]
\coordinate (p1) at (0,0);
\coordinate (p2) at (0.5,0.2);
\coordinate (p3) at (-0.5,0.4);
\coordinate (p4) at (2.5,0);
\coordinate (p5) at (1.5,0.2);
\coordinate (p6) at (2.0,0.4);
\coordinate (s) at ($(0,1.5)$);
\coordinate (s2) at ($(0,0)$);
\coordinate (d) at ($(4,0)$);
\coordinate (h) at ($(0,2.5)$);
\coordinate (t) at ($(0.1,0)$);
\coordinate (u) at ($(0,-1.1)$);
%
%
\fill [orange!40!white, opacity=0.8] 
($(p3)-(s2)$) .. controls +(0.5,0) and +(0,0.1) ..  ($(p2)-(s2)$)
-- ($(p2)-(s2)$) .. controls +(0,-0.9) and +(0,-0.9) ..  ($(p5)-(s2)$)
-- ($(p5)-(s2)$) .. controls +(0,0.1) and +(-0.3,0) .. ($(p6)-(s2)$)
-- ($(p6)-(s2)$) -- ($(p6)-(s)-(s2)$) -- ($(p3)-(s)-(s2)$);
%
\draw[thin] ($(p6)-(s)-(s2)$) -- ($(p3)-(s)-(s2)$);	
\draw[thin] ($(p6)-(s2)$) -- ($(p6)-(s)-(s2)$);
%
\fill [orange!20!white, opacity=0.8] 
($(p1)-(s2)$) .. controls +(0.3,0) and +(0,-0.1) ..  ($(p2)-(s2)$)
-- ($(p2)-(s2)$) .. controls +(0,-0.9) and +(0,-0.9) ..  ($(p5)-(s2)$)
-- ($(p5)-(s2)$) .. controls +(0,-0.1) and +(-0.5,0) ..  ($(p4)-(s2)$) 
-- ($(p4)-(s2)$) -- ($(p4)-(s)-(s2)$) -- ($(p1)-(s)-(s2)$);
%
\draw[thin] ($(p4)-(s)-(s2)$) -- ($(p1)-(s)-(s2)$); 
%
\draw[very thick, color=blue!50!black] ($(p2)-(s2)$) .. controls +(0,-0.9) and +(0,-0.9) ..  ($(p5)-(s2)$);
\draw[thin] ($(p1)-(s2)$) -- ($(p1)-(s)-(s2)$);
\draw[thin] ($(p3)-(s2)$) -- ($(p3)-(s)-(s2)$);
\draw[thin] ($(p4)-(s2)$) -- ($(p4)-(s)-(s2)$);
%
%
\fill [orange!40!white, opacity=0.8] 
(p3) .. controls +(0.5,0) and +(0,0.1) ..  (p2)
-- (p2) .. controls +(0,0.9) and +(0,0.9) ..  (p5)
-- (p5) .. controls +(0,0.1) and +(-0.3,0) .. (p6)
-- (p6) -- ($(p6)+(s)$) -- ($(p3)+(s)$);
%
\draw[thin] (p6) -- ($(p6)+(s)$);
%
\fill [orange!20!white, opacity=0.8] 
(p1) .. controls +(0.3,0) and +(0,-0.1) ..  (p2)
-- (p2) .. controls +(0,0.9) and +(0,0.9) ..  (p5)
-- (p5) .. controls +(0,-0.1) and +(-0.5,0) ..  (p4) 
-- (p4) -- ($(p4)+(s)$) -- ($(p1)+(s)$);
%
\draw[thin] ($(p4)+(s)$) -- ($(p1)+(s)$);
\draw[thin] ($(p6)+(s)$) -- ($(p3)+(s)$);
%
\draw[very thick, color=blue!50!black] (p2) .. controls +(0,0.9) and +(0,0.9) ..  (p5);
\draw[very thick, color=blue!50!black] ($(p2)-(s2)$) .. controls +(0,-0.9) and +(0,-0.9) ..  ($(p5)-(s2)$);
\draw[thin] (p1) -- ($(p1)+(s)$);
\draw[thin] (p3) -- ($(p3)+(s)$);
\draw[thin] (p4) -- ($(p4)+(s)$);
\fill[color=blue!50!black] (1,0.87) circle (3pt) node[above] (0up) { {\scriptsize$\varphi$} };
\fill[color=blue!50!black] (1,-0.46) circle (3pt) node[below] (0up) { {\scriptsize$\psi$} };
\fill[blue!50!black] (1.7,0.25) node[line width=0pt] (ore) {{\tiny$X$}};
\fill[blue!50!black] (0.3,0.25) node[line width=0pt] (ore) {{\tiny$Y$}};
\fill[red!80!black] (0.25,1.3) circle (0pt) node {{\scriptsize $u$}};
\fill[red!80!black] (-0.3,1.7) circle (0pt) node {{\scriptsize $v$}};
\fill[black] (0.2,-1.65) node[line width=0pt] (ore) {{\tiny$a$}};
\fill[black] (-0.25,-1.3) node[line width=0pt] (ore) {{\tiny$b$}};
\fill[black] (-0.65,-0.95) node[line width=0pt] (ore) {{\tiny$a$}};
\end{tikzpicture}
&
=
\begin{tikzpicture}[thick,scale=0.75,color=black, baseline=0cm]
\coordinate (p1) at (0,0);
\coordinate (p2) at (0.5,0.2);
\coordinate (p3) at (-0.5,0.4);
\coordinate (p4) at (1.5,0);
\coordinate (p5) at (0.5,0.2);
\coordinate (p6) at (1.0,0.4);
\coordinate (s) at ($(0,1.5)$);
\coordinate (s2) at ($(0,0)$);
\coordinate (d) at ($(4,0)$);
\coordinate (h) at ($(0,2.5)$);
\coordinate (t) at ($(0.1,0)$);
\coordinate (u) at ($(0,-1.1)$);
%
\fill [orange!40!white, opacity=0.8] 
($(p3)+(s)$) -- ($(p6)+(s)$) -- ($(p6)-(s)$) -- ($(p3)-(s)$);
\draw[thin] ($(p6)-(s)-(s2)$) -- ($(p3)-(s)-(s2)$);	
\draw[thin] ($(p6)-(s2)$) -- ($(p6)-(s)-(s2)$);
\draw[thin] (p6) -- ($(p6)+(s)$);
\fill [orange!20!white, opacity=0.8] 
($(p1)+(s)$) -- ($(p4)+(s)$) -- ($(p4)-(s)$) -- ($(p1)-(s)$);
%
%
\draw[thin] ($(p4)-(s)-(s2)$) -- ($(p1)-(s)-(s2)$); 
%
\draw[thin] ($(p1)-(s2)$) -- ($(p1)-(s)-(s2)$);
\draw[thin] ($(p3)-(s2)$) -- ($(p3)-(s)-(s2)$);
\draw[thin] ($(p4)-(s2)$) -- ($(p4)-(s)-(s2)$);
\draw[thin] ($(p4)+(s)$) -- ($(p1)+(s)$);
\draw[thin] ($(p6)+(s)$) -- ($(p3)+(s)$);
%
\draw[thin] (p1) -- ($(p1)+(s)$);
\draw[thin] (p3) -- ($(p3)+(s)$);
\draw[thin] (p4) -- ($(p4)+(s)$);
\fill[red!80!black] (0.25,1.3) circle (0pt) node {{\scriptsize $u$}};
\fill[red!80!black] (-0.3,1.7) circle (0pt) node {{\scriptsize $v$}};
\fill[black] (0.2,-1.65) node[line width=0pt] (ore) {{\tiny$a$}};
\fill[black] (-0.25,-1.3) node[line width=0pt] (ore) {{\tiny$b$}};
\fill[black] (-0.65,-0.95) node[line width=0pt] (ore) {{\tiny$a$}};
\end{tikzpicture}
= 1_{1_{1_{v\btimes u}}}
\, , 
\\
\begin{tikzpicture}[thick,scale=0.75,color=black, baseline=-1.05cm]
%
\coordinate (p1) at (0,0);
\coordinate (p2) at (1,0);
\fill [orange!40!white, opacity=0.8] 
(p1) .. controls +(0,-0.5) and +(0,-0.5) ..  (p2)
-- (p2) .. controls +(0,0.5) and +(0,0.5) ..  (p1)
;
\fill [orange!20!white, opacity=0.8] 
(p1) .. controls +(0,-0.5) and +(0,-0.5) ..  (p2)
-- (p2) .. controls +(0,-1.2) and +(0,-1.2) ..  (p1)
;
\draw[very thick, color=blue!50!black] (p1) .. controls +(0,-1.2) and +(0,-1.2) ..  (p2); 
\draw[thin] (p1) .. controls +(0,0.5) and +(0,0.5) ..  (p2); 
\draw[thin] (p1) .. controls +(0,-0.5) and +(0,-0.5) ..  (p2); 
%
\fill[red!80!black] (0.5,-0.35) circle (0pt) node[below] {{\scriptsize $u$}};
\fill[red!80!black] (0.5,0.3) circle (0pt) node[below] {{\scriptsize $v$}};
%
%
\coordinate (p1) at (0,-2.5);
\coordinate (p2) at (1,-2.5);
\fill [orange!40!white, opacity=0.8] 
(p1) .. controls +(0,-0.5) and +(0,-0.5) ..  (p2)
-- (p2) .. controls +(0,0.5) and +(0,0.5) ..  (p1);
\draw[thin] (p1) .. controls +(0,0.5) and +(0,0.5) ..  (p2); 
\fill[red!80!black] (0.5,-1.7) circle (0pt) node[below] {{\scriptsize $v$}};
%
\fill [orange!20!white, opacity=0.8] 
(p1) .. controls +(0,-0.5) and +(0,-0.5) ..  (p2)
-- (p2) .. controls +(0,1.2) and +(0,1.2) ..  (p1);
\draw[thin] (p1) .. controls +(0,-0.5) and +(0,-0.5) ..  (p2); 
\draw[very thick, color=blue!50!black] (p1) .. controls +(0,1.2) and +(0,1.2) ..  (p2); 
%
\fill[red!80!black] (0.5,-2.2) circle (0pt) node[below] {{\scriptsize $u$}};
\fill[color=blue!50!black] (0.5,-0.89) circle (3pt) node[right] (0up) { };
\fill[color=blue!50!black] (0.45,-1.05) circle (0pt) node[right] (0up) { {\scriptsize$\xi$} };
\fill[color=blue!50!black] (0.5,-1.6) circle (3pt) node[left] (0up) { };
\fill[color=blue!50!black] (0.53,-1.45) circle (0pt) node[left] (0up) { {\scriptsize$\zeta$} };
\fill[blue!50!black] (-0.15,-0.6) node[line width=0pt] (ore) {{\tiny$X$}};
\fill[blue!50!black] (1.1,-0.6) node[line width=0pt] (ore) {{\tiny$Y$}};
\fill[blue!50!black] (-0.15,-2) node[line width=0pt] (ore) {{\tiny$X$}};
\fill[blue!50!black] (1.1,-2) node[line width=0pt] (ore) {{\tiny$Y$}};
\fill[black] (-0.15,-1.2) node[line width=0pt] (ore) {{\tiny$a$}};
\end{tikzpicture}
&=
\begin{tikzpicture}[thick,scale=0.75,color=black, baseline=-1.05cm]
%
\coordinate (p1) at (0,0);
\coordinate (p2) at (1,0);
\fill [orange!40!white, opacity=0.8] 
(p1) .. controls +(0,-0.5) and +(0,-0.5) ..  (p2)
-- (p2) .. controls +(0,0.5) and +(0,0.5) ..  (p1);
%
\fill[red!80!black] (0.5,0.3) circle (0pt) node[below] {{\scriptsize $v$}};
%
%
\coordinate (q1) at (0,-2.5);
\coordinate (q2) at (1,-2.5);
\fill [orange!40!white, opacity=0.8] 
(q1) .. controls +(0,-0.5) and +(0,-0.5) ..  (q2)
-- (q2) .. controls +(0,0.5) and +(0,0.5) ..  (q1);
\draw[thin] (q1) .. controls +(0,0.5) and +(0,0.5) ..  (q2); 
%
\fill [orange!20!white, opacity=0.8] 
(q1) .. controls +(0,-0.5) and +(0,-0.5) ..  (q2)
-- (p2) .. controls +(0,-0.5) and +(0,-0.5) ..  (p1);
\draw[thin] (q1) .. controls +(0,-0.5) and +(0,-0.5) ..  (q2); 
\draw[thin] (p1) .. controls +(0,-0.5) and +(0,-0.5) ..  (p2); 
\draw[thin] (p1) .. controls +(0,+0.5) and +(0,+0.5) ..  (p2); 
\draw[very thick, color=blue!50!black] (q1) --  (p1); 
\draw[very thick, color=blue!50!black] (q2) --  (p2); 
%
\fill[red!80!black] (0.5,-1) circle (0pt) node[below] {{\scriptsize $u$}};
\fill[blue!50!black] (-0.2,-0.6) node[line width=0pt] (ore) {{\tiny$X$}};
\fill[blue!50!black] (1.2,-0.6) node[line width=0pt] (ore) {{\tiny$Y$}};
\fill[black] (-0.2,-1.2) node[line width=0pt] (ore) {{\tiny$a$}};
\end{tikzpicture}
= 
1_{X\otimes Y}
\, , 
& 
\begin{tikzpicture}[thick,scale=0.75,color=black, baseline=-0.1cm]
%
\coordinate (p1) at (0,0);
\coordinate (p2) at (1,0);
%
%
%
\fill[red!80!black] (0.5,0.5) circle (0pt) node[below] {{\tiny $v$}};
%
\fill [orange!20!white, opacity=0.8] 
(p1) .. controls +(0,-1.2) and +(0,-1.2) ..  (p2)
-- (p2) .. controls +(0,1.2) and +(0,1.2) ..  (p1);
\draw[very thick, color=blue!50!black] (p1) .. controls +(0,1.2) and +(0,1.2) ..  (p2); 
\draw[very thick, color=blue!50!black] (p1) .. controls +(0,-1.2) and +(0,-1.2) ..  (p2); 
%
\fill[color=blue!50!black] (0.5,-0.89) circle (3pt) node[right] (0up) { };
\fill[color=blue!50!black] (0.45,-1.05) circle (0pt) node[right] (0up) { {\scriptsize$\xi$} };
\fill[color=blue!50!black] (0.5,0.89) circle (3pt) node[left] (0up) { };
\fill[color=blue!50!black] (0.53,1.05) circle (0pt) node[left] (0up) { {\scriptsize$\zeta$} };
\fill[blue!50!black] (-0.2,-0.1) node[line width=0pt] (ore) {{\tiny$X$}};
\fill[blue!50!black] (1.2,-0.1) node[line width=0pt] (ore) {{\tiny$Y$}};
\fill[red!80!black] (0.5,0) circle (0pt) node[below] {{\scriptsize $u$}};
\fill[black] (-0.2,-0.7) node[line width=0pt] (ore) {{\tiny$a$}};
\end{tikzpicture}
&
=
1_{1_{1_a}}
\label{eq:3dBubble1}
\, , 
\\
\hspace{-0.5em}
\begin{tikzpicture}[thick,scale=0.75,color=black, baseline=1.15cm]
\coordinate (p1) at (0,0);
\coordinate (p2) at (0.5,0.2);
\coordinate (p3) at (-0.5,0.4);
\coordinate (p4) at (2.5,0);
\coordinate (p5) at (1.5,0.2);
\coordinate (p6) at (2.0,0.4);
\coordinate (s) at ($(0,1.5)$);
\coordinate (s2) at ($(0,-3)$);
\coordinate (d) at ($(4,0)$);
\coordinate (h) at ($(0,2.5)$);
\coordinate (t) at ($(0.1,0)$);
\coordinate (u) at ($(0,-1.1)$);
%
%
%
\fill[black] (2.2,0.2) node[line width=0pt] (ore) {{\tiny$b$}};
%
\fill [orange!20!white, opacity=0.8] 
(p3) .. controls +(0.5,0) and +(0,0.1) ..  (p2)
-- (p2) .. controls +(0,0.9) and +(0,0.9) ..  (p5)
-- (p5) .. controls +(0,0.1) and +(-0.3,0) .. (p6)
-- (p6) -- ($(p6)+(s)$) -- ($(p3)+(s)$);
%
%
\draw[thin] (p2) .. controls +(0,0.1) and +(0.5,0) ..  (p3); 
\draw[thin] (p6) .. controls +(-0.3,0) and +(0,0.1) ..  (p5); 
\draw[thin] (p6) -- ($(p6)+(s)$);
%
\fill [orange!40!white, opacity=0.8] 
(p1) .. controls +(0.3,0) and +(0,-0.1) ..  (p2)
-- (p2) .. controls +(0,0.9) and +(0,0.9) ..  (p5)
-- (p5) .. controls +(0,-0.1) and +(-0.5,0) ..  (p4) 
-- (p4) -- ($(p4)+(s)$) -- ($(p1)+(s)$);
%
\draw[thin] (p1) .. controls +(0.3,0) and +(0,-0.1) ..  (p2); 
\draw[thin] (p5) .. controls +(0,-0.1) and +(-0.5,0) ..  (p4); 
%
\draw[very thick, color=blue!50!black] (p2) .. controls +(0,0.9) and +(0,0.9) ..  (p5);
\draw[thin] (p1) -- ($(p1)+(s)$);
\draw[thin] (p3) -- ($(p3)+(s)$);
\draw[thin] (p4) -- ($(p4)+(s)$);
%
%
%
\fill [orange!20!white, opacity=0.8] 
($(p3)-(s2)$) .. controls +(0.5,0) and +(0,0.1) ..  ($(p2)-(s2)$)
-- ($(p2)-(s2)$) .. controls +(0,-0.9) and +(0,-0.9) ..  ($(p5)-(s2)$)
-- ($(p5)-(s2)$) .. controls +(0,0.1) and +(-0.3,0) .. ($(p6)-(s2)$)
-- ($(p6)-(s2)$) -- ($(p6)-(s)-(s2)$) -- ($(p3)-(s)-(s2)$);
%
\draw[thin] ($(p2)-(s2)$) .. controls +(0,0.1) and +(0.5,0) .. ($(p3)-(s2)$); 
\draw[thin] ($(p6)-(s2)$) .. controls +(-0.3,0) and +(0,0.1) .. ($(p5)-(s2)$); 
\draw[thin] ($(p6)-(s2)$) -- ($(p6)-(s)-(s2)$);
%
\fill [orange!40!white, opacity=0.8] 
($(p1)-(s2)$) .. controls +(0.3,0) and +(0,-0.1) ..  ($(p2)-(s2)$)
-- ($(p2)-(s2)$) .. controls +(0,-0.9) and +(0,-0.9) ..  ($(p5)-(s2)$)
-- ($(p5)-(s2)$) .. controls +(0,-0.1) and +(-0.5,0) ..  ($(p4)-(s2)$) 
-- ($(p4)-(s2)$) -- ($(p4)-(s)-(s2)$) -- ($(p1)-(s)-(s2)$);
%
\draw[thin] ($(p1)-(s2)$) .. controls +(0.3,0) and +(0,-0.1) ..  ($(p2)-(s2)$); 
\draw[thin] ($(p5)-(s2)$) .. controls +(0,-0.1) and +(-0.5,0) ..  ($(p4)-(s2)$); 
%
\draw[very thick, color=blue!50!black] ($(p2)-(s2)$) .. controls +(0,-0.9) and +(0,-0.9) ..  ($(p5)-(s2)$);
\draw[thin] ($(p1)-(s2)$) -- ($(p1)-(s)-(s2)$);
\draw[thin] ($(p3)-(s2)$) -- ($(p3)-(s)-(s2)$);
\draw[thin] ($(p4)-(s2)$) -- ($(p4)-(s)-(s2)$);
\fill[red!80!black] (0.3,1.7) circle (0pt) node {{\scriptsize $v$}};
\fill[red!80!black] (-0.3,2.5) circle (0pt) node {{\scriptsize $u$}};
\fill[red!80!black] (1.8,3.2) circle (0pt) node {{\scriptsize $u$}};
\fill[color=blue!50!black] (1,0.87) circle (3pt) node[above] (0up) { {\scriptsize$\varphi'$} };
\fill[color=blue!50!black] (1,2.53) circle (3pt) node[below] (0up) { {\scriptsize$\psi'$} };
\fill[blue!50!black] (1.6,2.65) node[line width=0pt] (ore) {{\tiny$X'$}};
\fill[blue!50!black] (0.4,2.65) node[line width=0pt] (ore) {{\tiny$Y'$}};
\fill[blue!50!black] (1.65,0.65) node[line width=0pt] (ore) {{\tiny$X'$}};
\fill[blue!50!black] (0.4,0.65) node[line width=0pt] (ore) {{\tiny$Y'$}};
\fill[black] (1,0.2) node[line width=0pt] (ore) {{\tiny$b\vphantom{b}$}};
\fill[black] (-0.3,0.2) node[line width=0pt] (ore) {{\tiny$a\vphantom{b}$}};
\end{tikzpicture}
& =
\begin{tikzpicture}[thick,scale=0.75,color=black, baseline=0cm]
\coordinate (p1) at (0,0);
\coordinate (p2) at (0.5,0.2);
\coordinate (p3) at (-0.5,0.4);
\coordinate (p4) at (2.5,0);
\coordinate (p5) at (1.5,0.2);
\coordinate (p6) at (2.0,0.4);
\coordinate (s) at ($(0,1.5)$);
\coordinate (s2) at ($(0,0)$);
\coordinate (d) at ($(0,0)$);
\coordinate (h) at ($(0,2.5)$);
\coordinate (t) at ($(0.1,0)$);
\coordinate (u) at ($(0,-1.1)$);
\fill[black] (2.3,-1.3) node[line width=0pt] (ore) {{\tiny$b$}};
%
\fill [orange!20!white, opacity=0.8] 
($(p2)+(d)-(s)-(s2)$) .. controls +(0,0.1) and +(0.5,0) ..  ($(p3)+(d)-(s)-(s2)$)
-- ($(p3)+(d)-(s)-(s2)$) -- ($(p3)+(d)+(s)$) 
-- ($(p3)+(d)+(s)$) .. controls +(0.5,0) and +(0,0.1) ..($(p2)+(d)+(s)$); 
%
\draw[thin] ($(p2)+(d)-(s)-(s2)$) .. controls +(0,0.1) and +(0.5,0) ..  ($(p3)+(d)-(s)-(s2)$); 
\draw[thin] ($(p2)+(d)+(s)$) .. controls +(0,0.1) and +(0.5,0) ..  ($(p3)+(d)+(s)$); 
%
\fill [orange!40!white, opacity=0.8] 
($(p1)+(d)-(s)-(s2)$) .. controls +(0.3,0) and +(0,-0.1) .. ($(p2)+(d)-(s)-(s2)$)
-- ($(p2)+(d)-(s)-(s2)$) -- ($(p2)+(d)+(s)$)
-- ($(p2)+(d)+(s)$) .. controls +(0,-0.1) and +(0.3,0) .. ($(p1)+(d)+(s)$);
%
\draw[thin] ($(p1)+(d)-(s)-(s2)$) .. controls +(0.3,0) and +(0,-0.1) ..  ($(p2)+(d)-(s)-(s2)$); 
\draw[thin] ($(p1)+(d)+(s)$) .. controls +(0.3,0) and +(0,-0.1) ..  ($(p2)+(d)+(s)$); 
%
\draw[thin] ($(p1)+(d)-(s)-(s2)$) -- ($(p1)+(d)+(s)$);
\draw[thin] ($(p3)+(d)-(s)-(s2)$) -- ($(p3)+(d)+(s)$);
\draw[very thick, color=blue!50!black] ($(p2)+(d)-(s)-(s2)$) -- ($(p2)+(d)+(s)$);
%
%
\fill [orange!20!white, opacity=0.8] 
($(p5)-(d)-(s)-(s2)$) .. controls +(0,0.1) and +(-0.3,0) .. ($(p6)-(d)-(s)-(s2)$)
-- ($(p6)-(d)-(s)-(s2)$) -- ($(p6)-(d)+(s)$)
-- ($(p6)-(d)+(s)$) .. controls +(-0.3,0) and +(0,0.1) .. ($(p5)-(d)+(s)$);
%
%
\draw[thin] ($(p6)-(d)-(s)-(s2)$) -- ($(p6)-(d)+(s)$);
%
\draw[thin] ($(p5)-(d)-(s)-(s2)$) .. controls +(0,0.1) and +(-0.3,0) ..  ($(p6)-(d)-(s)-(s2)$); 
\draw[thin] ($(p5)-(d)+(s)$) .. controls +(0,0.1) and +(-0.3,0) ..  ($(p6)-(d)+(s)$); 
colouring front: 
\fill [orange!40!white, opacity=0.8] 
($(p5)-(d)-(s)-(s2)$) .. controls +(0,-0.1) and +(-0.5,0) ..  ($(p4)-(d)-(s)-(s2)$)
-- ($(p4)-(d)-(s)-(s2)$) -- ($(p4)-(d)+(s)$) 
-- ($(p4)-(d)+(s)$) .. controls +(-0.5,0) and +(0,-0.1) ..($(p5)-(d)+(s)$); 
%
\draw[thin] ($(p4)-(d)+(s)$) .. controls +(-0.5,0) and +(0,-0.1) ..($(p5)-(d)+(s)$);
\draw[thin] ($(p4)-(d)-(s)-(s2)$) .. controls +(-0.5,0) and +(0,-0.1) ..($(p5)-(d)-(s)-(s2)$);
%
\draw[thin] ($(p4)-(d)-(s)-(s2)$) -- ($(p4)-(d)+(s)$);
\draw[very thick, color=blue!50!black] ($(p5)-(d)-(s)-(s2)$) -- ($(p5)-(d)+(s)$);
\fill[red!80!black] (0.25,1.0) circle (0pt) node {{\scriptsize $v$}};
\fill[red!80!black] (-0.25,1.7) circle (0pt) node {{\scriptsize $u$}};
\fill[red!80!black] (1.8,1.75) circle (0pt) node {{\scriptsize $u$}};
\fill[red!80!black] (2.25,1.0) circle (0pt) node {{\scriptsize $v$}};
\fill[blue!50!black] (1.25,0.3) node[line width=0pt] (ore) {{\tiny$X'$}};
\fill[blue!50!black] (0.75,0.3) node[line width=0pt] (ore) {{\tiny$Y'$}};
\fill[black] (1,-1.3) node[line width=0pt] (ore) {{\tiny$b\vphantom{b}$}};
\fill[black] (-0.3,-1.3) node[line width=0pt] (ore) {{\tiny$a\vphantom{b}$}};
\end{tikzpicture}
= 1_{1_{Y'\otimes X'}}
\, , 
&
\begin{tikzpicture}[thick,scale=0.75,color=black, baseline=0cm]
\coordinate (p1) at (0,0);
\coordinate (p2) at (0.5,0.2);
\coordinate (p3) at (-0.5,0.4);
\coordinate (p4) at (2.5,0);
\coordinate (p5) at (1.5,0.2);
\coordinate (p6) at (2.0,0.4);
\coordinate (s) at ($(0,1.5)$);
\coordinate (s2) at ($(0,0)$);
\coordinate (d) at ($(4,0)$);
\coordinate (h) at ($(0,2.5)$);
\coordinate (t) at ($(0.1,0)$);
\coordinate (u) at ($(0,-1.1)$);
%
%
\fill [orange!20!white, opacity=0.8] 
($(p3)-(s2)$) .. controls +(0.5,0) and +(0,0.1) ..  ($(p2)-(s2)$)
-- ($(p2)-(s2)$) .. controls +(0,-0.9) and +(0,-0.9) ..  ($(p5)-(s2)$)
-- ($(p5)-(s2)$) .. controls +(0,0.1) and +(-0.3,0) .. ($(p6)-(s2)$)
-- ($(p6)-(s2)$) -- ($(p6)-(s)-(s2)$) -- ($(p3)-(s)-(s2)$);
%
\draw[thin] ($(p6)-(s)-(s2)$) -- ($(p3)-(s)-(s2)$);	
\draw[thin] ($(p6)-(s2)$) -- ($(p6)-(s)-(s2)$);
%
\fill [orange!40!white, opacity=0.8] 
($(p1)-(s2)$) .. controls +(0.3,0) and +(0,-0.1) ..  ($(p2)-(s2)$)
-- ($(p2)-(s2)$) .. controls +(0,-0.9) and +(0,-0.9) ..  ($(p5)-(s2)$)
-- ($(p5)-(s2)$) .. controls +(0,-0.1) and +(-0.5,0) ..  ($(p4)-(s2)$) 
-- ($(p4)-(s2)$) -- ($(p4)-(s)-(s2)$) -- ($(p1)-(s)-(s2)$);
%
\draw[thin] ($(p4)-(s)-(s2)$) -- ($(p1)-(s)-(s2)$); 
%
\draw[very thick, color=blue!50!black] ($(p2)-(s2)$) .. controls +(0,-0.9) and +(0,-0.9) ..  ($(p5)-(s2)$);
\draw[thin] ($(p1)-(s2)$) -- ($(p1)-(s)-(s2)$);
\draw[thin] ($(p3)-(s2)$) -- ($(p3)-(s)-(s2)$);
\draw[thin] ($(p4)-(s2)$) -- ($(p4)-(s)-(s2)$);
%
%
\fill [orange!20!white, opacity=0.8] 
(p3) .. controls +(0.5,0) and +(0,0.1) ..  (p2)
-- (p2) .. controls +(0,0.9) and +(0,0.9) ..  (p5)
-- (p5) .. controls +(0,0.1) and +(-0.3,0) .. (p6)
-- (p6) -- ($(p6)+(s)$) -- ($(p3)+(s)$);
%
\draw[thin] (p6) -- ($(p6)+(s)$);
%
\fill [orange!40!white, opacity=0.8] 
(p1) .. controls +(0.3,0) and +(0,-0.1) ..  (p2)
-- (p2) .. controls +(0,0.9) and +(0,0.9) ..  (p5)
-- (p5) .. controls +(0,-0.1) and +(-0.5,0) ..  (p4) 
-- (p4) -- ($(p4)+(s)$) -- ($(p1)+(s)$);
%
\draw[thin] ($(p4)+(s)$) -- ($(p1)+(s)$);
\draw[thin] ($(p6)+(s)$) -- ($(p3)+(s)$);
%
\draw[very thick, color=blue!50!black] (p2) .. controls +(0,0.9) and +(0,0.9) ..  (p5);
\draw[very thick, color=blue!50!black] ($(p2)-(s2)$) .. controls +(0,-0.9) and +(0,-0.9) ..  ($(p5)-(s2)$);
\draw[thin] (p1) -- ($(p1)+(s)$);
\draw[thin] (p3) -- ($(p3)+(s)$);
\draw[thin] (p4) -- ($(p4)+(s)$);
\fill[color=blue!50!black] (1,0.87) circle (3pt) node[above] (0up) { {\scriptsize$\varphi'$} };
\fill[color=blue!50!black] (1,-0.46) circle (3pt) node[below] (0up) { {\scriptsize$\psi'$} };
\fill[blue!50!black] (1.75,0.25) node[line width=0pt] (ore) {{\tiny$X'$}};
\fill[blue!50!black] (0.25,0.25) node[line width=0pt] (ore) {{\tiny$Y'$}};
\fill[red!80!black] (0.25,1.3) circle (0pt) node {{\scriptsize $v$}};
\fill[red!80!black] (-0.3,1.7) circle (0pt) node {{\scriptsize $u$}};
\fill[black] (0.2,-1.65) node[line width=0pt] (ore) {{\tiny$b$}};
\fill[black] (-0.25,-1.3) node[line width=0pt] (ore) {{\tiny$a$}};
\fill[black] (-0.65,-0.95) node[line width=0pt] (ore) {{\tiny$b$}};
\end{tikzpicture}
&
=
\begin{tikzpicture}[thick,scale=0.75,color=black, baseline=0cm]
\coordinate (p1) at (0,0);
\coordinate (p2) at (0.5,0.2);
\coordinate (p3) at (-0.5,0.4);
\coordinate (p4) at (1.5,0);
\coordinate (p5) at (0.5,0.2);
\coordinate (p6) at (1.0,0.4);
\coordinate (s) at ($(0,1.5)$);
\coordinate (s2) at ($(0,0)$);
\coordinate (d) at ($(4,0)$);
\coordinate (h) at ($(0,2.5)$);
\coordinate (t) at ($(0.1,0)$);
\coordinate (u) at ($(0,-1.1)$);
%
\fill [orange!20!white, opacity=0.8] 
($(p3)+(s)$) -- ($(p6)+(s)$) -- ($(p6)-(s)$) -- ($(p3)-(s)$);
\draw[thin] ($(p6)-(s)-(s2)$) -- ($(p3)-(s)-(s2)$);	
\draw[thin] ($(p6)-(s2)$) -- ($(p6)-(s)-(s2)$);
\draw[thin] (p6) -- ($(p6)+(s)$);
\fill [orange!40!white, opacity=0.8] 
($(p1)+(s)$) -- ($(p4)+(s)$) -- ($(p4)-(s)$) -- ($(p1)-(s)$);
%
%
\draw[thin] ($(p4)-(s)-(s2)$) -- ($(p1)-(s)-(s2)$); 
%
\draw[thin] ($(p1)-(s2)$) -- ($(p1)-(s)-(s2)$);
\draw[thin] ($(p3)-(s2)$) -- ($(p3)-(s)-(s2)$);
\draw[thin] ($(p4)-(s2)$) -- ($(p4)-(s)-(s2)$);
\draw[thin] ($(p4)+(s)$) -- ($(p1)+(s)$);
\draw[thin] ($(p6)+(s)$) -- ($(p3)+(s)$);
%
\draw[thin] (p1) -- ($(p1)+(s)$);
\draw[thin] (p3) -- ($(p3)+(s)$);
\draw[thin] (p4) -- ($(p4)+(s)$);
\fill[red!80!black] (0.25,1.3) circle (0pt) node {{\scriptsize $v$}};
\fill[red!80!black] (-0.3,1.7) circle (0pt) node {{\scriptsize $u$}};
\fill[black] (0.2,-1.65) node[line width=0pt] (ore) {{\tiny$b$}};
\fill[black] (-0.25,-1.3) node[line width=0pt] (ore) {{\tiny$a$}};
\fill[black] (-0.65,-0.95) node[line width=0pt] (ore) {{\tiny$b$}};
\end{tikzpicture}
= 1_{1_{1_{u\btimes v}}}
\, , 
\\
\begin{tikzpicture}[thick,scale=0.75,color=black, baseline=-1.05cm]
%
\coordinate (p1) at (0,0);
\coordinate (p2) at (1,0);
\fill [orange!20!white, opacity=0.8] 
(p1) .. controls +(0,-0.5) and +(0,-0.5) ..  (p2)
-- (p2) .. controls +(0,0.5) and +(0,0.5) ..  (p1)
;
\fill [orange!40!white, opacity=0.8] 
(p1) .. controls +(0,-0.5) and +(0,-0.5) ..  (p2)
-- (p2) .. controls +(0,-1.2) and +(0,-1.2) ..  (p1)
;
\draw[very thick, color=blue!50!black] (p1) .. controls +(0,-1.2) and +(0,-1.2) ..  (p2); 
\draw[thin] (p1) .. controls +(0,0.5) and +(0,0.5) ..  (p2); 
\draw[thin] (p1) .. controls +(0,-0.5) and +(0,-0.5) ..  (p2); 
%
\fill[red!80!black] (0.5,-0.35) circle (0pt) node[below] {{\scriptsize $v$}};
\fill[red!80!black] (0.5,0.3) circle (0pt) node[below] {{\scriptsize $u$}};
%
%
\coordinate (p1) at (0,-2.5);
\coordinate (p2) at (1,-2.5);
\fill [orange!20!white, opacity=0.8] 
(p1) .. controls +(0,-0.5) and +(0,-0.5) ..  (p2)
-- (p2) .. controls +(0,0.5) and +(0,0.5) ..  (p1);
\draw[thin] (p1) .. controls +(0,0.5) and +(0,0.5) ..  (p2); 
\fill[red!80!black] (0.5,-1.7) circle (0pt) node[below] {{\scriptsize $u$}};
%
\fill [orange!40!white, opacity=0.8] 
(p1) .. controls +(0,-0.5) and +(0,-0.5) ..  (p2)
-- (p2) .. controls +(0,1.2) and +(0,1.2) ..  (p1);
\draw[thin] (p1) .. controls +(0,-0.5) and +(0,-0.5) ..  (p2); 
\draw[very thick, color=blue!50!black] (p1) .. controls +(0,1.2) and +(0,1.2) ..  (p2); 
%
\fill[red!80!black] (0.5,-2.2) circle (0pt) node[below] {{\scriptsize $v$}};
\fill[color=blue!50!black] (0.5,-0.89) circle (3pt) node[right] (0up) { };
\fill[color=blue!50!black] (0.45,-1.05) circle (0pt) node[right] (0up) { {\scriptsize$\xi'$} };
\fill[color=blue!50!black] (0.5,-1.6) circle (3pt) node[left] (0up) { };
\fill[color=blue!50!black] (0.53,-1.45) circle (0pt) node[left] (0up) { {\scriptsize$\zeta'$} };
\fill[blue!50!black] (-0.1,-0.6) node[line width=0pt] (ore) {{\tiny$X'$}};
\fill[blue!50!black] (1.15,-0.6) node[line width=0pt] (ore) {{\tiny$Y'$}};
\fill[blue!50!black] (-0.1,-2) node[line width=0pt] (ore) {{\tiny$X'$}};
\fill[blue!50!black] (1.15,-2) node[line width=0pt] (ore) {{\tiny$Y'$}};
\fill[black] (-0.15,-1.2) node[line width=0pt] (ore) {{\tiny$b$}};
\end{tikzpicture}
&=
\begin{tikzpicture}[thick,scale=0.75,color=black, baseline=-1.05cm]
%
\coordinate (p1) at (0,0);
\coordinate (p2) at (1,0);
\fill [orange!20!white, opacity=0.8] 
(p1) .. controls +(0,-0.5) and +(0,-0.5) ..  (p2)
-- (p2) .. controls +(0,0.5) and +(0,0.5) ..  (p1);
%
\fill[red!80!black] (0.5,0.3) circle (0pt) node[below] {{\scriptsize $u$}};
%
%
\coordinate (q1) at (0,-2.5);
\coordinate (q2) at (1,-2.5);
\fill [orange!20!white, opacity=0.8] 
(q1) .. controls +(0,-0.5) and +(0,-0.5) ..  (q2)
-- (q2) .. controls +(0,0.5) and +(0,0.5) ..  (q1);
\draw[thin] (q1) .. controls +(0,0.5) and +(0,0.5) ..  (q2); 
%
\fill [orange!40!white, opacity=0.8] 
(q1) .. controls +(0,-0.5) and +(0,-0.5) ..  (q2)
-- (p2) .. controls +(0,-0.5) and +(0,-0.5) ..  (p1);
\draw[thin] (q1) .. controls +(0,-0.5) and +(0,-0.5) ..  (q2); 
\draw[thin] (p1) .. controls +(0,-0.5) and +(0,-0.5) ..  (p2); 
\draw[thin] (p1) .. controls +(0,+0.5) and +(0,+0.5) ..  (p2); 
\draw[very thick, color=blue!50!black] (q1) --  (p1); 
\draw[very thick, color=blue!50!black] (q2) --  (p2); 
%
\fill[red!80!black] (0.5,-1) circle (0pt) node[below] {{\scriptsize $v$}};
\fill[blue!50!black] (-0.15,-0.6) node[line width=0pt] (ore) {{\tiny$X'$}};
\fill[blue!50!black] (1.25,-0.6) node[line width=0pt] (ore) {{\tiny$Y'$}};
\fill[black] (-0.2,-1.2) node[line width=0pt] (ore) {{\tiny$b$}};
\end{tikzpicture}
= 
1_{X'\otimes Y'}
\, , 
& 
\begin{tikzpicture}[thick,scale=0.75,color=black, baseline=-0.1cm]
%
\coordinate (p1) at (0,0);
\coordinate (p2) at (1,0);
%
%
%
\fill[red!80!black] (0.5,0.5) circle (0pt) node[below] {{\tiny $v$}};
%
\fill [orange!40!white, opacity=0.8] 
(p1) .. controls +(0,-1.2) and +(0,-1.2) ..  (p2)
-- (p2) .. controls +(0,1.2) and +(0,1.2) ..  (p1);
\draw[very thick, color=blue!50!black] (p1) .. controls +(0,1.2) and +(0,1.2) ..  (p2); 
\draw[very thick, color=blue!50!black] (p1) .. controls +(0,-1.2) and +(0,-1.2) ..  (p2); 
%
\fill[color=blue!50!black] (0.5,-0.89) circle (3pt) node[right] (0up) { };
\fill[color=blue!50!black] (0.45,-1.05) circle (0pt) node[right] (0up) { {\scriptsize$\xi'$} };
\fill[color=blue!50!black] (0.5,0.89) circle (3pt) node[left] (0up) { };
\fill[color=blue!50!black] (0.53,1.05) circle (0pt) node[left] (0up) { {\scriptsize$\zeta'$} };
\fill[blue!50!black] (-0.25,-0.1) node[line width=0pt] (ore) {{\tiny$X'$}};
\fill[blue!50!black] (1.25,-0.1) node[line width=0pt] (ore) {{\tiny$Y'$}};
\fill[red!80!black] (0.5,0) circle (0pt) node[below] {{\scriptsize $v$}};
\fill[black] (-0.2,-0.7) node[line width=0pt] (ore) {{\tiny$b$}};
\end{tikzpicture}
&
=
1_{1_{1_b}}
\label{eq:3dBubble2}
\, .
\end{align}
Hence~$u$ and~$v$ are \textsl{weakly inverse} to one another, in the sense that their compositions are equivalent
(in the respective Hom 2-categories) to identity 1-morphisms. 

A weaker notion is that of a \textsl{(3-)condensation of~$a$ onto~$b$}, which consists only of the morphisms $X',Y',\zeta',\xi'$ in~\eqref{eq:table} such that they witness a 2-condensation of $u\btimes v$ onto~$1_b$, cf.\ Section~\ref{subsec:defect2categories}. 
Concretely, this means that only the bubble condition $\zeta'\circ\xi' = 1_{1_{1_b}}$ must hold. 
If moreover~$b$ also condenses onto~$a$, i.e.\ if also $\zeta\circ\xi=1_{1_{1_a}}$, then~$a$ and~$b$ are equivalent in the condensation completion $\textrm{Kar}(\mathcal T)$ of \cite{GaiottoJohnsonFreyd}. 
Analogously to the 2-dimensional case, it is expected that this translates into an equivalence of \textsl{framed} TFTs. 

If the 3-category~$\mathcal T$ comes with the structure of a Gray category with duals, the variant relevant for \textsl{oriented} TFTs is as follows. 
Two objects $a,b\in\mathcal T$ are \textsl{orbifold equivalent} if there is a 1-morphism $u\colon a\to b$ such that $u$-labelled spheres evaluate to the identity:\footnote{These are precisely the {bubble identities} in~\eqref{eq:3dBubble1} and~\eqref{eq:3dBubble2} if $v=u^\#$ and if the 1- and 2-morphisms in~\eqref{eq:3dCalcExample} are the associated duality data, see e.g.\ \cite{Carqueville:2023qrk} for details.}
\begin{equation}
\label{eq:spheres}
\begin{tikzpicture}[thick,scale=0.6, color=green!50!black, baseline=0cm]
\fill[ball color=orange!40!white] (0,0) circle (1.55);
\draw[line width=0pt] (0.5,0.4) node[red!80!black] (D) {{\tiny$u´$}};
\fill[orange!30!white, opacity=0.7] (0,0) circle (1.55);
\draw[line width=0pt] (-0.2,-0.4) node[red!80!black] (D) {{\footnotesize$u^\#$}};
\draw[line width=0pt] (1.3,-1.3) node[black] (D) {{\tiny$b$}};
\end{tikzpicture}
= 
\dim_{\textrm{l}}(\textrm{coev}_u) \stackrel{!}{=} 1_{1_{1_b}}
\, , \quad 
\begin{tikzpicture}[thick,scale=0.6, color=green!50!black, baseline=0cm]
\fill[ball color=orange!20!white] (0,0) circle (1.55);
\draw[line width=0pt] (0.5,0.4) node[red!80!black] (D) {{\tiny$u^\#$}};
\fill[orange!25!white, opacity=0.7] (0,0) circle (1.55);
\draw[line width=0pt] (-0.2,-0.4) node[red!80!black] (D) {{\footnotesize$u$}};
\draw[line width=0pt] (1.3,-1.3) node[black] (D) {{\tiny$a$}};
\end{tikzpicture}
= 
\dim_{\textrm{r}}(\textrm{ev}_u) \stackrel{!}{=} 1_{1_{1_a}}
\, . 
\end{equation}
This corresponds to choosing the data in~\eqref{eq:table} to be appropriate adjunction morphisms in~$\mathcal T$. 
Moreover, one checks that
\begin{equation}
\label{eq:3d-orb-datum}
\mathcal A = 
\left( 
a, \;
A := 
\begin{tikzpicture}[thick,scale=0.75,color=black, baseline=0cm]
\coordinate (p1) at (0,0);
\coordinate (p2) at (0.5,0.2);
\coordinate (p3) at (-0.5,0.4);
\coordinate (p4) at (1.5,0);
\coordinate (p5) at (0.5,0.2);
\coordinate (p6) at (1.0,0.4);
\coordinate (s) at ($(0,1)$);
\coordinate (s2) at ($(0,0)$);
\coordinate (d) at ($(4,0)$);
\coordinate (h) at ($(0,2.5)$);
\coordinate (t) at ($(0.1,0)$);
\coordinate (u) at ($(0,-1.1)$);
%
\fill [orange!40!white, opacity=0.8] 
($(p3)+(s)$) -- ($(p6)+(s)$) -- ($(p6)-(s)$) -- ($(p3)-(s)$);
\draw[thin] ($(p6)-(s)-(s2)$) -- ($(p3)-(s)-(s2)$);	
\draw[thin] ($(p6)-(s2)$) -- ($(p6)-(s)-(s2)$);
\draw[thin] (p6) -- ($(p6)+(s)$);
\fill [orange!20!white, opacity=0.8] 
($(p1)+(s)$) -- ($(p4)+(s)$) -- ($(p4)-(s)$) -- ($(p1)-(s)$);
%
%
\draw[thin] ($(p4)-(s)-(s2)$) -- ($(p1)-(s)-(s2)$); 
%
\draw[thin] ($(p1)-(s2)$) -- ($(p1)-(s)-(s2)$);
\draw[thin] ($(p3)-(s2)$) -- ($(p3)-(s)-(s2)$);
\draw[thin] ($(p4)-(s2)$) -- ($(p4)-(s)-(s2)$);
\draw[thin] ($(p4)+(s)$) -- ($(p1)+(s)$);
\draw[thin] ($(p6)+(s)$) -- ($(p3)+(s)$);
%
\draw[thin] (p1) -- ($(p1)+(s)$);
\draw[thin] (p3) -- ($(p3)+(s)$);
\draw[thin] (p4) -- ($(p4)+(s)$);
\fill[red!80!black] (0.25,0.8) circle (0pt) node {{\scriptsize $u$}};
\fill[red!80!black] (-0.15,1.2) circle (0pt) node {{\scriptsize $u^\#$}};
\fill[black] (0.2,-1.15) node[line width=0pt] (ore) {{\tiny$a$}};
\fill[black] (-0.25,-0.8) node[line width=0pt] (ore) {{\tiny$b$}};
\fill[black] (-0.65,-0.45) node[line width=0pt] (ore) {{\tiny$a$}};
\end{tikzpicture}
\, , \;
\mu := 
\begin{tikzpicture}[thick,scale=0.75,color=black, baseline=0cm]
\coordinate (p1) at (0,0);
\coordinate (p2) at (0.5,0.2);
\coordinate (p3) at (-0.5,0.4);
\coordinate (p4) at (2.5,0);
\coordinate (p5) at (1.5,0.2);
\coordinate (p6) at (2.0,0.4);
\coordinate (q1) at (0,0.6);
\coordinate (q2) at (1.8,0.6);
\coordinate (s) at ($(0,1)$);
\coordinate (s2) at ($(0,0)$);
\coordinate (d) at ($(0,0)$);
\coordinate (h) at ($(0,2.5)$);
\coordinate (t) at ($(0.1,0)$);
\coordinate (u) at ($(0,-1.1)$);
\coordinate (y) at ($(1,-0.8)$);
%
\fill [orange!40!white, opacity=0.8] 
($(q1)+(s)$) -- ($(q2)+(s)$) -- ($(q2)-(s)$) -- ($(q1)-(s)$);
\draw[thin] ($(q1)+(s)$) -- ($(q2)+(s)$) -- ($(q2)-(s)$) -- ($(q1)-(s)$) -- ($(q1)+(s)$);
%
\fill [orange!40!white, opacity=0.8] 
($(p5)-(d)-(s)-(s2)$) .. controls +(0,0.1) and +(-0.3,0) .. ($(p6)-(d)-(s)-(s2)$)
-- ($(p6)-(d)-(s)-(s2)$) -- ($(p6)-(d)+(s)$)
-- ($(p6)-(d)+(s)$) .. controls +(-0.3,0) and +(0,0.1) .. ($(p5)-(d)+(s)$);
%
%
\draw[thin] ($(p6)-(d)-(s)-(s2)$) -- ($(p6)-(d)+(s)$);
%
\draw[thin] ($(p5)-(d)-(s)-(s2)$) .. controls +(0,0.1) and +(-0.3,0) ..  ($(p6)-(d)-(s)-(s2)$); 
\draw[thin] ($(p5)-(d)+(s)$) .. controls +(0,0.1) and +(-0.3,0) ..  ($(p6)-(d)+(s)$); 
\fill[black] (2.2,-0.8) node[line width=0pt] (ore) {{\tiny$a$}};
%
\fill [orange!40!white, opacity=0.8] 
($(p5)-(d)-(s)-(s2)$) .. controls +(0,-0.1) and +(-0.5,0) ..  ($(p4)-(d)-(s)-(s2)$)
-- ($(p4)-(d)-(s)-(s2)$) -- ($(p4)-(d)+(s)$) 
-- ($(p4)-(d)+(s)$) .. controls +(-0.5,0) and +(0,-0.1) ..($(p5)-(d)+(s)$); 
\fill[black] (2.2,-0.8) node[line width=0pt] (ore) {{\tiny$a$}};
%
\fill [orange!40!white, opacity=0.8] 
($(q1)+(s)+(y)$) -- ($(q2)+(s)+(y)$) -- ($(q2)-(s)+(y)$) -- ($(q1)-(s)+(y)$);
%
\draw[thin] ($(p4)-(d)+(s)$) .. controls +(-0.5,0) and +(0,-0.1) ..($(p5)-(d)+(s)$);
\draw[thin] ($(p4)-(d)-(s)-(s2)$) .. controls +(-0.5,0) and +(0,-0.1) ..($(p5)-(d)-(s)-(s2)$);
%
\draw[thin] ($(p4)-(d)-(s)-(s2)$) -- ($(p4)-(d)+(s)$);
\draw[thick, black] ($(p5)-(d)-(s)-(s2)$) -- ($(p5)-(d)+(s)$);
\fill[red!80!black] (1.8,1.2) circle (0pt) node {{\scriptsize $u$}};
\fill[red!80!black] (1.25,-0.9) circle (0pt) node {{\scriptsize $u$}};
\fill[red!80!black] (2.25,0.5) circle (0pt) node {{\scriptsize $u^\#$}};
\fill[black] (0.7,-0.8) node[line width=0pt] (ore) {{\tiny$b\vphantom{b}$}};
\fill[black] (2.2,-1.35) node[line width=0pt] (ore) {{\tiny$a$}};
%
\fill [orange!20!white, opacity=0.8] 
($(q1)+(s)+(y)$) -- ($(q2)+(s)+(y)$) -- ($(q2)-(s)+(y)$) -- ($(q1)-(s)+(y)$);
\draw[thin] ($(q1)+(s)+(y)$) -- ($(q2)+(s)+(y)$) -- ($(q2)-(s)+(y)$) -- ($(q1)-(s)+(y)$) -- ($(q1)+(s)+(y)$);
\fill[red!80!black] (1.25,-0.9) circle (0pt) node {{\scriptsize $u$}};
\fill[red!80!black] (0.4,1.35) circle (0pt) node {{\scriptsize $u^\#$}};
\end{tikzpicture}
\,, \; 
\alpha := 
\begin{tikzpicture}[thick,scale=0.75,color=black, baseline=-0.2cm]
\coordinate (p1) at (0,0);
\coordinate (p2) at (0.5,0.2);
\coordinate (p3) at (-0.5,0.4);
\coordinate (p4) at (2.5,0);
\coordinate (p5) at (1.5,0.2);
\coordinate (p6) at (2.0,0.4);
\coordinate (q1) at (0,0.6);
\coordinate (q2) at (1.8,0.6);
\coordinate (s) at ($(0,1)$);
\coordinate (s2) at ($(0,0)$);
\coordinate (d) at ($(0,0)$);
\coordinate (h) at ($(0,2.5)$);
\coordinate (t) at ($(0.1,0)$);
\coordinate (u) at ($(0,-1.1)$);
\coordinate (y) at ($(1,-0.8)$);
\coordinate (w) at ($(0.75,-0.6)$);
\coordinate (v) at ($(0.75,-0.6)$);
%
\fill [orange!40!white, opacity=0.8] 
($(q1)+(s)$) -- ($(q2)+(s)$) -- ($(q2)-(s)$) -- ($(q1)-(s)$);
\draw[thin] ($(q1)+(s)$) -- ($(q2)+(s)$) -- ($(q2)-(s)$) -- ($(q1)-(s)$) -- ($(q1)+(s)$);
%
\fill [orange!20!white, opacity=0.8] 
($(p5)-(d)-(s)-(s2)-(1,0)$) .. controls +(0,0.1) and +(-0.3,0) .. ($(p6)-(d)-(s)-(s2)$)
-- ($(p6)-(d)-(s)-(s2)$) -- ($(p6)-(d)+(s)$)
-- ($(p6)-(d)+(s)$) .. controls +(-0.3,0) and +(0,0.1) .. ($(p5)-(d)+(s)$);
%
%
\draw[thin] ($(p6)-(d)-(s)-(s2)$) -- ($(p6)-(d)+(s)$);
%
\draw[thin] ($(p5)-(d)-(s)-(s2)-(1,0)$) .. controls +(0,0.1) and +(-0.3,0) ..  ($(p6)-(d)-(s)-(s2)$); 
\draw[thin] ($(p5)-(d)+(s)$) .. controls +(0,0.1) and +(-0.3,0) ..  ($(p6)-(d)+(s)$); 
\fill[black] (2.2,-0.8) node[line width=0pt] (ore) {{\tiny$a$}};
%
\fill [orange!40!white, opacity=0.8] 
($(p5)-(d)-(s)-(s2)-(1,0)$) .. controls +(0,-0.1) and +(-0.5,0) ..  ($(p4)-(d)-(s)-(s2)$)
-- ($(p4)-(d)-(s)-(s2)$) -- ($(p4)-(d)+(s)$) 
-- ($(p4)-(d)+(s)$) .. controls +(-0.5,0) and +(0,-0.1) ..($(p5)-(d)+(s)$); 
\fill[black] (2.2,-0.8) node[line width=0pt] (ore) {{\tiny$a$}};
%
\draw[thin] ($(p4)-(d)+(s)$) .. controls +(-0.5,0) and +(0,-0.1) ..($(p5)-(d)+(s)$);
\draw[thin] ($(p4)-(d)-(s)-(s2)$) .. controls +(-0.5,0) and +(0,-0.1) ..($(p5)-(d)-(s)-(s2)-(1,0)$);
%
\draw[thin] ($(p4)-(d)-(s)-(s2)$) -- ($(p4)-(d)+(s)$);
\draw[thick, black] ($(p5)-(d)-(s)-(s2)-(1,0)$) -- ($(p5)-(d)+(s)$);
%
\fill [orange!20!white, opacity=0.8] 
($(p5)-(d)-(s)-(s2)+(w)$) .. controls +(0,0.1) and +(-0.3,0) .. ($(p6)-(d)-(s)-(s2)+(w)$)
-- ($(p6)-(d)-(s)-(s2)+(w)$) -- ($(p6)-(d)+(s)+(w)$)
-- ($(p6)-(d)+(s)+(w)$) .. controls +(-0.3,0) and +(0,0.1) .. ($(p5)-(d)+(s)+(w)-(1.2,0)$);
%
%
\draw[thin] ($(p6)-(d)-(s)-(s2)+(w)$) -- ($(p6)-(d)+(s)+(w)$);
%
\draw[thin] ($(p5)-(d)-(s)-(s2)+(w)$) .. controls +(0,0.1) and +(-0.3,0) ..  ($(p6)-(d)-(s)-(s2)+(w)$); 
\draw[thin] ($(p5)-(d)+(s)+(w)-(1.2,0)$) .. controls +(0,0.1) and +(-0.3,0) ..  ($(p6)-(d)+(s)+(w)$); 
%
\fill [orange!40!white, opacity=0.8] 
($(p5)-(d)-(s)-(s2)+(w)$) .. controls +(0,-0.1) and +(-0.5,0) ..  ($(p4)-(d)-(s)-(s2)+(w)$)
-- ($(p4)-(d)-(s)-(s2)+(w)$) -- ($(p4)-(d)+(s)+(w)$) 
-- ($(p4)-(d)+(s)+(w)$) .. controls +(-0.5,0) and +(0,-0.1) ..($(p5)-(d)+(s)+(w)-(1.2,0)$); 
%
\draw[thin] ($(p4)-(d)+(s)+(w)$) .. controls +(-0.5,0) and +(0,-0.1) ..($(p5)-(d)+(s)+(w)-(1.2,0)$);
\draw[thin] ($(p4)-(d)-(s)-(s2)+(w)$) .. controls +(-0.5,0) and +(0,-0.1) ..($(p5)-(d)-(s)-(s2)+(w)$);
%
\draw[thin] ($(p4)-(d)-(s)-(s2)+(w)$) -- ($(p4)-(d)+(s)+(w)$);
\draw[thick, black] ($(p5)-(d)-(s)-(s2)+(w)$) -- ($(p5)-(d)+(s)+(w)-(1.2,0)$);
%
\fill [orange!20!white, opacity=0.8] 
($(q1)+(s)+(y)+(v)+(-1.4,0)$) -- ($(q2)+(s)+(y)+(v)$) -- ($(q2)-(s)+(y)+(v)$) -- ($(q1)-(s)+(y)+(v)+(-1.4,0)$);
\draw[thin] ($(q1)+(s)+(y)+(v)+(-1.4,0)$) -- ($(q2)+(s)+(y)+(v)$) -- ($(q2)-(s)+(y)+(v)$) -- ($(q1)-(s)+(y)+(v)+(-1.4,0)$) -- ($(q1)+(s)+(y)+(v)+(-1.4,0)$);
\fill[red!80!black] (0.6,-1.55) circle (0pt) node {{\scriptsize $u$}};
\fill[red!80!black] (0.4,1.35) circle (0pt) node {{\scriptsize $u^\#$}};
\end{tikzpicture}
\;
\right) 
\end{equation}
forms an ``orbifold datum'' in~$\mathcal T$ (for this the first condition in~\eqref{eq:spheres} is sufficient) in the sense of \cite[Def.\,4.2]{Carqueville:2017aoe}. 
If $\mathcal T = \mathcal D_\zz$ is obtained from a defect TFT~$\zz$, then an orbifold datum is (by definition) a gaugeable symmetry, and the gauged theory $\zz_a/{\mathcal A}$ as defined in \cite{Carqueville:2017aoe} is an oriented TFT. 
The results of \cite{Carqueville:2023aak} imply that the TFT~$\zz_b$ associated to $b\in\mathcal D_\zz$ is isomorphic to the gauged TFT~$\zz_a/{\mathcal A}$, see \cite{Carqueville:2023qrk} for an overview.\footnote{If the evaluation of~$\zz$ on the sphere defects in~\eqref{eq:spheres} is not equal to~$1_{1_{1_b}}$ but still invertible, one finds that~$\zz_b$ stacked with an Euler TFT is isomorphic to~$\zz_a/{\mathcal A}$.}

\section{Quantum field theory in more than two dimensions}\label{sec:QFT+}

In this section we give a brief overview of how topological defects and the corresponding higher categories arise in QFTs.
While most of our examples are in the context of QFTs described by quantisation of action functionals, we stress that the characterisation of symmetries in terms of topological defects is independent the existence of a  Lagrangian description. 
There is at this point no widely accepted rigorous framework for QFTs with observables (topological or not) localised on submanifolds of arbitrary dimension, and so the considerations here are more heuristic as compared to the previous sections.
We view $d$-dimensional QFTs as black boxes which produce correlation functions on a given spacetime $X$ of dimension~$d$ with a metric of Euclidean or Minkowskian signature.\footnote{For most of this section we work in Euclidean signature and we use the terms ``topological defect'' and ``topological operator'' synonymously. 
In Minkowskian signature one must distinguish the two (for the former, the support has a time-like component, for the latter it does not), for example one obtains locality constraints, see e.g.~\cite{Shao:2023gho}.}
Correlation functions are expressions of the form
\begin{equation}\label{eq:corr}
\big\langle \mathcal O_1(x_1) \cdots  \mathcal O_n(x_n) \mathcal W_1(\gamma_1) \cdots \mathcal W_m(\gamma_m)\big\rangle_X
\end{equation}
which give the value that the QFT assigns to a collection of local operators $\mathcal O_i$ inserted at points $x_i \in X$ and non-local operators $\mathcal W_j$ supported along extended regions $\gamma_j \subset X$. 

The defect TFTs discussed in Section~\ref{sec:TFT} provide examples in which the correlation functions \eqref{eq:corr} are independent of the metric. 
The focus of this section is on non-topological examples, but 
the higher categorical structure organising topological defects in QFT is the same as that found in TFTs.
Since TFTs have a solid mathematical foundation, on the one hand they provide a useful guide to topological symmetries in general QFTs, and on the other hand they directly enter their description via $(d+1)$-dimensional symmetry TFTs which we discuss below. 
While here we use simple examples in four dimensions as illustrations, topological defects and their higher categories are a uniform feature of field theories across all dimensions.
See \cite{Cordova:2022ruw,McGreevy:2022oyu,Schafer-Nameki:2023jdn,Brennan:2023mmt,Bhardwaj:2023kri,Shao:2023gho} for a selection of reviews which cover the topics of this section (and much more).

\subsection{Invertible topological defects from $p$-forms}

A basic class of examples of invertible topological defects for $d$-dimensional QFT are the so-called \textit{$p$-form symmetries} which act on extended operators of dimension~$p$ or larger \cite{Gaiotto:2014kfa}. 
Here we discuss the case of $\textrm{U}(1)$ $p$-form symmetries. 
These occur when there is a conserved $(p+1)$-form current whose Hodge dual is exact: $d \, {\ast} j^{(p+1)} = 0$.\footnote{Identities for fields like this one are understood as being valid in all correlators \eqref{eq:corr} as long as the insertion point is away from the $x_i$ and $\gamma_j$.}
The corresponding topological operators are
\begin{equation}\label{eq:pform}
\topD{\theta}^{(p)}{(\Sigma^{d-p-1})} := \text{exp}\Big(i \theta \int_{\Sigma^{d-p-1}}\ast j^{(p+1)} \Big)
\,,
\end{equation}
where $\theta \in \mathds{R} / (2 \pi \mathds{Z})$ and $\Sigma^{d-p-1} \subset X$ is a closed oriented submanifold of dimension $d-p-1$ in the $d$-dimensional spacetime $X$.\footnote{The charges for a $\textrm{U}(1)$ symmetry are quantised and we normalise $j^{(p+1)}$ so that they are integers.}
Stokes' theorem implies the invariance of correlators under deformations of the support $\Sigma^{d-p-1}$ away from other insertions. 

The operators $\topD{\theta}^{(p)}(S^{d-p-1})$ 
measure the $\textrm{U}(1)$-charge of (possibly) non-topological defects of dimension $p$:  
Let $\gamma_j^{p}$ be the $p$-dimensional support of a non-topological operator $\mathcal W_j$, and let $B^{d-p} \subset X$ be a $(d-p)$-dimensional ball which intersects $\gamma_j^{p}$ 
transversally and in precisely one point, with $S^{d-p-1} = \partial B^{d-p}$. 
If for some $q_j \in 
\mathds{Z}$
we have the relation
\begin{equation}\label{eq:actionlocal}
\topD{\theta}^{(p)}(S^{d-p-1}) \, {\mathcal W}_j(\gamma_j^{p}) 
\,=\, \textrm{e}^{\textrm{i}\theta q_j} \, \mathcal W_j(\gamma^{p}_j) \ ,
\end{equation}
we say \textsl{$\mathcal W_j$ is of charge $q_j$}. This symmetry group is often denoted $\textrm{U}(1)^{(p)}$.

\begin{remark}\label{rem:screening}
It can happen that an operator $\mathcal W (\gamma_j^{p})$ of charge~$q$ as above ends on a $(p-1)$-dimensional operator.
This gives a generalisation of the 't~Hooft screening mechanism to $p$-form symmetries, and in this case we say that $\textrm{U}(1)^{(p)}$ is \textit{screened} to $\mathds{Z}^{(p)}_q$. 
See e.g.\ \cite{DelZotto:2015isa,Bhardwaj:2021mzl,DelZotto:2022ras} for applications of this idea to determine $p$-form symmetries of various theories without conventional Lagrangian formulations,
as well as their topological symmetry theories, which we introduce in Section \ref{sec:symTFT}.
\end{remark}

\begin{example}\label{ex:Max}\noindent\textbf{(Maxwell theory)} 
The Maxwell action on a 4-dimensional spacetime~$X$
is 
\begin{equation}\label{eq:Maxwell}
\mathcal S_{\textrm{Maxwell}} 
    = 
    \frac{1}{2 e^2} \int_X f^{(2)} \wedge \ast f^{(2)}
\ ,
\end{equation}
where the fundamental field is a $\textrm{U}(1)$-connection, locally a 1-form $a^{(1)}$ on~$X$, and $f^{(2)} = da^{(1)}$ is the field strength. 
Maxwell theory enjoys a $U(1)^{(1)}_{\textrm{e}} \times U(1)^{(1)}_{\textrm{m}}$ 1-form symmetry with 2-form currents
\begin{equation}
j^{(2)}_{\textrm{e}} = \frac{f^{(2)}}{e^2} \ , \quad\qquad j^{(2)}_{\textrm{m}} = 
\frac{\ast f^{(2)}}{2\pi} \ .
\end{equation}
Their conservation follows from the vacuum Maxwell equations
$d\,{\ast}j^{(2)}_{\textrm{e}} = \frac{1}{e^2} d\,{\ast}  f^{(2)} = 0$ and  $d\,{\ast}j^{(2)}_{\textrm{m}} = \frac{1}{2 \pi} d  f^{(2)} = 0$.\footnote{
These 1-form symmetries are spontaneously broken and the resulting massless Goldstone particles are the photons \cite{Rosenstein:1990py,Gaiotto:2014kfa}. There is a higher form version of the Coleman--Mermin--Wagner 
theorem which states that higher $p$-form symmetry cannot be spontaneously broken whenever $d-p \leqslant 2$ \cite{Gaiotto:2014kfa,Lake:2018dqm}.} 
The charged operators are Wilson lines and 't~Hooft lines. 
The Wilson lines $\mathcal W_{q_{\textrm{e}}}$ of charge $q_{\textrm{e}}$ are the holonomies of $a^{(1)}$, 
\begin{equation}
\mathcal W_{q_{\textrm{e}}}(\gamma^{1}) 
    = 
    \exp \Big( \textrm{i} q_{\textrm{e}} \int_{\gamma^{1}} a^{(1)} \Big) \,,
\end{equation}
while an 't~Hooft line $\mathcal W^D_{q_{\textrm{m}}}$ of charge $q_{\textrm{m}}$ 
is defined as a line $\gamma^{1}$ where $f^{(2)}$ becomes singular such that for an $S^2$ linking $\gamma^{1}$, we have $q_{\textrm{m}} = \frac{1}{2\pi}\int_{S^2} f^{(2)}$. 
The action of the symmetry operators justifies calling $q_{\textrm{e}}$ and $q_{\textrm{m}}$ charges:
\begin{equation}
\mathcal D^{(1)}_{\textrm{e},\theta}(S^{2}) \,\mathcal W_{q_{\textrm{e}}}(\gamma^{1}) = \textrm{e}^{\textrm{i} \theta q_{\textrm{e}}}  \mathcal W_{q_{\textrm{e}}}(\gamma^{1}) \ , \quad
\mathcal D^{(1)}_{\textrm{m},\theta}(S^{2}) \, \mathcal W^D_{q_{\textrm{m}}}(\gamma^{1}) = \textrm{e}^{\textrm{i} \theta q_{\textrm{m}} }  \mathcal W^D_{q_{\textrm{m}}}(\gamma^{1})
\ .
\end{equation}
The second identity is manifest from the definition of 't~Hooft lines. 
For the first identity, one interprets the insertion of a Wilson line as the presence of an electric current density $\ast J^{(1)}_{\textrm{e}} = \delta(\gamma^{1}) q_{\textrm{e}}$. 
The latter modifies the equation of motion to $d \, {\ast} F = \delta(\gamma^{1}) q_{\textrm{e}}$, from which the action of the symmetry operator follows.

We note that these two symmetries have a mixed 't~Hooft anomaly and hence cannot be gauged simultaneously, see e.g.\ \cite[App.\,C]{Cordova:2018cvg} for a detailed discussion.\footnote{An \textit{'t~Hooft anomaly} is a feature of the symmetry which signals that when coupled to a background, it breaks the gauge invariance of the background gauge fields. A mixed 't~Hooft anomaly between two symmetries signals that the background fields for one symmetry break the gauge invariance of the backgrounds of the other. The background gauge invariance can be restored thanks to the \textit{inflow mechanism}, which consists of coupling the $d$-dimensional QFT to a $(d+1)$-dimensional invertible TFT (known as an \textit{anomaly theory} or a \textit{symmetry protected topological phase}) -- see e.g.\ \cite[Sect.\,1]{Cordova:2019jnf} or \cite[Sect.\,3]{McGreevy:2022oyu} for reviews.\label{footnote:inflow}}
It is however possible to gauge finite subgroups of each of these symmetries separately.
\end{example}

There are more general invertible topological defects than those described in this section, corresponding to finite global symmetry (higher) groups. 
For more details on the resulting (higher) group structures and their applications in field theory, see e.g.\ \cite{Gukov:2008sn,Kapustin:2013uxa,Gaiotto:2014kfa,Barkeshli:2014cna,Tachikawa:2017gyf,Cordova:2018cvg,Benini:2018reh} as well as \cite{Schafer-Nameki:2023jdn,Bhardwaj:2023kri} for reviews.

\subsection{Identities between correlators via topological defects} 
\begin{figure}
$$
\begin{aligned}
& \Bigg\langle
\raisebox{-0.47\height}{\scalebox{.5}{\includegraphics{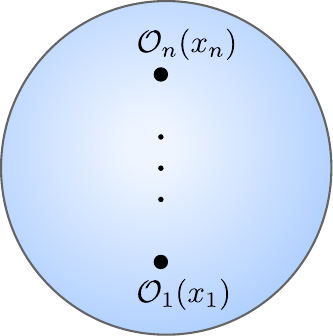}}}
\Bigg\rangle_{\!\!S^d} 
\overset{(1)}=
~
\textcolor{blue}{N^{-1}}\, 
\Bigg\langle
\raisebox{-0.47\height}{\scalebox{.5}{\includegraphics{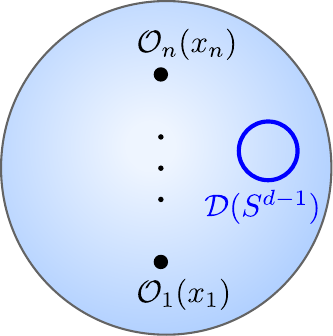}}}
\Bigg\rangle_{\!\!S^d} 
\overset{(2)}= 
~
\textcolor{blue}{N^{-1}}\, 
\Bigg\langle
\raisebox{-0.47\height}{\scalebox{.5}{\includegraphics{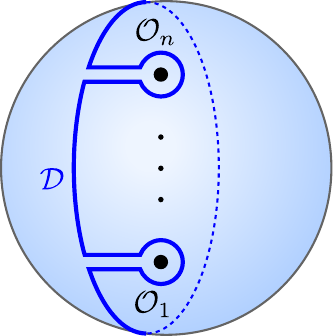}}}
\Bigg\rangle_{\!\!S^d}
\\[-1em]
&\hspace{3.5em}
\overset{(3)}= 
~
\textcolor{blue}{N^{-1}}\, 
\Bigg\langle
\raisebox{-0.47\height}{\scalebox{.5}{\includegraphics{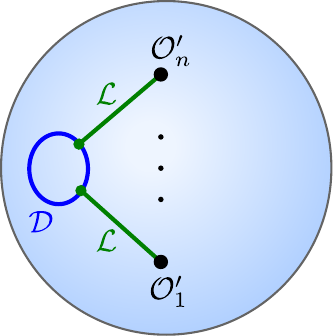}}}
\Bigg\rangle_{\!\!S^d} 
~
\overset{(4)}= 
~~~
\Bigg\langle
\raisebox{-0.47\height}{\scalebox{.5}{\includegraphics{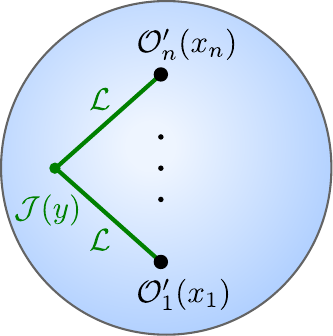}}}
\Bigg\rangle_{\!\!S^d}
\end{aligned}
$$
\caption{Example of an identity among correlators induced by a topological defect. 
Here, $\mathcal{D} = \mathcal{D}^{(0)}$ is a codimension-1 topological defect of $\mathcal{T}$
for which we assume that $\mathcal{D}(S^{d-1}) = N \cdot \one_{\mathcal T}$ for some $N\in\mathds C^\times$. 
(1) Insert $\mathcal D(S^{d-1})$ and multiply the correlators by $N^{-1}$; 
(2) expanding the bubble across $S^d$, the topological defect $\mathcal D$ wraps the supports of the operators $\mathcal O_i$ and moves to the opposite side of the sphere; (3) the topological defect shrinks on the opposite side of $S^d$, and the local operators $\mathcal{O}_i$ become twist operators $\mathcal{O}_i'$ sitting at the end of topological line defects $\mathcal L^{(d-2)}$;
(4) the remaining $\mathcal D$-sphere shrinks into a topological point defect $\mathcal{J}(y)$, into which we also absorb the prefactor $N^{-1}$.} 
\label{fig:WTid}
\end{figure}

Topological defects $\mathcal{D}^{(p)}$ of a $d$-dimensional QFT~$\mathcal{T}$ 
give identities between correlators of~$\mathcal{T}$. 
As a simple example, consider the spacetime $X=S^d$ with insertions of local point operators $\mathcal O_i(x_i)$. Suppose $\mathcal{T}$ possesses a (possibly non-invertible) codimension-1 topological defect $\mathcal{D}^{(0)}$ whose value on the sphere is invertible in the following sense:
For a sphere $S^{d-1}$ bounding a ball $B^{d} \subset X$ which does not contain any operator insertions,
$\mathcal{D}^{(0)}(S^{d-1})=N \cdot \one_{\mathcal{T}}$ 
for some $N \in \mathds{C}^\times$, and where $\one_{\mathcal{T}}$ denotes the identity field of $\mathcal{T}$. 
In other words, the defect sphere is shrunk to a point, resulting in a non-zero multiple of the identity field. 
The moves in Figure~\ref{fig:WTid} then give the identity
\begin{equation}\label{eq:corr-identity-from-defect}
\big\langle \mathcal O_1(x_1) \cdots  \mathcal O_n(x_n) \big\rangle_{S^d}
=
\big\langle 
\mathcal O'_1(x_1)\mathcal L^{(d-2)}(x_1,y) \cdots \mathcal O'_n(x_n) \mathcal L^{(d-2)}(x_n,y) \mathcal J^{(d-1)}(y) 
\big\rangle_{S^d} \ ,
\end{equation}
where the $\mathcal O'_i(x_i)$ are disorder operators sitting at the end of topological line defects $\mathcal L^{(d-2)}$ and meeting in a topological point junction $\mathcal J^{(d-1)}(y)$. 
Such identities are also referred to as generalised Ward--Takahashi identities. 

If $\mathcal{D}^{(0)}$ is an invertible topological defect, the $\mathcal{O}'_i$ on the right-hand side of \eqref{eq:corr-identity-from-defect} are again local fields, resulting in an identity between correlators of local fields
\begin{equation}
\big\langle \mathcal O_1(x_1) \cdots  \mathcal O_n(x_n) \big\rangle_{S^d}
=
\big\langle 
\mathcal O'_1(x_1) \cdots  \mathcal O'_n(x_n)\big\rangle_{S^d} \ .
\end{equation}
If furthermore $N=1$ and $\mathcal{D}^{(0)}(S^{d-1})$ acts diagonally on the $\mathcal{O}_i$ in the sense that $\mathcal{O}'_i = R_i \mathcal{O}_i$ for some $R_i \in \mathds{C}$, we obtain the selection rule that correlators can be non-zero only if $\prod_{i=1}^n R_i = 1$.

For topological defects $\mathcal{D}^{(p)}$ with $p>0$ and for which $\mathcal{D}^{(p)}(S^{d-p-1})$ is a non-zero multiple of the identity field, one can carry out manipulations of correlators analogous to Figure~\ref{fig:WTid}. Such $\mathcal{D}^{(p)}$ can move past operators of dimension less than $p$ but act on operators $\mathcal{W}_j(\gamma_j^{q})$ 
with support of dimension $q\geqslant p$ possibly generating additional topological defects attached to $\gamma_j^{q}$ in the process.

\medskip

Just as for TFTs, it is expected that $d$-dimensional QFTs~$\mathcal T$ and their codimension-$k$ topological defects form the objects and $k$-morphisms of a $d$-category. 
In this language, a topological defect $\mathcal{D}^{(p)}$ as above is a $(p+1)$-morphism, which for $p=0$ describes a codimension-1 endo-defect of the fixed bulk QFT~$\mathcal T$, and for $p>0$ is an endomorphism of the $p$-fold unit $1_{{1_{\dots}}_{1_{\mathcal T}}}$, i.e.\ of the unit $p$-morphism of the unit $(p-1)$-morphism of\dots\ the object~$\mathcal T$. 
In terms of defects this means that $\mathcal{D}^{(p)}$ is not attached to defects of lower codimension; e.g.\ for $d=4$ and $p=2$ we have that $1_{1_\mathcal T}$ is the invisible surface defect so that $\mathcal D^{(2)}$ is a line defect in~$\mathcal T$.

Write $[\mathrm{End}(1_{{1_{\dots}}_{1_{\mathcal T}}})]$ for the ring of equivalence classes of $(p+1)$-endomorphisms of $1_{{1_{\dots}}_{1_{\mathcal T}}}$, where addition is given by direct sums and multiplication by composition (and which is completed with respect to taking differences as in $K$-theory).
This describes the fusion of defects $\mathcal{D}^{(p)}$ up to equivalences induced by defects of lower codimension and is also called a \textit{fusion ring}. 
Note that by the arguments at the end of Section~\ref{subsec:DefectTFTs}, for $p>0$ this fusion ring is necessarily commutative.
Consider a non-topological dimension-$p$ operator $\mathcal{W}$, and let $A_{\mathcal{W}}$ be the $\mathds{C}$-algebra of topological point operators on $\mathcal{W}$, i.e.\ field insertions on the support $\gamma^{p}$ 
of $\mathcal{W}$ such that correlators do not depend on the precise location of the insertion. We call $\mathcal{W}$ \textsl{simple} if $A_{\mathcal{W}}$ is 1-dimensional (if $A_{\mathcal{W}}$ is semisimple and of dimension greater than one, $\mathcal{W}$ can be decomposed into a direct sum). 
If $\mathcal{W}$ is simple, by the procedure of linking spheres as in \eqref{eq:actionlocal} we obtain a 1-dimensional representation of $[\mathrm{End}(1_{{1_{\dots}}_{1_{\mathcal T}}})]$ in terms of charges.
We conclude that the charges of $p$-dimensional operators give us information about the fusion ring $[\mathrm{End}(1_{{1_{\dots}}_{1_{\mathcal T}}})]$, i.e.\ about the composition of topological defects of codimension $p+1$ up to equivalence (and determines it completely if the combined action on all charged objects is faithful). 
E.g.\ for $d=4$ and $p=2$, $[\mathrm{End}(1_{1_{\mathcal T}})]$ consists of isomorphism classes of topological line defects with product given by fusion, and each non-topological simple surface operator $\mathcal{W}$ gives a 1-dimensional representation of $[\mathrm{End}(1_{1_{\mathcal T}})]$.

\subsection{The topological symmetry theory}\label{sec:symTFT}

Topological (finite) symmetries of a $d$-dimensional QFT $\mathcal T$ can often be described via a TFT $\zz$ in $n = d+1$ dimensions, the so-called \textit{topological symmetry theory} (or ``symTFT'').
The TFT $\zz$ is placed on an interval times a $d$-dimensional region, and the left and right boundary are labelled by a  topological boundary condition~$\mathcal B$ and a non-topological boundary condition~$\widehat{\mathcal T}$, respectively.
Since~$\zz$ is topological, the size of the interval does not matter, and we can collapse it without affecting correlators.
We assume that the result is equivalent as a $d$-dimensional QFT to~$\mathcal T$:\footnote{This is meant in the sense that with proper identifications of the operators, correlators on both sides are equal. 
It may be possible to give a formulation in terms of a $(d+1)$-category whose objects are $(d+1)$-dimensional TFTs and whose 1-morphisms are possibly non-topological codimension-1 defects, but where all higher codimension defects are again topological.}
\begin{equation}\label{eq:symTFT-setup}
\begin{tikzpicture}[scale=0.7, baseline=0.35cm]
\fill[orange!40!white, opacity=0.7] 
(0,0) to [out=90, in=-90]  (0,1.5)
	to [out=0,in=180] (6,1.5)
	to [out = -90, in =90] (6,0)
	to [out=180, in =0]  (0,0);
	\node[above] at (3,0.3) {{\footnotesize$\zz$}};
	\draw[very thick,color=blue!50!black] (0,0) -- (0,1.5);
	\draw[very thick, color=blue!90!black] (6,0) -- (6,1.5);
	\node[right] at (6,0.6) {{\footnotesize$\widehat{\mathcal T}$}}; 
	\node[left] at (0,0.6) {{\footnotesize$\mathcal B\vphantom{\widehat{\mathcal T}}$}};
\end{tikzpicture}
\quad \cong \quad
\begin{tikzpicture}[scale=0.7, baseline=0.35cm]
	\draw[very thick, color=blue!90!black] (3,0) -- (3,1.5);
	\node[right] at (3,0.6) {{\footnotesize$\mathcal T\vphantom{\widehat{\mathcal T}}$}}; 
\end{tikzpicture}
\end{equation}
Different choices of $\mathcal B$ describe different QFTs~$\mathcal T$ that can be obtained from $\widehat{\mathcal T}$, and topological defects of~$\mathcal T$ are described by topological defects of~$\zz$. 
Depending on $\mathcal B$, some topological defects of $\zz$ may become trivial when collapsing the interval.
We stress that a given QFT~$\mathcal{T}$ can allow for several different topological symmetry theories $\mathcal{Z}$, analogous to the fact that a given vector space can be part of representations of different groups.\footnote{
Theories like $\widehat{\mathcal T}$ (and also $\mathcal{B}$) which are defined on the boundary of a $(d+1)$-dimensional TFT are sometimes called \textsl{relative field theories} \cite{Freed:2012bs}. 
A hallmark of relative field theories is that to a closed compact $d$-dimensional manifold $X^d$ they assign a vector in the state space $\zz(X^d)$ which depends on the cobounding $(d+1)$-manifold, rather than a partition function.
In the present case, $\zz$ has topological boundary conditions, and combining the topological boundary condition $\mathcal{B}$ with the non-topological one given by $\widehat{\mathcal{T}}$ as in \eqref{eq:symTFT-setup} produces a $d$-dimensional theory $\mathcal T$, which is independent of 
a higher-dimensional bulk theory
(possibly up to an invertible anomaly theory -- see Footnote \ref{footnote:inflow}). 
There are examples of TFTs that do not allow for topological boundary conditions, for example 3d $\textrm{U}(1)_k$ Chern--Simons theory (or more generally Reshetikhin--Turaev TFTs for modular categories that are not Drinfeld centres) \cite{Kapustin:2010if,Fuchs:2012dt}. For such theories, one can still combine two non-topological boundaries to obtain an absolute theory (as in Figure~\ref{fig:2dCFT-3dTFT}\,a), or work with the doubled theory (as in Figure~\ref{fig:2dCFT-3dTFT}\,c).
} 

Consider the case $d=2$ and let $\mathcal{E}_{\mathcal{T}}$ be the pivotal category of topological defects of the 2d\,QFT $\mathcal{T}$ as in Section~\ref{sec:2cat-form-def}. 
Suppose $\mathcal{T}$ has a spherical fusion category $\mathcal{F}$ as topological symmetry via a realisation $R\colon \mathcal{F} \to \mathcal{E}_{\mathcal{T}}$ as in \eqref{eq:piv-fusion-cat-sym}. Then the topological symmetry theory $\mathcal{Z}$ is given by  the Turaev--Viro--Barrett--Westbury state sum TFT for $\mathcal F$ \cite[App.\,B]{Lin:2022dhv}, and $\widehat{\mathcal T}$ is obtained by coupling $\mathcal{Z}$ to $\mathcal{T}$ via $R$. 
When taking $\mathcal{B}$ to be the ``end-of-foam'' boundary condition for $\mathcal{Z}$ (i.e.\ the locus where the foam-like defect network in the state sum construction ends), in \eqref{eq:symTFT-setup} one recovers the theory $\mathcal{T}$ one started from. 
     More generally, for any $d$ a $(d+1)$-dimensional state sum model~$\mathcal Z_{\mathcal A}$ is constructed from a network of defects~$\mathcal A$ (``orbifold data''), generalising the setting in Sections~\ref{subsec:defect2categories}--\ref{subsec:defect3categories}, see \cite{Carqueville:2017aoe,Carqueville:2023qrk}; 
     such a symmetry TFT~$\mathcal Z_{\mathcal A}$ has a canonical end-of-foam boundary condition~$\mathcal B$ (often called ``Dirichlet''), which algebraically is~$\mathcal A$ as a right module over itself. 
For more details on the symmetry TFT see e.g.~\cite{Kapustin:2014gua,Gaiotto:2014kfa,Ji:2019,Gaiotto:2020iye,Kong:2020,Apruzzi:2021nmk,Freed:2022qnc,Kaidi:2022cpf,Bhardwaj:2023ayw,Bartsch:2023wvv}.

One can also give a symmetry TFT for lattice models \cite{Freed:2018cec, Aasen:2020jwb}. The symmetry TFT description of rational 2d\,CFT was discussed in Section~\ref{sec:2dCFT-3dTFT}.
When $\mathcal T$ has 't~Hooft anomalies in the presence of background fields, 
background gauge invariance can be restored by turning the left boundary $\mathcal B$ into an interface between the theory $\mathcal Z$ and the $(d+1)$-dimensional invertible anomaly theory arising from inflow (see Footnote~\ref{footnote:inflow}).

\medskip

\noindent\textbf{Application: Duality defects from the symmetry TFT perspective. } 
Placing a $d$-di\-men\-sional topological endo-defect $\mathcal{D}$ of $\zz$ on the interval between $\mathcal{B}$ and $\widehat{\mathcal T}$ and fusing it with $\mathcal{B}$ produces a new topological boundary condition~$\mathcal B'$ for~$\zz$, and so a new QFT ${\mathcal{T}'}$:
\begin{equation}
\begin{gathered}
\begin{tikzpicture}[scale=0.7, yscale=0.5]
\fill[orange!40!white, opacity=0.7] 
(0,0) to [out=90, in=-90]  (0,3)
	to [out=0,in=180] (6,3)
	to [out = -90, in =90] (6,0)
	to [out=180, in =0]  (0,0);
	\node[above] at (3,1) {{\footnotesize$\zz$}};
	\draw[very thick,color=green!60!black] (1.5,0) -- (1.5,3);
	\node at (1.2,1.2) {{\footnotesize$\mathcal D$}};
	\draw[very thick] (0,0) -- (0,3);
	\draw[very thick, color=blue!90!black] (6,0) -- (6,3);
        \node[right] at (6,1.2) {{\footnotesize$\widehat{\mathcal T}$}}; 
	\node[left] at (0,1.2) {{\footnotesize$\mathcal B\vphantom{\widehat{\mathcal T}}$}};
\end{tikzpicture}
\end{gathered} 
\quad
\cong 
\quad 
\begin{gathered}
\begin{tikzpicture}[scale=0.7, yscale=0.5]
\fill[orange!40!white, opacity=0.7] 
(0,0) to [out=90, in=-90]  (0,3)
	to [out=0,in=180] (6,3)
	to [out = -90, in =90] (6,0)
	to [out=180, in =0]  (0,0);
	
	\node[above] at (3,1) {{\footnotesize$\zz$}};
	
	\draw[very thick,color=green!40!black] (0,0) -- (0,3);
	\draw[very thick, color=blue!90!black] (6,0) -- (6,3);
        \node[right] at (6,1.2) {{\footnotesize$\widehat{\mathcal T}$}}; 
	\node[left] at (0,1.2) {{\footnotesize$\mathcal B'\vphantom{\widehat{\mathcal T}}$}};
\end{tikzpicture}
\end{gathered}
\quad\cong\quad
\begin{gathered}
\begin{tikzpicture}[scale=0.7, yscale=0.5]
	\draw[very thick, color=blue!60!green] (3,0) -- (3,3);
	\node[right] at (3,1.2) {{\footnotesize${\mathcal T'}$}}; 
\end{tikzpicture}
\end{gathered}    
\end{equation}
If $\mathcal{D}$ can end on a $(d-1)$-dimensional defect 
$\mathcal{E}$ of $\zz$ (a so-called \textsl{twist defect})
one obtains a topological interface $\mathcal{G}$
between $\mathcal{T}'$ and ${\mathcal{T}}$.
The non-topological theories often depend on parameters $\tau$, and the topological interface $\mathcal{G}$ then sits between $\mathcal{T}(\tau)$ and ${\mathcal{T}'}(\tau)$:
\begin{equation}\label{eq:interface-via-symTFT}
\begin{gathered}
\begin{tikzpicture}[scale=0.7, yscale=0.5]
\fill[orange!40!white, opacity=0.7] 
(0,0) to [out=90, in=-90]  (0,3)
	to [out=0,in=180] (6,3)
	to [out = -90, in =90] (6,0)
	to [out=180, in =0]  (0,0);
	
	\node[above] at (3,1.2) {{\footnotesize$\zz$}};
	
	\draw[very thick] (0,0) -- (0,3);
	\draw[very thick, color=blue!90!black] (6,0) -- (6,3);
        \node[right] at (6,1.2) {{\footnotesize$\widehat{\mathcal T}(\tau)$}}; 
	\node[left] at (0,1.2) {{\footnotesize$\mathcal B\vphantom{\widehat{\mathcal T}}$}};
	
	\draw[very thick,color=green!60!black] (1.5,0) -- (1.5,1.5);
	\node[above right] at (1.5,0) {{\footnotesize$\mathcal D$}};
        \node[left] at (1.5,1.55) {{\footnotesize$\mathcal E$}};
        \draw[thick,color=red,fill=red] (1.5,1.55) circle (0.08);
\end{tikzpicture}
\end{gathered} 
\quad \cong \quad
\begin{gathered}
\begin{tikzpicture}[scale=0.7, yscale=0.5]
	\node at (3.5,2.7) {{\footnotesize$\quad\mathcal T(\tau)$}};
	\draw[very thick, color=blue!90!black] (3,1.5) -- (3,3);
	\node at (3.5,1.55) {{\footnotesize$\textcolor{red}{\mathcal{G}}$}}; 
	\draw[very thick, color=blue!60!green] (3,0) -- (3,1.5);
	\node[right] at (3,0.5) {{\footnotesize${\mathcal T}'(\tau)$}}; 
 \draw[thick,color=red,fill=red] (3,1.55) circle (0.08);
\end{tikzpicture}
\end{gathered}    
\end{equation}
Suppose now that $\mathcal T$ and ${\mathcal{T}'}$ are related by a \textit{field theory duality}, that is, there is an equivalence of the form
\begin{equation}\label{eq:QFT-duality-iso}
\mathcal T(\tau) \cong {\mathcal T}'(f(\tau))
\ .
\end{equation}
If there is a fixed point $\tau_* = f(\tau_*)$, the interface $\mathcal{G}$ above can be combined with the isomorphism $\mathcal T(\tau_*) \cong {\mathcal T}'(\tau_*)$ to obtain a new topological endo-defect $\mathcal{N}$ of $\mathcal{T}(\tau_*)$.

Let us specialise the situation further by taking the $d$-dimensional defect $\mathcal{D}$ in $\zz$ to be the result of gauging lower-dimensional defects $\mathcal A$ on the support of $\mathcal{D}$, i.e.\ by gauging an orbifold datum as described for $d=2,3$ in Section~\ref{sec:TFT}. 
We further take the $(d-1)$-dimensional twist defect $\mathcal{E}$ to be given by ``ending the defect foam''. 
When composing $\mathcal{N}$ with its orientation reversal  $\mathcal{N}^{\textrm{op}}$ (note that $(-)^*$ and $(-)^\#$ from Section~\ref{sec:TFT} are special cases of $(-)^{\textrm{op}}$), the equivalence
$\mathcal T(\tau_*) \cong {\mathcal T}'(\tau_*)$ cancels and we are left with
\begin{equation}
\begin{gathered}
\begin{tikzpicture}[scale=0.7]

\fill[orange!40!white, opacity=0.7] 
(0,0) to [out=90, in=-90]  (0,3)
	to [out=0,in=180] (6,3)
	to [out = -90, in =90] (6,0)
	to [out=180, in =0]  (0,0);
	
	\node[above] at (3,1.7) {{\footnotesize$\zz$}};
	
	\draw[thick] (0,0) -- (0,3);
	\draw[very thick, color=blue!90!black] (6,0) -- (6,3);
        \node[right] at (6,1.2) {{\footnotesize$\widehat{\mathcal T}(\tau_*)$}}; 
	\node[left] at (0,1.2) {{\footnotesize$\mathcal B\vphantom{\widehat{\mathcal T}}$}};
	
        \node[left] at (1.5,2) {{\footnotesize$\mathcal E$}};
	\draw[very thick,color=green!60!black] (1.5,1) -- (1.5,2);
	\node[right] at (1.5,1.5) {{\footnotesize$\mathcal D$}};
	
	\draw[thick,color=red,fill=red] (1.5,1) circle (0.07);
        \draw[thick,color=red,fill=red] (1.5,2) circle (0.07);
        \node[left] at (1.5,1) {{\footnotesize$\mathcal E^{\textrm{op}}$}};
\end{tikzpicture}
\end{gathered} 
\quad \cong \quad
\begin{gathered}
\begin{tikzpicture}[scale=0.7]
	\node at (3.4,2.8) {{\footnotesize$\qquad\mathcal T(\tau_*)$}};
	\draw[very thick, color=blue!90!black] (3,1.6) -- (3,3);
	\node at (3.4,2) {{\footnotesize$\textcolor{red}{\mathcal{G}}$}}; 
	\draw[very thick,color=blue!60!green] (3,1) -- (3,1.93);
	\node at (3.6,1) {{\footnotesize$\textcolor{red}{\mathcal{G}^{\textrm{op}}}$}};
	\draw[very thick, color=blue!90!black] (3,0) -- (3,1);
	\node at (3.4,0.3) {{\footnotesize$\qquad\mathcal T(\tau_*)$}}; 
 \draw[thick,color=red,fill=red] (3,2) circle (0.07);
 \draw[thick,color=red,fill=red] (3,1.0) circle (0.07);
\end{tikzpicture}
\end{gathered}
\quad ,
\end{equation}
where on the right-hand side the symmetry is gauged in the region between $\mathcal{G}^{\textrm{op}}$ and $\mathcal{G}$. 
The fusion $\mathcal{G} \otimes \mathcal{G}^{\textrm{op}}$ in $\mathcal{T}$ can be computed in $\zz$ as the composition of 2-morphisms $1_\zz \xrightarrow{\mathcal{E}^{\textrm{op}}} \mathcal{D} \xrightarrow{\mathcal{E}} 1_\zz$, which just returns the $(d-1)$-dimensional component $A$ of the orbifold datum $\mathcal A$ used to define~$\mathcal{D}$, compare~\eqref{eq:2d-orb-datum} and~\eqref{eq:3d-orb-datum}. 
Typically~$A$ is non-invertible and in that case it follows that $\mathcal{N}$ is an example of a non-invertible topological symmetry. 

Suppose that $A$ itself is obtained from gauging a topological symmetry, and so is given by a foam of defects of dimensions $\leqslant d-2$ which is ``porous'' in the sense that it is transparent to point operators. In this case one cannot detect whether $\mathcal{G} \otimes \mathcal{G}^{\textrm{op}}$ (and hence $\mathcal{N}$) is invertible or not from the action on point operators and one needs to include higher-dimensional operators. This happens in the 4-dimensional examples below.

For more on the symmetry TFT approach to topological duality defects we refer to \cite{Freed:2018cec,Kaidi:2022cpf}.

\begin{example}
\noindent\textbf{(Kramers--Wannier duality of the 2d Ising model)}  
In the case of the lattice model in Section~\ref{sec:Ising-lattice}, the interface $\mathcal{G}$ above corresponds to $G$ in \eqref{eq:interface}, and the duality isomorphism \eqref{eq:QFT-duality-iso} corresponds to the duality interface $D$ in \eqref{eq:weights-duality-defect}. 
Passing to the continuum Ising CFT at the critical point $\beta = \beta^*$, the 3d symmetry TFT $\mathcal{Z}$
is the Reshetikhin--Turaev TFT for $\mathcal{C} \boxtimes \overline{\mathcal{C}}$, where $\mathcal{C}$ denotes the Ising modular category. 
In this example, $\mathcal{Z}$ has up to equivalence a unique topological boundary condition $\mathcal{B}$ 
(see Example~\ref{ex:surface-defect-interfaces}\,(2) and take $\mathcal{M}=\mathcal{I}$ in Figure~\ref{fig:2dCFT-3dTFT}).
The defect $\mathcal{D}$ is obtained by gauging the orbifold datum $\mathcal{A} = (A,\mu,\Delta)$ described 
in Example~\ref{ex:surface-defect-interfaces}\,(2).\footnote{There a surface defect in $\mathcal{C}$ is described, but we can interpret it as living in $\mathcal{C} \boxtimes \overline{\mathcal{C}}$ by taking it to be trivial in $\overline{\mathcal{C}}$.} 
The 1d defect $\mathcal{E}$ is given by $A$ with regular $\mathcal{A}$-action. The composition of $\mathcal{E}$ and $\mathcal{E}^{\textrm{op}}$ is $A = \one \oplus \varepsilon$, which is non-invertible (and does act non-trivially on point operators). 
For several other perspectives on the defects of the Ising model, see \cite[Sect.\,3]{Shao:2023gho}.
\end{example}

\begin{example}\label{ex:MaxD}\noindent\textbf{(Duality defects of Maxwell theory)} 
To derive electric-magnetic duality of quantum Maxwell theory in the path integral formalism, one treats $f^{(2)}$ as the fundamental field rather than $a^{(1)}$. 
To impose the Bianchi identity $d f^{(2)}=0$ one introduces a Lagrange multiplier 1-form field $a^{(1)}_D$ and adds the coupling
$\frac{\textrm{i}}{2\pi}\int d a^{(1)}_D \wedge f^{(2)}$
to the Maxwell action in \eqref{eq:Maxwell}. Integrating out $f^{(2)}$ gives
\begin{equation}\label{eq:Sdul}
f^{(2)}_D = \frac{2\pi \textrm{i}}{e^{2}} \,{\ast}f^{(2)} 
\ , \qquad  f^{(2)}_D = d a^{(1)}_D \ .
\end{equation}
Using this relation, the Maxwell action can be rewritten purely in terms of dual variables, with a dual coupling $e_D$:
\begin{equation}\label{dualityact}
  \mathcal S_{\textrm{Maxwell}} =\frac{1}{2 e_D^{2}}\int f^{(2)}_D \wedge \ast f^{(2)}_D \,\qquad \frac{e^{2}}{2\pi}=\frac{2\pi}{e_D^{2}}~.
\end{equation}
This implies the famous electric-magnetic duality, which states that there is an equivalence between Maxwell theory with coupling~$e$ and with coupling~$e_D$,
\begin{equation}\label{eq:EM-duality}
\mathcal T_{\textrm{Maxwell}}(e) \cong \mathcal T_{\textrm{Maxwell}}(2\pi/e) \ .
\end{equation}

Now consider gauging a $\mathds Z^{(1)}_N$ 1-form symmetry subgroup of the electric 1-form symmetry $\textrm{U}(1)_{\textrm{e}}^{(1)}$ of Maxwell theory, for some positive integer~$N$. 
This has the effect of shifting the fields as
\begin{equation}
a^{(1)}
\leadsto
\frac{1}{N} a^{(1)} \,\qquad a_D^{(1)}
\leadsto
N a_D^{(1)} \ ,
\end{equation}
and since the Maxwell action is quadratic in $a^{(1)}$, this is equivalent to shifting the coupling as $e \leadsto N e$, see e.g.~\cite[Sect.\,4.1]{Gaiotto:2014kfa}.
Altogether, one obtains the equivalences
\begin{equation}\label{eq:EM-duality-with-gaugeing}
\mathcal T_{\textrm{Maxwell}}(e)/\mathds Z^{(1)}_N \cong   \mathcal T_{\textrm{Maxwell}}(Ne) 
 \overset{\eqref{eq:EM-duality}}\cong
 \mathcal T_{\textrm{Maxwell}}(2\pi/(Ne))
 \ .
\end{equation}

In terms of the notation in \eqref{eq:interface-via-symTFT} we have $\mathcal{T} = \mathcal T_{\textrm{Maxwell}}(e)$ and ${\mathcal{T}'} = \mathcal T_{\textrm{Maxwell}}(e)/\mathds Z^{(1)}_N$, and the above equivalence amounts to \eqref{eq:QFT-duality-iso}. The duality isomorphism \eqref{eq:EM-duality-with-gaugeing} has the fixed point $e_* = \sqrt{2\pi / N}$, and by the reasoning in the beginning of this section we obtain a topological endo-defect $\mathcal{N}$ of $\mathcal T_{\textrm{Maxwell}}(e_*)$. 
The composition $\mathcal{N} \otimes \mathcal{N}^{\textrm{op}}$ is equivalent to the defect obtained by gauging the $\mathds Z^{(1)}_N$ 1-form symmetry on its 3d support and so is transparent to point operators, but considering the action on line operators shows its non-invertibility. 
In summary, $\mathcal N$ is a simple non-invertible topological defect of Maxwell theory at coupling $e_*$.
\end{example}

\begin{example}\label{ex:SYM}\noindent\textbf{(Duality defects from S-duality)} 
Self-dual points of more complicated 
4d models are a further source of examples of higher-dimensional theories with non-invertible duality defects. 
The simplest explicit example is given by taking $\mathcal{T}$ to be the 
4d maximally supersymmetric Yang--Mills theory with gauge group $\mathrm{SU}(N)$, and $\mathcal{T}'$ the corresponding theory for $\mathrm{SU}(N) / \mathds{Z}_N$ obtained by gauging the 1-form centre symmetry. 
The equivalence \eqref{eq:QFT-duality-iso} is provided by S-duality which gives an equivalence between $\mathcal{T}$ at coupling $\tau$ and $\mathcal{T}'$ at coupling $\tau' = -1/\tau$.
At the self-dual coupling $\tau_* = \mathrm{i}$ we obtain a non-invertible codimension-1 topological endo-defect $\mathcal{N}$ of $\mathcal{T}$.
\end{example}

More details about the examples above can be found in
\cite{Choi:2021kmx,Kaidi:2021xfk,Choi:2022zal}. 
Extending the above construction to more general (even non-abelian) duality groups, many further
examples of 4d theories with and without conventional Lagrangian descriptions that support non-invertible duality defects can be found in
\cite{Bashmakov:2022uek}. Of course one can also describe the symmetry TFT in all these examples. In particular for Examples \ref{ex:MaxD} and \ref{ex:SYM} above, $\mathcal Z$ is given by a 5-dimensional $\mathds{Z}_N$-gauge theory, or in other words an untwisted Dijkgraaf--Witten theory.
For more details see e.g.\ \cite{Kapustin:2014gua,Freed:2022iao}.

\subsection{Topological defects in quantum electrodynamics}

Quantum electrodynamics in $3+1$ dimensions is the theory of a Dirac fermion coupled to a Maxwell field. In this section we review its non-invertible family of codimension-one topological defects labeled by elements in $\mathds{Q}/\mathds{Z}$, building on \cite{Choi:2022jqy,Cordova:2022ieu}. The action of a free Dirac fermion of mass~$m$ is
\begin{equation}
\begin{aligned}
\mathcal S_{\textrm{Dirac}} =& \int  \,d^4x \, \Big(\chi_{\textrm{L}}(x)^\dagger (\mathrm{i} \overline{\sigma}^\mu \partial_\mu) \chi_{\textrm{L}}(x) + \chi_{\textrm{R}}(x)^\dagger (\mathrm{i} \sigma^\mu \partial_\mu) \chi_{\textrm{R}}(x)\Big)
\\ &  - m \int d^4x \, \Big(\chi_{\textrm{R}}(x)^\dagger  \chi_{\textrm{L}}(x) + \chi_{\textrm{L}}(x)^\dagger \chi_{\textrm{R}}(x)\Big)
~,
\end{aligned}
\end{equation}
where we are decomposing the Dirac fermion as a left-handed and a right-handed Weyl spinor field $\chi_{\textrm{L}}(x)$ and $\chi_{\textrm{R}}(x)$, 
and we used the standard shorthands $\sigma^\mu = (1,\sigma^x,\sigma^y,\sigma^z)$ and $\overline{\sigma}^\mu = (1,-\sigma^x,-\sigma^y,-\sigma^z)$ in terms of Pauli matrices.
For $m=0$ this action is invariant under rotating $\chi_{\textrm{L}}(x)$ and $\chi_{\textrm{R}}(x)$ independently by a phase, and Noether's theorem gives two conserved 1-form currents
\begin{equation}
j^{(1)}_{\textrm{L}} =
\big(\chi_{\textrm{L}}(x)^\dagger \overline{\sigma}_\mu \chi_{\textrm{L}}(x)\big) dx^\mu \ , \qquad 
j^{(1)}_{\textrm{R}} =
\big(\chi_{\textrm{R}}(x)^\dagger \sigma_\mu \chi_{\textrm{R}}(x)\big)dx^\mu  \ .
\end{equation}
From these one obtains 3d invertible topological defects as in \eqref{eq:pform}. 
When $m\neq 0$ these currents are explicitly broken: the massive Dirac fermion is invariant only under simultaneous rotation of both of its chiral components. 
Thus, only the diagonal symmetry $j^{(1)}_{\textrm{V}} = j^{(1)}_{\textrm{L}} + j^{(1)}_{\textrm{R}}$ is conserved. 
The action of QED is given by gauging this symmetry by coupling it to a Maxwell field, 
\begin{equation}
\mathcal S_{\textrm{QED}} = \mathcal S_{\textrm{Maxwell}} + \mathcal S_{\textrm{Dirac}} + \textrm{i} q_{\textrm{e}} \int_X a^{(1)} \wedge \ast j^{(1)}_{\textrm{V}}\,,
\end{equation}
where $q_{\textrm{e}}$ is the electric charge of the electron. 
That $j^{(1)}_{\textrm{V}}$ is conserved is crucial for the gauge invariance of the action. 
This coupling has the effect of screening the $\textrm{U}(1)^{(1)}_{\textrm{e}}$ electric 1-form symmetry down to $\mathds Z^{(1)}_{q_{\textrm{e}}}$ -- in the language of Remark~\ref{rem:screening}, a Wilson line of charge~$q_\textrm{e}$ 
can end on this electron. 
The magnetic 1-form symmetry is not affected by this coupling, ${\textrm{U}}(1)^{(1)}_{\textrm{m}}$ is still a symmetry of QED. 
In what follows we will work with electrons of unit electric charge $q_{\textrm{e}}=1$,
so that the ${\textrm{U}}(1)^{(1)}_{\textrm{e}}$ is completely screened.

Classically, setting $m=0$, it is natural to expect that massless QED acquires a $\textrm{U}(1)^{(0)}_{\textrm{A}}$ symmetry, corresponding to the additional conserved current $j^{(1)}_{\textrm{A}} = j^{(1)}_{\textrm{L}} - j^{(1)}_{\textrm{R}}$. 
However, in correlators this current is not conserved: quantum effects imply that it satisfies the Adler--Bell--Jackiw anomaly equation
\begin{equation}
d \,{\ast} j^{(1)}_{\textrm{A}} 
= \frac{1}{8 \pi^2} f^{(2)} \wedge f^{(2)}\,.
\end{equation}
Inserting $j^{(1)}_{\textrm{A}}$ into the construction \eqref{eq:pform} would result in a non-topological 3-dimensional operator of QED. However, for every value of $\theta \neq 0$ the expression $\textrm{exp}(\textrm{i}\theta\int_{\Sigma^3} \ast j^{(1)}_{\textrm{A}})$ does give rise to a topological interface between two distinct 4d theories: 
Writing 
\begin{equation}
\mathcal S_{\textrm{QED}_\theta} 
:= 
\mathcal S_{\textrm{QED}}
+ \textrm{i}\frac{\theta}{8\pi^2} \int f^{(2)}\wedge f^{(2)}
\end{equation}
for the theory with a non-trivial $\theta$-term, we obtain a topological
codimension-1 defect as follows: 
\begin{equation}\label{eq:interth}
\begin{gathered}
\begin{tikzpicture}[scale=0.7]
\fill[blue!40!white, opacity=0.7] 
(-5,1) to [out=90, in=-90]  (-5,3)
	to [out=0,in=180] (0,3)
	to [out = -90, in =90] (0,1)
	to [out=180, in =0]  (-5,1);
 \fill[blue!20!white, opacity=0.7] 
(0,1) to [out=90, in=-90]  (0,3)
	to [out=0,in=180] (11,3)
	to [out = -90, in =90] (11,1)
	to [out=180, in =0]  (0,1);
	\node[above] at (-3,1.75) {{\footnotesize $\textrm{QED}_\theta$}};
	\node[above] at (8.5,1.75) {{\footnotesize $\textrm{QED}_0 = \textrm{QED}$}};
	\draw[ultra thick, color=green!60!black] (0,1) -- (0,3);
    \node at (2.1,1.5) {{\footnotesize$\text{exp}\left( \textrm{i}\theta \int_{\Sigma^3} \ast j^{(1)}_{\textrm{A}} \right)\,$}};
\end{tikzpicture}
\end{gathered} 
\end{equation}

It turns out that for $\frac{\theta}{2\pi} \in \mathds{Q}/\mathds{Z}$ one can construct another topological interface between $\textrm{QED}$ and $\textrm{QED}_\theta$, which, however, is non-invertible. 
Concretely, let us consider $\theta = 2\pi p/N$ where $p$ and $N$ are coprime integers. 
Due to the quantisation condition $\frac{1}{2\pi}\oint f^{(2)} \in \mathds{Z}$, with such a choice of $\theta$ the only difference between the path integrals on the two sides of the interface in~\eqref{eq:interth} arises from the contributions of those fluxes whose square is non-zero modulo $N$. 
These identify a $\mathds{Z}_N$ subgroup of the $\textrm{U}(1)^{(1)}_{\textrm{m}}$ magnetic 1-form symmetry of QED which suggests to look for a 3d (spin)
TFT $\zz^{N,p}_3$ with an anomalous 1-form symmetry $\mathds{Z}_{N}^{(1)}$ such that, in the language introduced in Footnote~\ref{footnote:inflow}, its 4d anomaly inflow theory $\zz^{N,p}_4$ cancels the mismatching contributions. 
One option is given by the 4d anomaly theory
\begin{equation}\label{eq:SPT}
\zz^{N,p}_4[B^{(2)}] = \text{exp}\left(-\frac{2\pi \textrm{i} p}{2N}\int \mathcal{P}(B^{(2)})\right)~,
\end{equation}
where $B^{(2)} \in H^2(X,\mathds{Z}_N)$
denotes the background field for $\mathds{Z}_{N}^{(1)}$,
$\mathcal{P}$ is the Pontryagin square operation,
and $\int \mathcal{P}(B^{(2)})$ refers to evaluating against a 4-chain in $X$.
Indeed, denoting by $[f^{(2)}/(2\pi)] \in H^2(X,\mathds{Z})$ the cohomology class of $f^{(2)}/(2\pi)$, we couple this 3d TFT to the Maxwell field of the 4d theory via the identification
\begin{equation}\label{eq:coupling}
B^{(2)} \equiv \, [f^{(2)}/(2\pi) ] \mod N~.
\end{equation}
In this way the anomaly inflow action precisely cancels all the mismatching contributions to the path integral and we obtain a non-invertible interface from of QED$_{\frac{2\pi p}{N}}$ to QED$_0$: 
\begin{equation}
\begin{gathered}
\begin{tikzpicture}[scale=0.7]
\fill[blue!40!white, opacity=0.7] 
(-5,1) to [out=90, in=-90]  (-5,3)
	to [out=0,in=180] (3,3)
	to [out = -90, in =90] (3,1)
	to [out=180, in =0]  (-5,1);
\fill[blue!20!white, opacity=0.7] 
(-9,1) to [out=90, in=-90]  (-9,3)
	to [out=0,in=180] (-5,3)
	to [out = -90, in =90] (-5,1)
	to [out=180, in =0]  (-9,1);
\fill[blue!20!white, opacity=0.7] 
(3,1) to [out=90, in=-90]  (3,3)
	to [out=0,in=180] (8,3)
	to [out = -90, in =90] (8,1)
	to [out=180, in =0]  (3,1);
    \node[above] at (-4.3,1.2) {\textcolor{black}{{\footnotesize$\zz_3^{N,p}$}}};
    \node[above] at (-6.7,1.75) {{\footnotesize $\textrm{QED}_0$}};
    \node[above] at (-1,1.75) {{\footnotesize $\textrm{QED}_\frac{2\pi p}{N}$}};
	\node[above] at (5.7,1.75) {{\footnotesize $\textrm{QED}_0$}};
    \node at (5,1.5) {{\tiny$\text{exp}\left( \textrm{i} \frac{2\pi p}{N} \int_{\Sigma^3} \ast j^{(1)}_{\textrm{A}} \right)\,$}};
	\draw[ultra thick, color=green!60!black] (3,1) -- (3,3);
	\draw[ultra thick, color=orange!90!black] (-5,1) -- (-5,3);
\end{tikzpicture}
\end{gathered} 
\end{equation}
By sending the distance between the above two defects to zero, we obtain a codimension-1 topological defect
\begin{equation}\label{eq:chiralsy}
\topD{\frac{p}{N}}{(\Sigma^3)} = \zz^{N,p}_3(\Sigma^3) \otimes \text{exp}\left(  \textrm{i}\frac{2\pi p}{N} \int_{\Sigma^3} \ast j^{(1)}_{\textrm{A}} \right)\,.
\end{equation}

The only missing ingredient in the construction is a 3d TFT with a $\mathds{Z}^{(1)}_N$ form symmetry and anomaly \eqref{eq:SPT}. 
As discussed in detail in \cite{Hsin:2018vcg}, any such 3d TFT can be written as a product $\zz^{N,p}_3 \cong \mathcal A^{N,p}_3 \otimes \widetilde{\zz}_3^{N,p}$ where $\widetilde{\zz}^{N,p}_3$ is a factor which is completely neutral with respect to the $\mathds{Z}^{(1)}_N$ charge. 
Then one can choose to use this minimal abelian TFT $\mathcal{A}_3^{N,p}$ as a candidate 3d TFT to define a chiral symmetry topological defect for QED$_0$. 
The $(\mathds Q/\mathds Z)^{(0)}_A$ family of such defects gives rise to the chiral symmetry of massless QED. 
Indeed, in massless QED the chirality of the fermions is conserved in scattering processes, which gives a selection rule on correlation functions. This selection rule is a consequence of the non-invertible chiral symmetry \cite{Choi:2022jqy}. 

\medskip
\setlength{\leftmargini}{1.4em}

\noindent
We conclude by highlighting some features of the non-invertible symmetry:
\begin{itemize}
\item[$\triangleright$]
\textsl{(Charged lines)} 
Notice that by the relation~\eqref{eq:coupling} the 1-form symmetry on the defect is identified with a $\mathds Z_N^{(1)}$ subgroup of the magnetic 1-form symmetry of QED. 
The simple line operators $\mathcal{L}^s$ of $\mathcal A^{N, p}_3$ have $\mathds{Z}_N$ fusion rules, and 
by the above identifications, they acquire the charge
\begin{equation}\label{eq:linecharged}
q_{\textrm{m}}(\mathcal L^s) 
= p s \in \mathds Z_N \subset \textrm{U}(1)^{(1)}_{\textrm{m}}\,.
\end{equation}
under the bulk $\textrm{U}(1)^{(1)}_{\textrm{m}}$.
For this reason suitable 1-form magnetic symmetry defects can end on the defect $\mathcal D_{\frac{p}{N}}$ ending on such magnetically charged lines.

\item[$\triangleright$]
\textsl{(Value on the sphere)} The value of $\topD{\frac{p}{N}}$ on a 3-sphere bounding an empty 3-ball is $\topD{\frac{p}{N}}{(S^3)} = \mathcal A_3^{N,p}(S^3)=1 / \sqrt{N}$ \cite{Hsin:2018vcg}, 
which gives an analogue of a quantum dimension for this operator.

\item[$\triangleright$]
\textsl{(Action on Hilbert space)} The non-invertibility of $\topD{\frac{1}{N}}$ can be made explicit by considering its action on the 
Hilbert space for $S^1 \times S^2$ of QED. 
It splits in superselection sectors graded by magnetic flux $q_\mathrm{m} \in \mathds{Z}$, and each operator $\topD{\frac{1}{N}}{(S^1 \times S^2)}$ acts as a projector onto sectors where the magnetic flux is $0$ mod $N$, see \cite[Sect.\,A]{Cordova:2022ieu}. 

\item[$\triangleright$] 
\textsl{(Action on 't~Hooft lines of QED)} 
The topological defects $\mathcal D_{\frac{p}{N}}$ act on the 't~Hooft lines of QED similar to the order/disorder transition in the Ising model. 
By the Witten effect, when crossing the topological defect the 't~Hooft line becomes an ill-quantised Wilson line with a magnetic 1-form surface defect connecting it to the worldvolume of $\mathcal D_{\frac{p}{N}}$.

\item[$\triangleright$] 
\textsl{(Higher structure)} The higher category of topological defects of QED, capturing properties of the $(\mathds Q/\mathds Z)^{(0)}_A$ topological defects $\topD{\frac{p}{N}}$, as well as their interplay with the $\textrm{U}(1)^{(1)}_{\textrm{m}}$ topological defects (along with the higher gaugings \cite{Roumpedakis:2022aik} of its finite subgroups), is investigated in detail in
\cite{Copetti:2023mcq}. 
The two main outcomes are that the higher structure
\begin{enumerate}
\item \textit{is physical}: there are observable consequences on correlators of QED of the non-triviality of associators like the ones defined around \eqref{eq:3dCalcExample} for chiral symmetry topological defects;
\item \textit{gives relations among correlators on more general topologies}: 
expressing a closed compact 4-manifold via a handle decomposition allows one to compute the defect network resulting from a bubble expansion as in Figure~\ref{fig:WTid}.
\end{enumerate}
\end{itemize}
For further details about the material of this section, 
see~\cite{Putrov:2022pua,Copetti:2023mcq}. 

\newcommand\arxiv[2]      {\href{https://arXiv.org/abs/#1}{#2}}
\newcommand\doi[2]        {\href{https://dx.doi.org/#1}{#2}}

\end{document}